\documentclass[onecolumn,A4,10pt,draftcls]{IEEEtran}
\ifCLASSINFOpdf
\else
\fi
\usepackage{multirow}
\usepackage{graphicx}
\usepackage{color}
\usepackage{placeins}
\usepackage{float}
\usepackage{tabularx,colortbl}
\usepackage{makeidx}
\usepackage[toc]{glossaries}
\usepackage{tikz}
\usetikzlibrary{shapes,fit,backgrounds,positioning,calc,decorations.pathreplacing,patterns}
\usepackage{amsthm}
\usepackage{amssymb}
\usepackage{amsfonts}
\usepackage{subfig}
\usepackage{mathtools}
\usepackage{xfrac}
\usepackage{todonotes}
\usepackage{dsfont}
\theoremstyle{plain}
\newtheorem{theo}{Theorem}
\newtheorem{lem}{Lemma}
\newtheorem{cor}{Corollary}

\theoremstyle{definition}
\newtheorem{defi}{Definition}

\theoremstyle{remark}

\newtheorem{remark}{Remark}
\newtheorem{ex}{Example}

\newacronym{CSI}{CSI}{\textbf{C}hannel \textbf{S}tate \textbf{I}n\-for\-ma\-tion}
\newacronym{CDF}{CDF}{\textbf{C}umulative \textbf{D}istribution \textbf{F}unction}
\newacronym{DMS}{DMS}{\textbf{D}iscrete \textbf{M}emoryless \textbf{S}ource}
\newacronym{DMC}{DMC}{\textbf{D}iscrete \textbf{M}emoryless \textbf{C}hannel}
\newacronym{DOF}{DoF}{\textbf{D}egrees \textbf{o}f \textbf{F}reedom}
\newacronym{IA}{IA}{\textbf{I}nterference \textbf{A}lignment}
\newacronym{IC}{IC}{\textbf{I}nterference \textbf{C}hannel}
\newacronym{MAC}{MAC}{\textbf{M}ultiple \textbf{A}ccess \textbf{C}hannel}
\newacronym{MIMO}{MIMO}{\textbf{M}ultiple \textbf{I}nput \textbf{M}ultiple \textbf{O}utput}
\newacronym{RV}{RV}{\textbf{R}andom \textbf{V}ariable}
\newacronym{SNR}{SNR}{\textbf{S}ignal to \textbf{N}oise \textbf{R}atio}
\newacronym{AVWC}{AVWC}{\textbf{A}rbitrarily \textbf{V}arying \textbf{W}iretap \textbf{C}hannel}
\newacronym{AVWC-CJ}{AVWC-CJ}{\textbf{A}rbitrarily \textbf{V}arying \textbf{W}iretap \textbf{C}hannel with \textbf{C}hannel input non-causally known at the \textbf{J}ammer}
\newacronym{CAVWC}{CAVWC}{\textbf{C}ompound-\textbf{A}rbitrarily \textbf{V}arying \textbf{W}iretap \textbf{C}hannel}
\newacronym{AWGN}{AWGN}{\textbf{A}dditive \textbf{W}hite \textbf{G}aussian \textbf{N}oise}
\newacronym{AVC}{AVC}{\textbf{A}rbitrarily \textbf{V}arying \textbf{C}hannel}
\newacronym{AVMAC}{AVMAC}{\textbf{A}rbitrarily \textbf{V}arying \textbf{M}ultiple \textbf{A}ccess \textbf{C}hannel}
\newacronym{AVBC}{AVBC}{\textbf{A}rbitrarily \textbf{V}arying \textbf{B}roadcast \textbf{C}hannel}
\newacronym{AVRC}{AVRC}{\textbf{A}rbitrarily \textbf{V}arying \textbf{R}elay \textbf{C}hannel}
\newacronym{AVQC}{AVQC}{\textbf{A}rbitrarily \textbf{V}arying \textbf{Q}uantum \textbf{C}hannel}
\newacronym{AVCQC}{AVCQC}{\textbf{A}rbitrarily \textbf{V}arying \textbf{C}lassical-\textbf{Q}uantum \textbf{C}hannel}
\newacronym{AVC-CJ}{AVC-CJ}{\textbf{A}rbitrarily \textbf{V}arying \textbf{C}hannel with \textbf{C}hannel input non-causally known at the \textbf{J}ammer}
\newacronym{RT}{RT}{\textbf{R}obustification \textbf{T}echnique}
\newacronym{ET}{ET}{\textbf{E}limination \textbf{T}echnique}
\newacronym{WTC}{WTC}{\textbf{W}ire\textbf{t}ap \textbf{C}hannel}
\newacronym{CWC}{CWC}{\textbf{C}ompound \textbf{W}iretap \textbf{C}hannel}
\newacronym{SREX}{SREX}{\textbf{S}ecure \textbf{R}emote \textbf{Ex}ecution}
\newacronym{BSC}{BSC}{\textbf{B}inary \textbf{S}ymmetric \textbf{C}hannel}
\newacronym{SA}{SA}{\textbf{S}uper-\textbf{A}ctivation}
\newacronym{CC}{CC}{\textbf{C}ompound \textbf{C}hannel}
\newacronym{IID}{i.i.d.}{\textbf{i}ndependent and \textbf{i}dentically \textbf{d}istributed}
\newacronym{KKT}{KKT}{\textbf{K}arush \textbf{K}uhn \textbf{T}ucker}
\newacronym{JE}{JE}{\textbf{J}oint \textbf{E}ncoding}
\newacronym{BBT}{BBT}{\textbf{B}reiman, \textbf{B}lackwell and \textbf{T}homasian}
\newacronym{CR}{CR}{\textbf{C}ommon \textbf{R}andomness}
\newcommand{\CSI}{\gls{CSI}}

\newcommand{\DMC}{\gls{DMC}}

\newcommand{\RV}{\gls{RV}}

\newcommand{\AVWC}{\gls{AVWC}}

\newcommand{\AVC}{\gls{AVC}}
\newcommand{\AVMAC}{\gls{AVMAC}}
\newcommand{\AVBC}{\gls{AVBC}}
\newcommand{\AVRC}{\gls{AVRC}}

\newcommand{\AVCQC}{\gls{AVCQC}}
\newcommand{\RT}{\gls{RT}}
\newcommand{\ET}{\gls{ET}}
\newcommand{\WTC}{\gls{WTC}}
\newcommand{\CWC}{\gls{CWC}}

\newcommand{\SA}{\gls{SA}}
\newcommand{\CC}{\gls{CC}}
\newcommand{\IID}{\gls{IID}}

\newcommand{\CR}{\gls{CR}}

\newcommand{\DMCs}{\glspl{DMC}}

\newcommand{\RVs}{\glspl{RV}}

\newcommand{\AVWCs}{\glspl{AVWC}}

\newcommand{\AVCs}{\glspl{AVC}}

\newcommand{\AVBCs}{\glspl{AVBC}}

\newcommand{\AVQCs}{\glspl{AVQC}}

\newcommand{\CWCs}{\glspl{CWC}}

\makeindex
\makeglossaries

\newcommand{\overbar}[1]{\mkern 1.5mu\overline{\mkern-1.5mu#1\mkern-1.5mu}\mkern 1.5mu}

\allowdisplaybreaks[4]

\IEEEoverridecommandlockouts

\usepackage[outdir=./]{epstopdf}
\usepackage{pgfplots}
\usepackage{bbm}

\usepackage{longtable}
\usepackage{hhline}
\usepackage[T1]{fontenc}            
\usepackage[utf8]{inputenc}
\begin{document}
%

\title{Arbitrarily Varying Wiretap Channels with Non-Causal Side Information at the Jammer}


\author{\IEEEauthorblockN{Carsten Rudolf Janda\IEEEauthorrefmark{1}, Moritz Wiese\IEEEauthorrefmark{2}, Eduard A. Jorswieck\IEEEauthorrefmark{1} and Holger Boche\IEEEauthorrefmark{2}}\\
\IEEEauthorblockA{\IEEEauthorrefmark{1}Institute of Communications Technology,\\Department of Information Theory and Communication Systems, TU Braunschweig, Lower Saxony, Germany\\
\{janda, jorswieck\}@ifn.ing.tu-bs.de}\\
\IEEEauthorblockA{\IEEEauthorrefmark{2}Chair of Theoretical Information Technology,
Munich University of Technology
M\"{u}nchen, Bavaria, Germany\\
\{wiese,boche\}@tum.de}
\thanks{This work is partly funded by the german research foundation (DFG) within the project Play Scate (DFG  JO 801/21-1).}
}

%


\maketitle

\begin{abstract}
Secure communication in a potentially malicious environment becomes more and more important. 
The \AVWC{} provides information theoretical bounds on how much information can be exchanged even in the presence of an active attacker. 
If the active attacker has non-causal side information, situations in which a legitimate communication system has been hacked, can be modeled.
 
We investigate the \AVWC{} with non-causal side information at the jammer for the case that there exists a best channel to the eavesdropper.  
Non-causal side information means that the transmitted codeword is known to an active adversary before it is transmitted. 
By considering the maximum error criterion, we allow also messages to be known at the jammer before the corresponding codeword is transmitted. 
A single letter formula for the \CR{} secrecy capacity is derived.
Additionally, we provide a single letter formula for the \CR{} secrecy capacity, for the cases that the channel to the eavesdropper is strongly degraded, strongly noisier, or strongly less capable with respect to the main channel. 
Furthermore, we compare our results to the random code secrecy capacity for the cases of maximum error criterion but without non-causal side information at the jammer, maximum error criterion with non-causal side information of the messages at the jammer, and the case of average error criterion without non-causal side information at the jammer.
\end{abstract}

%
\IEEEpeerreviewmaketitle
\section{Introduction}
Secrecy in an adversarial environment is an essential requirement in modern communication systems.  
It was Wyner, \cite{Wyner1975}, who  considered secure communications over noisy channels and introduced the \WTC . 
Later, his work was extended by \cite{Csiszar1978} to the broadcast channel with confidential messages, and in \cite{Cheong1978} to the Gaussian \WTC .
In \cite{Ozarow1984}, Ozarow et. al introduced the wiretap channel of type II \footnote{Essentially, the eavesdropper is able to perfectly receive a fraction of the transmitted codeword. In contrast to a "random" erasure channel, here the eavesdropper can choose the exact symbols he wants to obtain.}. 
The secrecy metrics in the aforementioned works are considered "weak". 
There exist other secrecy metrics, such as strong secrecy, or semantic secrecy. 
In \cite{Goldfeld2016a}, the authors investigated wiretap channels of type I and type II. 
They provided achievable semantic secrecy rates for the \WTC{} of type I, and gave a single letter formula for the semantic secrecy capacity for the \WTC{} of type II. 
In \cite{Nafea2018}, the authors presented a generalized \WTC{} model. 
This model consists of a mixture of the \WTC{} of type I and II. 
During one fraction of the transmission of one codeword, the eavesdropping channel behaves like a \WTC{} of type I, in the remaining time instances it behaves like a \WTC{} of type II. 
For this model, \cite{Nafea2018} contributed a single letter secrecy capacity formula under the strong secrecy criterion.
The previous works combat a passive eavesdropper by cleverly taking the physical properties of the transmit medium into account and come up with a coding strategy which can guarantee information theoretic security, confidentiality, and reliable communication at the same time. 
%

\subsection{\acrfullpl{AVC}}
By introducing channel states, active adversaries who can arbitrarily modify the channel state can be modeled by the \AVC . 
For the \AVC , different code concepts are introduced in \cite{Blackwell1960}.  
In \cite{Ahlswede1970a}, the existence of "weak" capacities for \AVCs{} is investigated. 
A channel capacity is called a weak capacity, if the channel coding theorem contains a weak converse \footnote{All converses based on Fano's inequality are weak. The weak converse states that when using transmission rates above the channel capacity, the error probability is bounded away from 0. In contrast to that, the strong converse states that, when using transmission rates above the channel capacity, the error probability approaches 1 (exponentially fast).}.
The existence of a weak capacity for the deterministic code capacity under the maximum error for an \AVC{} is connected to Shannon's zero error capacity for an \DMC , \cite{Shannon1956}.  
In \cite{Ahlswede1973}, the \AVC{} with a noiseless feedback channel is considered.
Using a method from a coding theorem for the \DMC{} with feedback, \cite{Ahlswede1973a}, which is not based on random coding or maximal coding ideas, a coding theorem for the \AVC{} with feedback and a strong converse is presented. 
As an additional result, the zero error capacity formula for the \DMC{} with feedback is provided. 
In \cite{Ahlswede1978}, the discussion of \cite{Blackwell1960} is extended for different error criteria. 
It can be shown that the random code capacity of an \AVC{} under the average error criterion equals its random code capacity under the maximum error criterion. 
Even though \cite{Ahlswede1978} provides a necessary and sufficient condition for the deterministic code capacity under the maximum error criterion to be positive, the question about the exact formula remains an open problem.
In \cite{Csiszr1981},  the discussion on the maximum error criterion for \AVCs{} is extended. 
A deterministic code capacity formula for a class of \AVCs{} for which an additional type property holds is presented. 
In \cite{Jahn1981}, the deterministic code capacity region of the \AVMAC{} is derived under the condition that the interior of that region is non-empty.
Both the average and the maximum error criteria are considered. 
Furthermore, the  achievable rate regions for deterministic codes for the general and the degraded \AVBCs{} are provided under the conditions that these regions have non-empty interior. 
In \cite{Jahn1981}, the author pointed out, that the problem to determine whether those capacity regions possess empty interiors was still open at that point. It is solved later by \cite{Gubner1990}. 
Further, an achievable rate region for the general broadcast channel is provided.   
In \cite{Ericson1985}, random codes for the \AVC{} with limited amount of \CR{} are studied. 
The author limited the amount of \CR{} to increase only exponentially with respect to the block length. 
Furthermore, an exponential error bound is considered. 
Additionally, the author provided a sufficient condition for when the deterministic code capacity is zero. 
This condition is called symmetrizability. 
The author proved that if the symmetrizability condition is fulfilled, the (average) error probability is bounded away from zero, and is lower bounded by $\frac{1}{4}$. 
In \cite{Gubner1990}, the \AVMAC{} is investigated. 
Specifically, the author considered deterministic codes and extended the symmetrizability condition to the multi user scenario. 
Based on this multi user symmetrizability, a condition is derived, for when the deterministic code capacity region of the \AVMAC{} has a non-empty interior, such that both transmitters can communicate reliably. 
In doing so, the author solved one open problem of \cite{Jahn1981}, which had left the question whether those capacity regions possess empty interiors unanswered. 
In \cite{Ahlswede1991}, it is proved that for an \AVC{} every rate below the random code capacity is achievable with deterministic list codes of constant list size, if the average error criterion is used. 
The authors presented two different proofs, one based on the \ET , the other based on an adaptation of the \RT .
In \cite{Hughes1997}, the deterministic list code capacity of an \AVC{} is studied. 
The author presented a bound, called symmetrizability ({\cite[Definition 3]{Hughes1997}}),  on the smallest list size, for which the deterministic list code capacity equals the random code capacity. 
Below this bound the deterministic list code capacity equals zero.
In \cite{Gohari2008}, upper bounds on the admissible source region of the general \AVBC{} with arbitrarily correlated sources are investigated, using \CR{} assisted codes and the average error criterion. 
The capacity region of the general \AVBC{} relates to the admissible source region in the way that it is a set of rates for which an admissible source (messages) exists \footnote{The term admissible source region might be confusing at first, but it is nothing else, than computing the maximum rate at which the error probability vanishes. The connection gets clearer when remembering the connection $|\mathcal{J}|=\exp\{\lfloor nR \rfloor\}$. Hence, when $|\mathcal{J}|\leq\exp\{\lfloor nC \rfloor\}$, the source is called admissible.}.
When specializing to the case of independent sources and no channel variation it is shown that the presented outer bound is included in the outer bound of \cite{Liang2008}.
In \cite{Wyrembelski2009a}, random and deterministic coding strategies for a bidirectional \AVRC , consisting of an \AVMAC{} and bidirectional \AVBC{} phase, are investigated .
For the multiple access and the broadcast phases, the authors gave descriptions of the random and deterministic code capacity (if the interiors are non-empty). 
Their proof is based on the \RT{} and \ET{} by Ahlswede.    
In \cite{Wyrembelski2009}, the same set of authors extended their work in \cite{Wyrembelski2009a}, to derive necessary conditions for which the interior of the deterministic code capacity region of the bidirectional \AVBC{} is non-empty.   
In \cite{Nitinawarat2013}, the deterministic code capacity region of an \AVMAC{} under list decoding is considered  and the results of \cite{Jahn1981}, using a similar approach as in \cite{Gubner1990}, adapted to list decoding, are extended. 
The author was able to show that the capacity region using list codes with list sizes $L$ equals the random code capacity region if the interior of the capacity region using list codes is non-empty. 
He then proved  list size symmetrizability conditions for when the capacity region using list codes for the \AVMAC{} is empty, and for when the capacity region using list codes equals the random code capacity region.
In \cite{Schaefer2014}, the bidirectional \AVBC{} is investigated and the question how much randomness is sufficient and how much coordination between nodes is necessary to guaranty reliable communication is considered. 
Also weaker forms of \CR{} are considered, specifically correlated randomness, causal correlated randomness and no randomness at all. 
It is shown that the capacity regions of the investigated cases for the bidirectional \AVBC{} are subsets of each other, and derived symmetrizability conditions, for when the deterministic code capacity region equals the random code capacity region and for when the deterministic code capacity region has an empty interior, respectively. 
Furthermore, it is shown that as long as the correlated randomness at the relay and the other nodes is indeed correlated (and the nodes do not obtain independent observations), the correlated code capacity region equals the random code capacity region.
In \cite{Boche2014}, the continuity behavior of the randomness assisted and deterministic code capacities for \AVQCs{} is studied. 
While the randomness assisted code capacity is indeed continuous, the deterministic code capacity exhibits discontinuities. 
The authors considered furthermore the effect of limited \CR{} and finite block lengths with respect to the decoding error.
In \cite{Cai2016}, bipartite graphs are used  to prove necessary and sufficient conditions for the \AVMAC{} list code capacity to have a non-empty interior. 
Further, the auhtor extended the work of \cite{Nitinawarat2013}, and proved that the minimum list size is finite if and only if the correlated code capacity region has a non-empty interior.
%

\subsection{\acrfullpl{AVC} with Side Information}
In the literature, also different cases of side information at the transmitter and/or the jammer have been considered. 
In the following we want to provide a short overview. 
In \cite{Kiefer1962}, deterministic codes for the \AVC{} under different \CSI{} cases are investigated. 
Necessary and sufficient conditions are provided for positive deterministic code transmission rate for cases of no \CSI , \CSI R, \CSI T, and perfect \CSI .  
Additionally, the authors determined for the latter case the deterministic code capacity.
In \cite{Ahlswede1969}, different code classes and average and maximum error criteria for different \CSI T/\CSI R and side information at the jammer for the \AVC{} are considered. 
The authors showed the equivalence of certain cases, where the jammer randomizes (arbitrary or in an \IID{} manner) over the state space or uses a deterministic jamming strategy, and where the communication partners possess different \CSI . 
Furthermore, the random code capacity of the \AVC{} is provided and the authors showed that for different \CSI{} and error cases the capacity equals the random code capacity. 
In \cite{Ahlswede1970b}, the deterministic code capacity of the \AVC , where the channel output alphabet is binary, is determined. 
Additionally, the cases of \CSI T and \CSI R are investigated and the capacities for these cases are provided as well.

In \cite{Ahlswede1986}, the  deterministic code capacities under both the average and the maximum error criterion are derived, under the condition that the entire state sequence is non-causally known at the transmitter, while the jammer and the receiver have no further side information. 
The author used the \RT{} and \ET{} to derive the deterministic code capacity. 
This means, he started by proving a coding theorem for the \CC .
Hence, there exist codes for the \CC{} with exponentially vanishing error probability. 
Then, via permutations (\RT ) on the code for the \CC , he obtained a random code for the \AVC{} with slightly higher error probability, which is still exponentially vanishing. 
From this random code, he chose a subset of codes (\ET ). 
For this subset of codes the error probability vanishes linearly, instead of exponentially.  
If the deterministic code capacity is greater than zero, a prefix code can be concatenated with the smaller random code, to indicate which codebook realization is used during the transmission. 
If the amount of messages for this prefix code grows subexponentially (e.g. $n^2$), then the length of this prefix code grows sublinearly. 
Hence, the amount of codeword symbols  of the prefix code in the concatenated code vanishes and the deterministic code capacity equals the random code capacity. 
In \cite{Ahlswede1991a}, the \AVC{} theory is applied to computer memory and capacity formulas for different \CSI{} cases are derived.
In \cite{Ahlswede1999}, the deterministic code capacity region under the average error criterion for cooperating transmitters in an \AVMAC{} is described . 
Further, the authors provided a condition, when the achievable rate region has a non-empty interior.
In \cite{Winshtok2006}, the degraded \AVBC{} with non-causal \CSI T, full \CSI R at the stronger receiver, and statistical \CSI R at the weaker receiver is studied. 
The authors presented three main results: 
First, the single user deterministic code capacity under maximum error criterion of the degraded user is greater than zero if and only if the separation lemma in \cite{Ahlswede1978} holds with respect to the channel to the degraded receiver. 
Second, the deterministic code capacity region under the maximum error criterion equals the intersection of all capacity regions with respect to the jamming distributions if the single user deterministic code capacity under maximum error criterion of the degraded user is greater than zero. 
Otherwise it corresponds to the single user rate of the stronger receiver. 
Third, the capacity regions using deterministic, random or correlated codes under the average or maximums error criteria are equivalent. 
In \cite{Sarwate2007}, an \AVC{} is considered, where the jammer has non-causal access to the channel input and the message. 
Since the message is known non-causally at the jammer, the considered error probability has to be the maximum error probability. 
The authors used a list-code under the maximum error criterion approach to prove the random code capacity for this model.
In \cite{Cai2010}, the authors investigated  the \AVC{} with non-causal side information at the jammer. 
Furthermore, the authors imposed peak input and state constraints and derived the \CR{} assisted code capacity under the average and the maximum error criteria and compared these results. 
They limited the amount of \CR{}, that is needed to achieve the capacity, and stated that non-causal knowledge of the channel input at the jammer is leads to lower secrecy capacity than non-causal knowledge of the messages.
In \cite{Sarwate2010}, the situation of "nosy noise" where the channel input is perfectly known at the jammer, \cite{Sarwate2007}, is generalized to a "myopic adversary", where a jammer has a noisy version of the channel input as side information. 
Furthermore, a random code capacity formula under the maximum error criterion is derived.
In \cite{Wiese2013}, the random and deterministic code capacity regions for the \AVMAC{} with cooperating encoders is derived. 
Furthermore, the authors provide symmetrizability conditions for when the deterministic code capacity region for the \AVMAC{} with collaborating encoders has non-empty interior.
In \cite{Dey2013}, a variation of the \AVC{} is investigated. 
In this model, the attacker has causal knowledge of the channel input and can change a fraction of the codeword. 
The authors provided upper and lower bounds on the deterministic code capacity under the average and the maximum error criteria.
In \cite{Boche2014c}, the \AVMAC{} with cooperating encoders is studied, and the work of \cite{Wiese2013} extended. 
In contrast to \cite{Wiese2013}, here list codes are used and the deterministic list code capacity region is derived, which equals its random code capacity region if it is the channel is not list-symmetrizable. 
Otherwise the deterministic list code capacity region has empty interior.
In \cite{Schaefer2014a}, the deterministic and random code capacity regions under the average error criterion for the \AVBC{} with side information at the receiver are derived . 
Additionally (and as a counterpart to \cite{Boche2014c} and as an extension of \cite{Schaefer2014}), the authors considered the deterministic list code capacity and were able to prove a similar behavior as for deterministic list codes for the \AVC{} or the \AVMAC : the deterministic list code capacity either equals the random code capacity or has an empty interior if list size symmetrizability conditions are not fulfilled.   
In \cite{Pereg2017}, the degraded \AVBC{} with causal \CSI T is investigated.
For the random and deterministic code capacity regions, lower and upper bounds are derived, and the capacity regions for a class of channels, fulfilling the condition that there exists a jamming strategy which minimizes the mutual information terms between transmitter and the two receivers simultaneously, is derived. 
Here, the authors explicitly did not consider independent states for the individual channels. 
Furthermore, they provided the example of a binary symmetric \AVBC{} and presented for this example the capacity region. 
In \cite{Budkuley2017}, a version of the \AVC{} is considered, where the jammer and the transmitter have non-causal knowledge about the messages and the channel state (here not controlled by the jammer). 
Based on this knowledge the jammer can adopt its jamming signal, while the transmitter uses Gel'fand Pinsker or dirty paper coding to optimize the random code capacity under the maximum error criterion. 
For the dirty paper \AVC{} it was shown, that a memoryless Gaussian jamming strategy is the jammer's optimal choice. 
In \cite{Boche2019a}, an \AVCQC{} is investigated, where the jammer has side information about the channel input or both the channel input and the message. 
The authors determined the random code capacity for both average and maximum error criteria, and established a strong converse. 
Furthermore, all derived capacities are equal, the additional knowledge of the message does not decrease the capacity further.
In \cite{Beemer2020}, the authentication problem in the presence of an myopic adversary is considered. 
Equivalent to the symmetrizability condition for deterministic code for message transmission, the authors introduced the U-overwritability, and have shown that the authentication capacity either equals the authentication capacity without adversary, or equals zero if the channel is U-overwritable. 
%

\subsection{\acrfullpl{AVC} with Constraints}
Various works have  considered input and state constraints. 
We would like to give a brief overview. 
In \cite{Ahlswede1971}, the existence of channel capacities for the Gaussian \AVC{}(G\AVC{}) is proved. 
The author considered amplitude and average power constraints, as well as feedback, and provided explicit formulas for the capacities.
In \cite{Hughes1987}, a G\AVC{} under peak and average power constraints is considered . 
The authors were able to derive a random code capacity formula for the case of peak input and peak state constraints. 
In the cases of average constraints (on either input or states), the authors derived $\epsilon$-capacities for random codes.
In \cite{Csiszar1988a}, the \AVC{} with peak and average constraints on the channel input and the channel states is investigated. 
The authors have shown that for peak constraints the random code capacity exists. 
On the other hand, for any case of average constraints, only $\epsilon$-capacities have been proven to exist.
In \cite{Csiszar1988},  the \AVC{} with peak constraints is considered. 
The authors introduced a "cost"-function and have shown that if the jammer is not able to symmetrize the channel because of his state peak constraint, the deterministic code capacity might be positive, but less than the random code capacity. 
Furthermore, the authors proved that the symmetrizability condition from \cite{Ericson1985}, is not only sufficient but also necessary for the deterministic code capacity of an \AVC{} to be zero.
In \cite{Csiszar1991}, a Gaussian \AVC{} is investigated. 
The authors proved a deterministic code capacity for the case of peak input and peak jamming power constraints. 
In the case, where the peak input constraint is more stringent than the peak jamming constraint, the deterministic code capacity equals zero. 
This behavior is equivalent to the symmetrizability condition for finite \AVCs .
In \cite{Gubner1991a}, a discrete \AVMAC{} with state constraints is studied . 
In case of state constraints, the deterministic code capacity region might possess a non-empty interior, even if the channel is symmetrizable. 
Furthermore, the author provided a new weak converse under state constraints.
In \cite{Gubner1992}, the deterministic code capacity region for an additive \AVMAC{} under state constraint is provides.
In this scenario, the capacity region is a $45$ degree triangle and can be described by single letter expressions.
In \cite{Gubner}, convexity properties of the  \AVMAC{} with constraints are considered. 
The authors showed that the capacity region of independent stochastic encoders is not convex, in general. 
In \cite{Bross2003}, the single user Poisson \AVC{} and the two user Poisson \AVMAC , both under peak and average input and state constraints are studied. 
For both scenarios the authors derived the deterministic code capacity/capacity region under the average error criterion. 
They explicitly specified the decoders for each model, attaining the capacity/capacity region.
In \cite{Hof2006}, the discrete two user general \AVBC{} is studied. 
Based on \cite{Ericson1985}, the authors defined symmterizability conditions for the two user general \AVBC{} for when the interior of the deterministic code achievable rate region with and without state and input constraints is non-empty. 
They further considered achievable rate regions for degraded message sets.
In \cite{Wyrembelski2010}, a bidirectional \AVBC{} with peak input and state constraints is investigated. 
For this model, the authors derived the random and deterministic code capacity regions, and provided a symmetrizability and cost condition for the deterministic code capacity to have empty interior, based on \cite{Csiszar1988}.
In \cite{Sarwate2012}, list decoding for \AVCs{} under state constraints is considered . 
The authors have shown that rates (up to $\epsilon$ close) for random codes for the \AVCs{} with informed jammer can be achieved with small list size (of order $\mathcal{O} (\frac{1}{\epsilon})$). 
Furthermore, upper and lower bounds on the list-code capacity under the average error criterion with lists of size $L$ are provided.
In \cite{Sarwate2012a}, the author extended his work, \cite{Sarwate2007,Sarwate2010}, to the Gaussian case. 
Here, the jammer obtains a noisy version of the channel input and can choose his jamming signal, based on what he observed. 
Meanwhile, the transmitter and the jammer have peak power constraints.
In \cite{Mirmohseni2014}, two different attack strategies for the \AVC , while imposing a distortion constraint at the jammer, are studied. 
For the first attack strategy (memoryless), the authors derived a single letter capacity. 
For the second (foreseer), where the adversary has non-causal knowledge of the codeword, the authors differentiated between erasing and substituting attacks. 
For both, the authors gave lower and upper bounds on the capacity.
In \cite{Sarwate2018}, an \AVC{} with myopic adversary, who is subject to a quadratic state constraint is considered. 
For a specific range of noise-to-signal-ratios (NSR), the authors were able to characterize the deterministic code capacity. 
For the remaining region, they limited the amount of \CR{}. 
Furthermore, they introduced two new proof techniques, a myopic list-decoding result for the achievability, and a Plotkin-type push attack for the converse.
In \cite{Hosseinigoki2018}, the Gaussian \AVC{} under peak constraints using list decoding is investigated. 
The authors presented  a single letter formula for the deterministic list code capacity and showed that if the list size is smaller than the ratio of the transmit and jamming power, the capacity equals zero.
In \cite{Pereg2019}, the \AVC{} under peak and average input and state constraint with causal and non-causal \CSI T is studied. 
For the causal \CSI T case, the authors derived a lower bound on the deterministic code capacity for an message average input constraint, an lower and upper bounds on the random code capacity, which match if there are only constraints on the states but not on the input.
For the latter case, the authors provided a generalized symmetrizability condition for which the deterministic code capacity equals the random code capacity. 
For non-causal \CSI T, the random code capacity with constraints imposed on the states was derived, and again a condition was provided under which the deterministic code capacity equals the random code capacity.
\subsection{\acrfullpl{AVWC}}
If confidentiality requirements are combined with active attacks on communication systems, the \AVWC{} is the correct channel model. 
In the case where the channel state is determined by nature and there are secrecy requirements, the \CWC{} is an appropriate model. 
In the following, we give a brief literature overview of the \CWC{} and the \AVWC , without claiming completeness. 

%
In \cite{MolavianJazi2009}, random codes for the \AVWC{} are considered. 
The authors presented a single letter formula for achievable \CR{} assisted secrecy rates. 
Furthermore, the authors provided a single letter formula for the \CR{} assisted secrecy capacity for the strongly degraded case with independent states.
In \cite{Bjelakovic2013}, the \AVWC{} under the average error criterion is investigated. 
The authors combined strong secrecy requirements with Ahlswede's \ET , and were able to derive a single letter formula for the \CR{} assisted achievable secrecy rates. 
Additionally, the authors presented a multi letter formula for the deterministic code secrecy capacity. 
In \cite{Boche2015}, continuity properties of the secrecy capacities of \CWCs{} and \AVWCs{} are studied . 
The authors were able to show that for the \CWC{} the secrecy capacity is continuous with respect to the channel states. 
In contrast to the compound case, the authors were able to prove that the deterministic code secrecy capacity of an \AVWC{} possesses discontinuity properties with respect to the channel state. 
The authors presented an example in which the deterministic code secrecy capacity equals zero for a specific choice of the convex combination of channel states, while approaching this convex combination of channel states from above, the deterministic code secrecy capacity remains strictly lager than zero.  
In \cite{Noetzel2016}, the \AVWC{} is investigated and multi letter formulas for the \CR{} and deterministic code secrecy capacities for the case that the eavesdropper is kept ignorant about the \CR{} are derived. 
The authors proved that even though the deterministic code secrecy capacity possesses discontinuities, it is still stable around its positivity points.
Furthermore, the authors provided a complete characterization of \AVWCs{} which might possess the \SA{} property.
In \cite{Wiese2016}, a multi letter formula for the \CR{} assisted secrecy capacity in the general case and a single letter formula for the \CR{} assisted secrecy capacity in the strongly degraded case are proved. 
The authors considered both, average and maximum error criteria, and showed that the capacities are equivalent under both criteria. 
In \cite{Chen2021a}, multiple access \AVWC{} is considered. 
The authors derived a single letter achievable secrecy rate region and an multi letter upper bound. 
Furthermore, the authors calculated the secrecy capacity for the special case of a semi-noiseless \WTC . 
%
\subsection{\acrfullpl{AVWC} with Side Information}
In the literature, also different cases of side information at the transmitter and/or the jammer have been considered with secrecy constraints.

In \cite{Aggarwal2009}, the  binary \WTC{} of type II with an active eavesdropper, who observes a fraction of the transmitted codeword causally, is considered. 
The authors specifically investigated the cases where the eavesdropper erases his observed symbols, and where the eavesdropper flips his observed symbols. 
For these models, achievable secrecy rates are proved.
In \cite{Boche2012}, an \AVWC , where the active adversary has access to the \CR{}, is studied. 
This work relates the dichotomy behavior of the deterministic code capacity of \AVC{} to the case with secrecy requirements. 
The authors showed, that the if the \AVWC{} is symmetrizable then the \CR{} secrecy capacity of the \AVWC{} with knowledge of the commone randomness at the active adversary equals zero. Otherwise, it equals the  \CR{} secrecy capacity of the \AVWC .
In \cite{Bjelakovic2013a}, the \CWC{} with different \CSI{} cases is investigated. 
In the case of no \CSI , the authors derived a multi letter formula for the secrecy capacity. 
For different \CSI T cases, authors determined a single letter formula for the secrecy capacity. 
In \cite{Boche2013b}, the effects of causal knowledge of the \CR{} and \SA{} of \AVWCs{} are studied. 
The authors showed that the causal secrecy capacity equals the \CR{} assisted secrecy capacity. 
Furthermore, the authors demonstrated how the capacity of \AVWCs , when encoding jointly over the \AVWCs{} instead of encoding individually, can be strictly larger than the sum of the individual capacities. This phenomenon, known from the quantum case, is called \SA . 
Additionally, it is shown that weaker forms (e.g. correlated \CR{} instead of perfectly shared \CR{}) is sufficient to achieve the randomness assisted code capacity of an \AVWC . 
In \cite{S.Mansour2019}, the deterministic list code secrecy capacity of an \AVWC  is investigated . 
The authors provided a multi letter formula and presented a symmetrizability condition on the list size for the secrecy capacity to be zero.
In \cite{Goldfeld2020}, a \WTC{} with non-causal \CSI T is investigated. 
Under the maximum error and semantic security criteria a single letter formula for the achievable secrecy rate is derived. 
In \cite{Tahmasbi2020}, an \AVWC{} with causal \CSI T is considered. 
Based on the causal side information at the transmitter, a joint learning transmission scheme is established in order to learn the adversary's strategy. 
The authors showed that this transmission scheme leads to achievable rates (for some channel models), where the adversary's jamming choice is known non-causally at the transmitter.

%
\subsection{\acrfullpl{AVWC} with Constraints}
In \cite{Liang2009}, the discrete memoryless \CWCs , Gaussian \CWCs , and MIMO Gaussian \CWCs{} are  investigated. 
The authors derived single letter achievable secrecy rates for the general case and provided a single letter formula for the secrecy capacity for the strongly degraded case. 
For the Gaussian \CWC , they assumed peak input constraints (average transmit power constraints), 
provided a capacity formula for the strongly degraded, and calculated the secure degrees of freedom. 
For the MIMO Gaussian \CWC , the authors assumed peak input constraints, i.e. constraints on the transmit covariance matrix. 
They provided a secrecy capacity formula for the strongly degraded case, and presented a lower bound for the secure degrees of freedom.
In \cite{He2010}, a MIMO Gaussian \WTC{} is considered, where the eavesdropping channel is an \AVC . 
Under peak input constraints (constraints on the input covariance matrix), the authors contributed a single letter formula for achievable secrecy rates and calculated the secure degrees of freedom.
In \cite{Chou2013}, the channel model of secret key generation, where the eavesdropping channel is an \AVC , is considered. 
For the cases of finite alphabets and for MIMO Gaussian with peak input constraints (average transmit power constraints) the authors provided achievable secret key rates, and for the latter the secure degrees of freedom.
In \cite{Janda2014}, a \CWC{} with a distortion constraint and derived an achievable secrecy rate is considered. 
The authors studied symbol-, peak-, and average constraints on the state, and computed the jammer's best attack strategy. 
In each case, the attacker's best strategy is to flip each symbol with an \IID{} strategy. 
Furthermore, the Gaussian \CWC{} with equivalent constraints is investigated . 
Here, the jammer's best strategy is to jam every symbol with the same power.
In \cite{Janda2015}, the \AVWC{} with input and state peak constraints is investigated . 
The authors derived a multi letter formula for the achievable secrecy rate.
In \cite{Wang2016}, the author scrutinized a variation of the \AVWC , in which an adversary receives a fraction of the codeword perfectly (in terms of \WTC{} of type II) and modifies another fraction of the codeword, where the adversary can use his observed side information. 
For this model, the author determined upper and lower bounds on the semantic secrecy capacity.
In \cite{Goldfeld2016}, the authors used a strong soft covering lemma to derive a single letter formula of the random code semantic secrecy capacity of an \AVWC{} with type constrained states. 
In \cite{Chen2021}, deterministic wiretap codes for the \AVWC{} with input and state peak constraints  are considered.  
The authors provided a single letter formula for achievable secrecy rates.
%



%
%
\subsection{Contribution}
In this work, we consider the \AVWC{} with non-causal side information at the jammer. 
Non-causal side information means that codewords are known at an active adversary before they are transmitted.
We provide the single letter random code secrecy capacity under the maximum error criterion for the case that there exists a best channel to the eavesdropper. 
By considering the maximum error criterion, we allow the active attacker to know the messages, as well.
We use methods of \cite{Boche2019a}, hence random coding arguments instead of list codes, \cite{Sarwate2007}, which might be an alternative approach.
Furthermore, we derive a single letter random code secrecy capacity formula for the case that the eavesdropping channel is strongly degraded, strongly noisier, or strongly less capable with respect to the main channel. 
We compare our results to the random code secrecy capacity for the cases of maximum error criterion but without non-causal side information at the jammer, maximum error criterion with non-causal side information of the messages at the jammer, and the case of average error criterion without non-causal side information at the jammer.
By considering this model, we are able to describe situations, in which a communication system is subject to two different simultaneous attacks, eavesdropping and jamming attacks. 
For both, we individually assume worst case scenarios. 
By requiring a best channel to the eavesdropper, we also consider the case of colluding jammer and eavesdropper.
The eavesdropper obtains a perfect observation of the \CR{} shared between the legitimate communication partners. 
Hence, the \CR{} cannot be used as a key to encrypt the data.  
\begin{table}[t]
\begin{center}
\begin{tabular}{|p{.1\textwidth}|p{.15\textwidth}|p{.15\textwidth}|p{.05\textwidth}|p{.215\textwidth}|}
\hline
Reference & \CWC /\AVWC &  Side Information & Error & Result\\
\hhline{=====}
\cite{MolavianJazi2009}&\AVWC  &-&a&sl - r AR, r sd C \\
\hline
\cite{Aggarwal2009}& BAC, type II &CI&a&sl - d AR\\
\hline
\cite{Bjelakovic2013}&\AVWC &-&a& sl - r AR, ml - d C\\
\hline
\cite{Bjelakovic2013a}&\CWC & d-\CSI{} &a&ml - C, sl - C for special cases\\
\hline
\cite{Wiese2016}&\AVWC &-&a/m&ml - r C, sl - r sd C \\
\hline
\cite{Wang2016}& BAC , type I/II&CI&a&sl - d AR\\
\hline
\cite{Tahmasbi2020}&\AVWC &\CSI T&a&r C\\
\hline
\cite{Boche2019a}&CQ\AVC &MII/CII/MCII &a/m&sl C\\
\hline
\cite{Goldfeld2020}&\AVWC & non-causal \CSI T&m&sl - r AR, C (sp. cases)\\
\hline
\cite{Chen2021a}&\AVWC, MAC &-&a&sl - r AR, r C (sp. cases)\\
\hline
This work &\AVWC & MII/CII/MCII&a/m& sl - r C\\
\hline 
\end{tabular}
\caption{Literature overview related to the presented manuscript (without constraints and with secrecy requirements). Notation: Side Information - d-\CSI{} (different \CSI{} cases at the transmitter and receiver), MII/CII/MCII (message / channel input / message and channel input non-causally known at the jammer), PCI (a fraction of the channel input causally known at the jammer). Error - a (average error criterion), m (maximum error criterion). Result - sl (single letter), ml (multi letter), AR (achievable rate), C (capacity), r (randomness assisted), d (deterministic), sd (strongly degraded).}
\label{tab:literature}
\end{center}
\end{table}
In Table \ref{tab:literature}, we set our work into context. 
For this overview, we only considered state dependent channels with secrecy requirements, whose states are influenced by an external entity. 
But keep in mind, that there are publications without secrecy requirements, which are still highly related to this work, i.e, \cite{Sarwate2007,Cai2010,Sarwate2010}.
Since our work does not include constraints on the input or states, we excluded those works from the table, as well. 

The paper is organized as follows.  
We present the system model in Section \ref{sec:SYSMOD} and state our main result in Section \ref{sec:RESI}. 
Finally in Section \ref{sec:Discuss}, we compare our results to the the standard \AVWC , provide an example, and close with a discussion.
The proofs of the main results can be found in the appendices. 
\paragraph*{Notation}\label{sec:Not}
We folow the notation of \cite{Wiese2016}, and a list of the used symbols and their meanings can be found in Appendix \ref{ap:Nomenclature}. 
In particular, all logarithms are taken to the base $2$. 
Equivalently, the $\exp{\{.\}}$ function means $2^{\{.\}}$. 
Sets are denoted by calligraphic letters. 
The cardinality of a set $\mathcal{U}$ is denoted by $|\mathcal{U}|$.
The set of all probability measures on a set $\mathcal{U}$ is denoted by $\mathcal{P(U)}$. 
For $p\in \mathcal{P(U)}$ we define $p^n\in \mathcal{P}(\mathcal{U}^n)$ as $p^n(x^n)=\prod_i p(x_i)$. 
The entropies, and mutual information terms will be written in terms of the involved probability functions or in terms of the involved random variables.
For example
\begin{align*}
H(W|p) & := - \sum_{x,y} p(x)W(y|x) \log W(y|x)\\
I(p;W) & := H(pW)-H(W|p).
\end{align*} 
Furthermore, let the type of a sequence $s^n=(s_1, s_2, ..., s_n)$ be the probability measure $q\in \mathcal{P(S)}$ defined by $q(a)= \frac{1}{n}N(a|s^n)$, where $N(a|s^n)$ denotes the number of occurrences of $a$ in the sequence $s^n$. 
The set of all possible types of sequences of length $n$ is denoted by $\mathcal{P}_0^n(\mathcal{S})$. 
Additionally, for a $p\in \mathcal{P}(\mathcal{X})$ and $\delta >0$, we define the typical set $\mathcal{T}_{p,\delta}^n \subset \mathcal{X}^n$ as the set of sequences $x^n \in \mathcal{X}^n$ satisfying  for all $a \in \mathcal{X}$ the conditions 
\begin{align*}
\left| \frac{1}{n} N(a|x^n) - p(a)\right| & \leq \delta , \quad \text{if } p(a)>0,\quad \text{and }N(a|x^n)  = 0 \quad \text{if~} p(a)=0.
\end{align*} 

Similarly, for a $W\in \mathcal{P}(\mathcal{Y}|\mathcal{X})$ and a $\delta > 0$ we define the conditionally typical set $\mathcal{T}_{W,\delta}^n(x^n) \subset \mathcal{Y}^n$ as the set of sequences $y^n \in \mathcal{Y}^n$ satisfying for all $a\in \mathcal{X}$, $ b\in \mathcal{Y}$ the conditions

\begin{align*}
\left|\frac{1}{n}N(a,b|x^n,y^n) - W(b|a)\frac{1}{n}N(a|x^n)\right| & \leq \delta,\quad \text{if }W(b|a)>0,\\
N(a,b|x^n,y^n) = 0 \quad \text{if~} W(b|a) &= 0.
\end{align*}
See also {\cite[Chapter 2]{Csiszar2011}} for the method of types and the definitions of typical sequences.
\section{System Model}\label{sec:SYSMOD}
\begin{figure*}
	\centering
	
	\resizebox{!}{6cm}{
\providecommand*{\lvala}{2.5} 
\providecommand*{\lvalb}{1.5}
\providecommand*{\loffset}{0.3}
\tikzstyle{block} =  [rectangle, draw, text centered, minimum height=2em, node distance=1.5cm]
\tikzstyle{block2} = [rectangle, draw,dashed, text centered, minimum height=2em,fill=gray!20, node distance=0.5cm]
\tikzstyle{block3} = [rectangle, draw, text centered, minimum height=2em, text width=.75cm]
\tikzstyle{back} = [rectangle,rounded corners]
\begin{tikzpicture}[>=latex, thick, font=\scriptsize, scale=0.75]
\node[ellipse, draw](CR) {Common Randomness $ \mathcal{U}_n$};

\node[block,below of=CR, node distance=2cm](C2B){$W^n(y^n|x^n,s^n)$};

\node[block,below of=C2B](C2E){$V^n(z^n|x^n,s^n)$};

\begin{pgfonlayer}{background} 
\node[block2] (background) [fit = (C2B)(C2E)] {};
\end{pgfonlayer}

\node[block,left of=C2B,node distance=3.5cm](E){Encoder};

\node[left of=E, node distance=1.5cm](M){$J$};

\node[block,right of=C2B,node distance=4.5cm](DB){Decoder};
\node[block,right of=C2E,node distance=4.5cm](DE){Decoder};

\node[right of=DB, node distance=1.5cm](MB){$\hat{J}$};
\node[right of=DE, node distance=1.5cm](ME){$\hat{J^\prime}$};

\node[above of=M, node distance=.5cm](Alice){Alice};
\node[above of=MB, node distance=.5cm](Bob){Bob};
\node[below of=ME, node distance=.5cm](Eve){Eve};

\node[above of=E, node distance=1cm, xshift=.25cm](u1){$u$};
\node[above of=DB, node distance =1cm, xshift=-.25cm](u2){$u$};
\node[above of=DE, node distance =.5cm, xshift=-1.5cm](u3){$u$};

\node[block3, below of = E,node distance=2.75cm](f){$s^n$};
\node[left of=f, node distance=1.5cm](Jack){Jim};

\begin{pgfonlayer}{background} 
\node[back] (backgrounda) [fill=blue!20,fit = (Alice.north west)(E.south east)] {};
\node[back] (backgroundb) [fill=green!20,fit = (Bob.north east)(DB.south west)] {};
\node[back] (backgrounde) [fill=red!20,fit = (DE.north west)(Eve.south east)] {};
\node[back] (backgroundj) [fill=red!20,fit = (Jack.west)(f.north east)(f.south east)] {};
\end{pgfonlayer}

\draw[->](M)--(E);
\draw[->](DB)--(MB);
\draw[->](DE)--(ME);
\draw[->](E)--(C2B) node [midway, above](X){$X^n_u$};
\draw[->] (E.east) -- ($(E.east)!.5!(C2B.west)$) |- (C2E.west);
\draw[->] (E)--(f)node [midway,right](X1){$X^n_u$};
\draw[->](C2B)--(DB) node [midway, above](Y){$Y^n_{s^n}$};
\draw[->](C2E)--(DE) node [midway, below](Z){$Z^n_{s^n}$};
\draw[dashed,->](CR)[out=180 , in = 90]to(E);
\draw[dashed,->](CR)[out=0 , in = 90]to(DB);
\draw[dashed,->](CR)[out=0 , in = 180]to ($(DE.west)!.5!(DE.north west)$);
\draw[->](f)-|(background.south);
\draw[->](Jack)--(f);
\end{tikzpicture}
	}
	\caption{System model. Jammer has non-causal knowledge about the channel input.}
	\label{fig:SysMod}
\end{figure*}
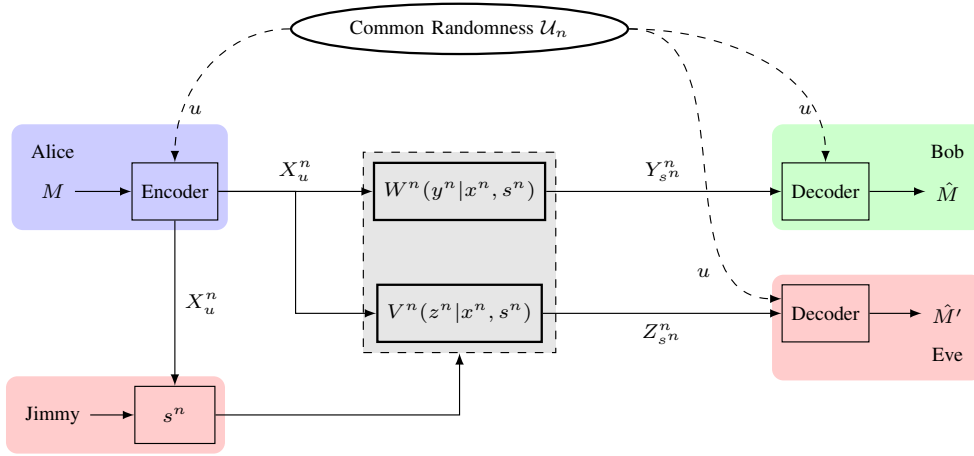
We consider a \CR{} assisted \AVWC{} as depicted in Fig. \ref{fig:SysMod}. 
A transmitter Alice tries to communicate reliably and securely with a legitimate receiver Bob in the presence of an eavesdropper Eve. 
The communication is done via state dependent \DMCs{} $W^n(y^n|x^n,s^n)$ and $V^n(z^n|x^n,s^n)$, where $s^n$ is the channel state, $x^n$ is the channel input, and $y^n$ and $z^n$ are the received sequences at Bob and Eve, respectively.
Alice, Bob, and Eve have access to a common source of randomness $\mathcal{U}_n$, whose realization can not be used as a key for encryption, since Eve also has access to it. 
The channel state $s^n$ is controlled by an external jammer Jim, who has non-causal access to the channel input $X_u^n$. 
The channel input of length $n$ is dependent on the the \CR{} realization, and hence indexed by it. 
Note that this system model  is considered without secrecy constraints by Sarwate \cite{Sarwate2007}, using a connection between deterministic list codes and random codes. 
Furthermore, this system model also is considered without secrecy constraints for the classical-quantum case by Boche et al. \cite{Boche2019a}. 
In the latter case, the authors use random coding arguments.  
\begin{remark}
By requiring a best channel to the eavesdropper, we can show that the jammer is not able to encode information about the channel input into the choice of the state sequence. 
Hence, if there is no other channel between the jammer and the eavesdropper, we also cover the situation of colluding attackers.   
\end{remark}
\begin{defi}[\acrlong{AVWC}]\label{def:AVWC}
We describe an \textbf{\acrlong{AVWC}}  by $(\mathcal{X}, \mathcal{S},\mathcal{W},\mathcal{V}, \mathcal{Y}, \mathcal{Z})$. 
Let $\mathcal{X}, \mathcal{S},\mathcal{Y}, \mathcal{Z}$ be finite sets.
The family of channels to the legitimate receiver is described by $\mathcal{W}=\{ (W_s:\mathcal{X}\to \mathcal{P}(\mathcal{Y})) : ~s\in \mathcal{S}\}$.
The family of channels to the illegitimate receiver is described by $\mathcal{V}=\{ (V_s:\mathcal{X}\to \mathcal{P}(\mathcal{Z})) : ~s\in \mathcal{S}\}$.
The channel is memoryless in the sense that the probability of receiving the sequences $y^n=(y_1,y_2,...,y_n)$ and $z^n=(z_1,z_2,...,z_n)$, when sending $x^n=(x_1,x_2,...,x_n)$ is 
\begin{align*}
W^n(y^n|x^n,s^n) & = \prod_{i=1}^{n}W(y_i|x_i,s_i) = \prod_{i=1}^{n}W_{s_i}(y_i|x_i) = W^n_{s^n}(y^n|x^n),\\
V^n(z^n|x^n,s^n) & = \prod_{i=1}^{n}V(z_i|x_i,s_i)= \prod_{i=1}^{n}V_{s_i}(z_i|x_i) = V^n_{s^n}(z^n|x^n).
\end{align*}
By $(\mathcal{W},\mathcal{V})$, we mean the \AVWC{} defined above.
\end{defi}
\begin{defi}[Deterministic Wiretap-Code]\label{def:DWTC}
An $(n,J_n)$ \textbf{deterministic wiretap-code} $\mathcal{K}_n$ consists of a stochastic encoder $E: \mathcal{J}_n\to \mathcal{P}(\mathcal{X}^n)$ and mutually disjoint decoding sets $\mathcal{D}_{j} \subset \mathcal{Y}^n$, $\mathcal{D}_{j} \cap \mathcal{D}_{j^\prime}= \emptyset, ~ j,j^\prime\in\mathcal{J}_n$.
We define $EW_{s^n}^n:\mathcal{J}_n \to \mathcal{P}(\mathcal{Y}^n)$ by
\begin{align*}
EW_{s^n}^n(y^n|j) & = \sum_{x^n\in\mathcal{X}^n}E(x^n|j)W^n(y^n|x^n,s^n).
\end{align*}
The maximum error $e(\mathcal{K}_n)$ for the \AVWC{} can be expressed as 
\begin{align*}
e(\mathcal{K}_n) & :=  \max_{s^n\in \mathcal{S}^n}\max_{j\in \mathcal{J}_n} \sum_{x^n\in \mathcal{X}^n} E(x^n|j)W^n(\mathcal{D}_{j}^c|x^n,s^n)
\end{align*}
If the jammer has non-causal knowledge about the channel input $x^n$, then the maximum error probability has to be expressed as 

\begin{align*}
\hat{e}(\mathcal{K}_n) & :=  \max_{\substack{f\in\mathcal{F}}}\max_{j\in \mathcal{J}_n} \sum_{x^n\in \mathcal{X}^n}E(x^n|j) W^n(\mathcal{D}_{j}^c|x^n,f(x^n)),
\end{align*}

for all deterministic jamming functions $\mathcal{F}:\mathcal{X}^n \to \mathcal{S}^n$.
\end{defi}
\begin{defi}[Common Randomness Assisted Wiretap Code]\label{def:RC}
An $(n,J_n,\mathcal{U}_n,p_{U})$ \textbf{\CR{} assisted wiretap code} $\mathcal{K}_n^{\text{ran}}$ consists of a family of stochastic encoders $\mathcal{E}=\{ (E_u: \mathcal{J}_n\to \mathcal{P}(\mathcal{X}^n)):~ u\in \mathcal{U}_n\}$ and mutually disjoint (for fixed $u$) decoding sets $\mathcal{D}_{j,u}\subset \mathcal{Y}^n,~ \mathcal{D}_{j,u} \cap \mathcal{D}_{j^\prime,u}\neq \emptyset, ~ j,j^\prime\in\mathcal{J}_n,~u\in\mathcal{U}_n$ with message set $\mathcal{J}_n:= \{1,...,J_n\}$, and $p_U\in \mathcal{P(U)}$. 
Note that for different realizations of the \CR{} $\mathcal{U}_n$, $u\neq u^\prime$, the decoding sets do not have to be disjoint, $\mathcal{D}_{j,u} \cap \mathcal{D}_{j^\prime,u^\prime}\neq \emptyset$
The maximum error probability averaged over all possible randomly chosen deterministic wiretap codebooks $e(\mathcal{K}_n^{\text{ran}})$ can be written as 
\begin{align*}
e(\mathcal{K}_n^{\text{ran}}) & :=  \max_{s^n\in \mathcal{S}^n}\max_{j\in \mathcal{J}_n} \sum_{u\in \mathcal{U}_n}p_U(u)\sum_{x^n\in \mathcal{X}^n}E_u(x^n|j)W^n(\mathcal{D}_{j,u}^c|x^n,s^n).
\end{align*}
Here, the jammer does not know the channel input non-causally.

We define the channel $p_{X^nU|J}:\mathcal{J}_n\to\mathcal{P}(\mathcal{X}^n\times \mathcal{U})$ as 
\begin{align*}
p_{X^nU|J}(x^n,u|j) & = p_{X^n|JU}(x^n|j,u)p_U(u) = E_u(x^n|j)p_U(u).
\end{align*}
Let $\mathcal{F}: \mathcal{X}^n\to\mathcal{S}^n$ describe the family of all deterministic mappings from $\mathcal{X}^n$ to $\mathcal{S}^n$. 
If the jammer has non-causal knowledge of the channel input $x^n$, then the maximum error probability has to be adapted to

\begin{align*}
\hat{e}(\mathcal{K}_n^{\text{ran}}) & :=  \max_{\substack{f\in\mathcal{F}}}\max_{j\in \mathcal{J}_n} \sum_{x^n\in \mathcal{X}^n}p_{X^n|J}(x^n|j)\sum_{u\in \mathcal{U}_n}p_{U|X^n,J}(u|x^n,j)  W^n(\mathcal{D}_{j,u}^c|x^n,f(x^n)).
\end{align*}

\end{defi}

\begin{remark}
In contrast to the standard \AVWC{}, here in the case of non-causal knowledge at the jammer the maximization of $s^n$ is done within each term of the sum. 
Since the jammer knows the channel input, he can adopt to that specific codeword choice. 

Furthermore, let $\mathcal{F}^\prime$ be the family of all deterministic mappings $\mathcal{J}_n\times \mathcal{X}^n \to \mathcal{S}^n $, and $\mathcal{F}^{\prime \prime}$ be the family of all deterministic mappings $\mathcal{J}_n\to \mathcal{S}^n $.
From of Lemma \ref{lem:maxeqmax}, we have 
\begin{align*}
e(\mathcal{K}_n) & =  \max_{s^n\in \mathcal{S}^n}\max_{j\in \mathcal{J}_n} \sum_{x^n\in \mathcal{X}^n} E(x^n|j)W^n(\mathcal{D}_{j}^c|x^n,s^n)\\
& =  \max_{j\in \mathcal{J}_n}\max_{f^{\prime \prime}\in \mathcal{F}^{\prime \prime}} \sum_{x^n\in \mathcal{X}^n} E(x^n|j)W^n(\mathcal{D}_{j}^c|x^n,f^{\prime \prime}(j)),\quad  \text{and}\\
\hat{e}(\mathcal{K}_n) & =  \max_{\substack{f\in\mathcal{F}}}\max_{j\in \mathcal{J}_n} \sum_{x^n\in \mathcal{X}^n}E(x^n|j) W^n(\mathcal{D}_{j}^c|x^n,f(x^n))\\
& =  \max_{j\in \mathcal{J}_n}\max_{\substack{f^\prime\in\mathcal{F}^\prime}} \sum_{x^n\in \mathcal{X}^n}E(x^n|j) W^n(\mathcal{D}_{j}^c|x^n,f^\prime(x^n,j)).
\end{align*}
That implies the following statement. 
Considering the maximum error probability (with respect to the messages) corresponds to the case, where the jammer additionally knows the messages, because the maximization orders can be exchanged according to Lemma \ref{lem:maxeqmax} (see also \cite{Cai2010}). 
Furthermore, the inner optimization is done for fixed parameter of the outer optimization. 
That means for each given message $j\in\mathcal{J}_n$, the worst case state sequence will be considered. This implies the above equalities. 
Equivalent statements hold for the \CR{} assisted codes.    
\end{remark}
\begin{defi}[Achievable Common Randomness Assisted Secrecy Rates and Common Randomness Assisted Secrecy Capacities]\label{def:RCSC}
A nonnegative number $R_S$ is called an \textbf{achievable \CR{} assisted secrecy rate} for the \AVWC{} if there exists a sequence $(\mathcal{K}_n^{\text{ran}})_{n=1}^{\infty}$ of $(n,J_n,\mathcal{U}_n,p_{U})$ \CR{} assisted codes for uniformly distributed messages, such that the following requirements are fulfilled
\begin{align}
\liminf_{n\to \infty} \frac{1}{n}\log J_n & \geq R_S,\label{eq:rateconditionI}\\
\lim_{n\to \infty} e(\mathcal{K}_n^{\text{ran}}) & = 0,\label{eq:errorconditionI}\\
\lim_{n\to \infty}\max_{s^n\in\mathcal{S}^n}\max_{u\in \mathcal{U}_n}I(p_{J};E_uV_{s^n}^n) &= 0.\label{eq:leakageconditionI}
\end{align}
A nonnegative number \textbf{$\widehat{\widehat{R}}_S$ is called an achievable \CR{} assisted secrecy rate for the \AVWC{} with non-causal knowledge of the channel input at the jammer} if there exists a sequence $(\mathcal{K}_n^{\text{ran}})_{n=1}^{\infty}$ of $(n,J_n,\mathcal{U}_n,p_{U})$ \CR{} assisted codes for uniformly distributed messages, such that the following requirements are fulfilled
\begin{align}
\liminf_{n\to \infty} \frac{1}{n}\log J_n & \geq \widehat{\widehat{R}}_S,\label{eq:rateconditionII}\\
\lim_{n\to \infty} \hat{e}(\mathcal{K}_n^{\text{ran}}) & = 0,\label{eq:errorconditionII}\\
\lim_{n\to \infty}\max_{\substack{{f\in\mathcal{F}}}}\max_{u\in \mathcal{U}_n}I(p_{J};E_uV_{f}^n) &= 0. \label{eq:leakageconditionII}
\end{align}
%
The supremum of all achievable \CR{} assisted secrecy rates for the \AVWC{} is called the \textbf{\CR{} assisted secrecy capacity} of the \AVWC{} $(\mathcal{W},\mathcal{V})$ and is denoted by $\widehat{C}_S^{\text{ran}}(\mathcal{W},\mathcal{V})$, when the jammer has no knowledge about the channel input, and $\widehat{\widehat{C}}_S^{\text{ran}}(\mathcal{W},\mathcal{V})$ if the jammer has non-causal knowledge of the channel input.
\end{defi}
The secrecy capacity $\widehat{C}_S^{\text{ran}}(\mathcal{W},\mathcal{V})$ is lower bounded by $\widehat{\widehat{C}}_S^{\text{ran}}(\mathcal{W},\mathcal{V})$.
Note that the eavesdropper has access to the \CR , too. 
Hence, the randomness cannot be used as a key to ensure secure communication between Alice and Bob. 
We explicitly do not bound the cardinality of the \CR{}. 
In \cite{Boche2019a}, the authors provide capacity formulas for quantum channels with an informed jammer but without secrecy constraints. The authors additionally relate and compare the capacity formulas for the cases that the jammer knows additionally the messages and that the jammer does not know the messages.
\begin{lem}\label{rem:theFunc}
Let $\mathcal{P}(\mathcal{S}^n|\mathcal{X}^n)$ be the set of all conditional probability distributions of the state sequences $s^n\in\mathcal{S}^n$ given the channel input $x^n\in\mathcal{X}^n$. 
We can in fact consider the maximization over $\theta \in \mathcal{P}(\mathcal{S}^n|\mathcal{X}^n)$ instead of considering the maximization over all deterministic mappings $\mathcal{F}:\mathcal{X}^n \to \mathcal{S}^n$. 
\end{lem}
\begin{proof}[Proof of Lemma \ref{rem:theFunc}]
See Appendix \ref{proof:theFunc}.
\end{proof}

\begin{defi}[Convex closure and row convex closure \cite{Ahlswede1978}]
Let $p\in\mathcal{P(S)}$ and $\hat{p}\in \mathcal{P(S|X)}$ be probability measures. 
The \textbf{convex closure} and the \textbf{row convex closure} of the \AVC{} are defined as 
\begin{align}
\widehat{\mathcal{W}} & :=\left\{W_p(\cdot|\cdot): \sum_{s\in\mathcal{S}}p(s)W(\cdot|\cdot,s), \quad p\in\mathcal{P(S)}\right\}\\
\widehat{\widehat{\mathcal{W}}} & :=\left\{W_{\hat{p}}(\cdot|x): \sum_{s\in\mathcal{S}}\hat{p}(s|x)W(\cdot|x,s),  \quad \hat{p}(s|x)\in \mathcal{P(S|X)}, x \in \mathcal{X},\right\}
\end{align}
\end{defi}
\begin{ex}\label{ex:conv}
Let $\mathcal{X}=\mathcal{Y}=\mathcal{S}=\{0,1\}$, and 
\begin{align*}
W(\cdot|\cdot,S=0) & = \begin{pmatrix}
1&0\\
0&1
\end{pmatrix},&W(\cdot|\cdot,S=1) & = \begin{pmatrix}
0&1\\
1&0
\end{pmatrix}.
\end{align*}
The convex closure and the row convex closure are given respectively as 
\begin{align*}
\widehat{\mathcal{W}} & = \left\{W(\cdot|\cdot):\begin{pmatrix}
\alpha&1-\alpha\\
1-\alpha&\alpha
\end{pmatrix},\quad \alpha\in [0,1]\right\},\quad \widehat{\widehat{\mathcal{W}}} & = \left\{W(\cdot|\cdot):\begin{pmatrix}
\alpha&1-\alpha\\
1-\beta&\beta
\end{pmatrix},\quad \alpha,\beta \in [0,1]\right\}.
\end{align*}
\end{ex}
\begin{defi}[$k$-Letter extension of $\widehat{\widehat{\mathcal{W}}}$]
The \textbf{$k$-letter extension of $\widehat{\widehat{\mathcal{W}}}$} is defined as the set 
\begin{align}
\widetilde{\mathcal{W}}^k & :=\left\{W_{\hat{p}}^k(Y^k|X^k): \sum_{s^k\in\mathcal{S}^k}\hat{p}(s^k|x^k)W^k(\cdot|x^k,s^k),\quad \hat{p}(s^k|x^k)\in \mathcal{P}(\mathcal{S}^k|\mathcal{X}^k), x^k \in \mathcal{X}^k\right\}
\end{align}
\end{defi}

\begin{remark}
Note that $\widetilde{\mathcal{W}}^k \neq \widehat{\widehat{\mathcal{W}}}^k$. 
It can be shown that the operations of the Kronecker product and taking the row convex closure are not commutative.
\end{remark}
\addtocounter{ex}{-1}
\begin{ex}[continued]
We have 
\begin{align*}
W(\cdot|\cdot,S=0) \otimes  W(\cdot|\cdot,S=0) & = \begin{pmatrix}
1&0&0&0\\
0&1&0&0\\
0&0&1&0\\
0&0&0&1
\end{pmatrix}&W(\cdot|\cdot,S=1) \otimes  W(\cdot|\cdot,S=1) & = \begin{pmatrix}
0&0&0&1\\
0&0&1&0\\
0&1&0&0\\
1&0&0&0
\end{pmatrix}\\
W(\cdot|\cdot,S=0) \otimes  W(\cdot|\cdot,S=1) & = \begin{pmatrix}
0&1&0&0\\
1&0&0&0\\
0&0&0&1\\
0&0&1&0
\end{pmatrix}&W(\cdot|\cdot,S=1) \otimes  W(\cdot|\cdot,S=0) & = \begin{pmatrix}
0&0&1&0\\
0&0&0&1\\
1&0&0&0\\
0&1&0&0
\end{pmatrix}
\end{align*}
Hence, when taking the row convex closure now, we obtain
\begin{align*}
\widetilde{\mathcal{W}}^2 & =  \left\{\begin{pmatrix}
\alpha_1&\alpha_2&\alpha_3&1-\alpha_1- \alpha_2 -\alpha_3\\
\beta_1&\beta_2&\beta_3&1-\beta_1 -\beta_2 -\beta_3\\
\gamma_1&\gamma_2&\gamma_3&1-\gamma_1 -\gamma_2 -\gamma_3\\
\delta_1&\delta_2&\delta_3&1-\delta_1- \delta_2 -\delta_3\\
\end{pmatrix}:\quad \alpha_i , \beta_i, \gamma_i, \delta_i \in[0,1],~ i\in\{1,2,3\},~ \sum_{i=1}^3\alpha_i=\sum_{i=1}^3\beta_i=\sum_{i=1}^3\gamma_i=1\right\}
\end{align*}
In contrast, when taking the row convex closure first, and then calculating the two letter extension, we obtain
\begin{align*}
&\widehat{\widehat{W}}_1(\cdot|\cdot) \otimes \widehat{\widehat{W}}_2(\cdot|\cdot)) = \begin{pmatrix}
\alpha_1&1-\alpha_1\\
1-\beta_1&\beta_1
\end{pmatrix} \otimes \begin{pmatrix}
\alpha_2&1-\alpha_2\\
1-\beta_2&\beta_2
\end{pmatrix}\\
&\widehat{\widehat{\mathcal{W}}}^2=\left\{\begin{pmatrix}
\alpha_1 \alpha_2&\alpha_1(1-\alpha_2)&(1-\alpha_1)\alpha_2&(1-\alpha_1)(1-\alpha_2)\\
\alpha_1 (1-\beta_2)&\alpha_1\beta_2&(1-\alpha_1)(1-\beta_2)&(1-\alpha_1)\beta_2\\
(1-\beta_1)\alpha_2&(1-\beta_1)(1-\alpha_2)&\beta_1 \alpha_2&\beta_1(1-\alpha_2)\\
(1-\beta_1)(1-\beta_2)&(1-\beta_1)\beta_2&\beta_1 (1-\beta_2)&\beta_1\beta_2
\end{pmatrix}: \quad  \alpha_i , \beta_i \in[0,1],~ i\in\{1,2\} \right\}.
\end{align*} 
It is easy to see that the row $\begin{bmatrix}\frac{1}{3}&\frac{1}{3}&\frac{1}{3}&0\end{bmatrix}$ is achievable in $\widetilde{\mathcal{W}}^2$ but not in $\widehat{\widehat{\mathcal{W}}}^2$.
\end{ex}

\begin{remark}[Notation]
With slight abuse of notation, we use the subscripts of $V$ and $W$ to show the dependence on the state sequence $s^n$, the deterministic mapping $f\in \mathcal{F}$, $\mathcal{F}:\mathcal{X}^n\to\mathcal{S}^n$ and stochastic mappings $\theta\in \mathcal{P}(\mathcal{S}^n|\mathcal{X}^n)$. 
	Since we use certain notations interchangeably, we clarify them in the following (shown for $V$).
	\begin{align}
	V^n(z^n|x^n,s^n) & = V^n_{s^n}(z^n|x^n),\label{eq:not1}\\
	V^n(z^n|x^n,f(x^n)) & = V^n_{f}(z^n|x^n),\label{eq:not2}\\
	V_{\theta}^n(z^n|x^n)  &= \sum_{s^n\in \mathcal{S}^n}\theta(s^n|x^n)V^n(z^n|x^n,s^n),\label{eq:not3}\\
	V^n_f(z^n|j) & = E_uV_{f}^n =\sum_{x^n\in \mathcal{X}^n}E_u(x^n|j)V^n(z^n|x^n,f(x^n)),\label{eq:not4}\\
	V^n_{\theta}(z^n|j) & =E_uV_{\theta}^n  =\sum_{x^n\in \mathcal{X}^n}E_u(x^n|j)\sum_{s^n\in \mathcal{S}^n}\theta(s^n|x^n)V^n(z^n|x^n,s^n).\label{eq:not5}
\end{align}
	Here, (\ref{eq:not1}) denotes the \AVC{} $V^n$ to the eavesdropper if the channel input equals $x^n$, the channel state is $s^n$, and the channel output equals $z^n$. 
	We use the notation in  (\ref{eq:not1}) interchangeably. 
	In (\ref{eq:not2}) we consider the same \AVC{}, but under the condition that the jammer applies the deterministic mapping $f\in\mathcal{F}$, $\mathcal{F}: \mathcal{X}^n\to \mathcal{S}^n$. 
	Again, we use the notation in  (\ref{eq:not2}) interchangeably. 
	In (\ref{eq:not3}), we consider a stochastic mapping $\theta \in \mathcal{P}(\mathcal{S}^n|\mathcal{X}^n)$ instead of a deterministic mapping. 
	Hence, we consider the averaged channel with respect to the channel state $s^n$ in dependence on the channel input $x^n$. 
	In (\ref{eq:not4}), we denote the conditional probability of obtaining the output sequence $z^n$ under the conditions that we transmitted the secure message $j\in\mathcal{J}_n$ and that the jammer applies the deterministic jamming strategy $f\in\mathcal{F}$. 
	Since we use the stochastic encoder $E_u$, we average with respect to the channel input $x^n\in\mathcal{X}^n$. 
	In (\ref{eq:not5}), the jammer applies a stochastic jamming strategy $\theta \in \mathcal{P}(\mathcal{S}^n|\mathcal{X}^n)$ instead of a deterministic mapping. 
	Since we use again a stochastic encoder $E_u$, we average with respect to the channel input $x^n$ and with respect to the channel states $s^n$.
\end{remark}

%
%
%

\begin{defi}[Best Channel to the Eavesdropper]\label{def:bestchannel}
Let $Z_{\theta}^n$ be the output of the channel $V_{\theta}^n$.
If there exists for all $n\in\mathbb{N}$ a $\theta^{\ast , n} \in \mathcal{P}^n(\mathcal{S}|\mathcal{X})$ with $\theta^{\ast , n}(s^n|x^n)=\prod_{i=1}^{n}\theta^{\ast}_i(s_i|x_i)=\prod_{i=1}^{n}\theta^{\ast}(s_i|x_i)$  such that for all other $\theta \in \mathcal{P}(\mathcal{S}^n|\mathcal{X}^n)$ the Markov chain 
\begin{align}
X^n \leftrightarrow Z_{\theta^{\ast , n}}^n \leftrightarrow Z_{\theta}^n, \
\end{align}
holds, then we say that there exists a \textbf{best channel to the eavesdropper} and all channels $V_{\theta}^n$ are degraded with respect to the channel $V_{\theta^{\ast , n}}^n$.
\end{defi}
\begin{remark}
Since the mutual information is convex (row convex) with respect to the channel for fixed input/ input distribution, the optimal jamming strategy is deterministic.
\begin{align*}
\theta^{\ast , n}(s^n|x^n) & = \mathds{1}_{s^{\ast ,n}}(x^n)
\end{align*}
In other words, the optimal state sequence (in terms of the secrecy constraint) results in a boundary point of $\widetilde{\mathcal{V}}^n$ and taking convex combinations of channel states does not increase the mutual information. 
A similar statement can be made with respect to the error probability. 
Since the mutual information is convex (row convex) with respect to the channel for fixed input/ input distribution, the optimal jamming strategy with respect to the reliability constraint is deterministic again, but is not a boundary point of $\widetilde{\mathcal{V}}^n$.
\end{remark}
Next, we will introduce the notions of strongly degraded, strongly noisier, and strongly less capable with independent states, respectively. 
Independent states mean that the states in the main and the eavesdropping channel can be chosen individually.
\begin{defi}[Strongly Degraded]
An \AVWC{}  is \textbf{strongly degraded} (with independent states, see \cite{MolavianJazi2009}) if the following Markov chain holds
\begin{align*}
X^n \leftrightarrow Y_{\theta}^n \leftrightarrow Z_{\theta^{\prime }}^n, \quad \forall \theta , \theta^{\prime} \in \mathcal{P}(\mathcal{S}^n|\mathcal{X}^n), \quad \forall n\in\mathbb{N}.
\end{align*} 
\end{defi}
\begin{defi}[Strongly Noisier with Independent States]
The family of channels to the illegitimate receiver $\mathcal{V}=\{ (V_s:\mathcal{X}\to \mathcal{P}(\mathcal{Z})) : ~s\in \mathcal{S}\}$ is \textbf{strongly noisier} with independent states than the family of channels to the legitimate receiver $\mathcal{W}=\{ (W_s:\mathcal{X}\to \mathcal{P}(\mathcal{Y})) : ~s\in \mathcal{S}\}$ if for every random variable $A^n$ such that $A^n\leftrightarrow X^n \leftrightarrow (Y_{\theta}^n,Z_{\theta^{\prime}}^n)$ we have for all $\theta,\theta^{\prime}\in\mathcal{P}(\mathcal{S}^n|\mathcal{X}^n)$
\begin{align*}
I(p_{A^n}^n;W_{\theta}^n) & \geq I(p_{A^n}^n;V_{\theta^{\prime}}^n), \quad \forall n\in\mathbb{N}.
\end{align*}

\end{defi}
\begin{defi}[Strongly Less Capable with Independent States]
The family of channels to the illegitimate receiver $\mathcal{V}=\{ (V_s:\mathcal{X}\to \mathcal{P}(\mathcal{Z})) : ~s\in \mathcal{S}\}$ is \textbf{strongly less capable} with independent states than the family of channels to the legitimate receiver $\mathcal{W}=\{ (W_s:\mathcal{X}\to \mathcal{P}(\mathcal{Y})) : ~s\in \mathcal{S}\}$ if for every  $p\in\mathcal{P}(\mathcal{X}^n)$ we have for all $\theta,\theta^{\prime}\in\mathcal{P}(\mathcal{S}^n|\mathcal{X}^n)$ 
\begin{align*}
I(p;W_{\theta}^n) & \geq I(p;V_{\theta^{\prime}}^n), \quad \forall n\in\mathbb{N}.
\end{align*}

\end{defi}
\begin{remark}
If there exist a $\theta \in \mathcal{P}(\mathcal{S}|\mathcal{X})$ fulfilling the strongly degraded, strongly less noisier, or strongly less capable condition, then there exists for all $n\in\mathbb{N}$ a $\theta^n \in \mathcal{P}^n(\mathcal{S}|\mathcal{X})$, with $\theta^n = \prod_{i=1}^n \theta_i$, fulfilling the strongly degraded, strongly less noisier, or strongly less capable condition, respectively.
\end{remark}
\begin{remark}
Just as in the stateless case \cite{Bloch2011a}, we have the following implication chain:
\begin{align*}
\text{Strongly Degraded} & \to \text{Strongly Noisier} \to \text{Strongly Less Capable}.
\end{align*}
Here, $X\to Y$ means $X$ implies $Y$, but not vice versa.
\end{remark}
\section{Main Results}\label{sec:RESI}
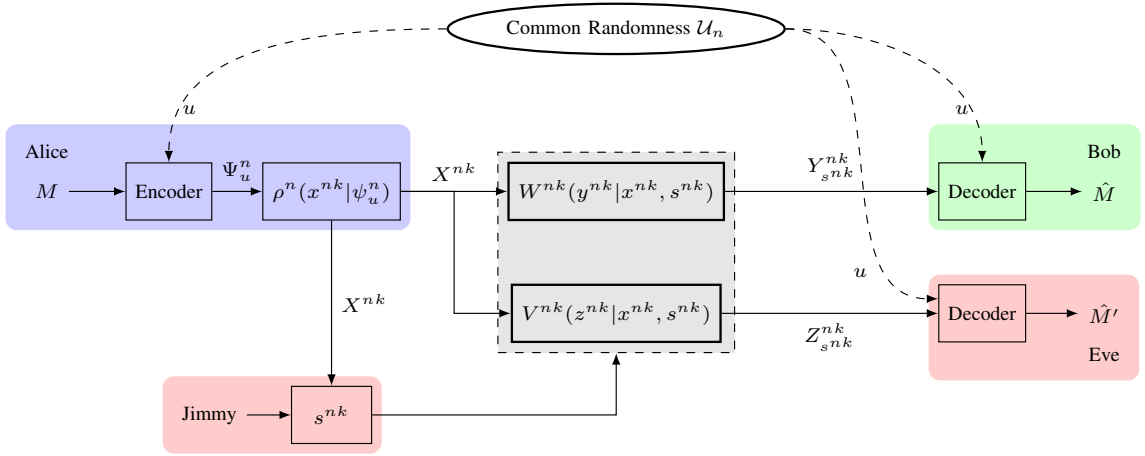
\begin{figure*}
	\centering
	
	\resizebox{!}{3cm}{
\providecommand*{\lvala}{2.5} 
\providecommand*{\lvalb}{1.5}
\providecommand*{\loffset}{0.3}
\tikzstyle{block} =  [rectangle, draw, text centered, minimum height=2em, node distance=1.5cm]
\tikzstyle{block2} = [rectangle, draw,dashed, text centered, minimum height=2em,fill=gray!20, node distance=0.5cm]
\tikzstyle{block3} = [rectangle, draw, text centered, minimum height=2em, text width=.75cm]
\tikzstyle{back} = [rectangle,rounded corners]
\begin{tikzpicture}[>=latex, thick, font=\scriptsize, scale=0.75]
\node[block](E){$\rho^n(x^{n}|\psi^n_u)$};
\node[block,left of=E,node distance=2cm](P){Encoder};
\node[right of=E, node distance=2cm](C2B){$\cdots$};
\node[left of=P, node distance=1.5cm](M){$J$};

\node[above of=M, node distance=.5cm](Alice){Alice};

\node[above of=P, node distance=1.5cm](u1){$u$};

\node[below of = E,node distance=1.5cm](f){$\vdots$};

\begin{pgfonlayer}{background} 
\node[back] (backgrounda) [fill=blue!20,fit = (Alice.north west)(E.south east)] {};
\end{pgfonlayer}

\draw[->](M)--(P);
\draw[->](P)--(E)node [midway, above](psi){$\Psi^n_u$};
\draw[->](E)--(C2B) node [midway, above](X){$X^{n}$};
\draw[->] (E)--(f)node [midway,right](X1){$X^{n}$};
\draw[dashed,->](u1)to(P);
\end{tikzpicture}
	}
	\caption{Adopted system model with prefixing at Alice. With \CR{} realization $u$, Alice encodes a secure message $J$ into a codeword $\Psi_u^n$, of length $n$. The codeword serves as the input of a prefix channel $\rho(x^n|)\psi_u^n$, and is mapped to the channel input $X^n$. Other parts remain the same.}
	\label{fig:SysModPre}
\end{figure*}
In the following, we state our main results. 
First, we present the secrecy capacity formulas for the general, and the strongly less capable cases, respectively, when the jammer has non-causal knowledge of the channel input. 
Then we provide the corresponding secrecy capacity formulas, when the jammer has no side information or only possesses knowledge of the messages.
\subsection{Capacity Formulas for the General and the Less Capable Cases}
\begin{theo}\label{th:general}
If there exists a best channel to the eavesdropper, the \CR{} assisted code secrecy capacity for the  \AVWC{} with side information at the jammer $\widehat{\widehat{C}}_S^{\text{ran}}(\mathcal{W},\mathcal{V})$ is given by 
\begin{align}
\widehat{\widehat{C}}_S^{\text{ran}}(\mathcal{W},\mathcal{V}) & = \max_{p_\Psi , \rho(X|\Psi)}\left(\min_{W\in\widehat{\widehat{\mathcal{W}}}}I(p_{\Psi};\rho W) - \max_{V\in\widehat{\widehat{\mathcal{V}}}}I(p_{\Psi};\rho V)\right),
\end{align}

with $\Psi$ as a prefixing random variable and concatenated channels $\rho W$ and $\rho V$, respectively.
\end{theo}
\begin{proof}[Proof of Theorem \ref{th:general}]
See Appendix \ref{proof:th-general}.
\end{proof}
%
%
%
%

\begin{theo}\label{th:degraded}
Let an \AVWC{} $(\mathcal{W},\mathcal{V})$ be given. If for $(\mathcal{W},\mathcal{V})$, the channel $\mathcal{V}$ is strongly less capable with respect to the channel $\mathcal{W}$ and if there exists a best channel to the eavesdropper, then the \CR{} assisted code secrecy capacity $\widehat{\widehat{C}}_S^{\text{ran}}(\mathcal{W},\mathcal{V})$ is given by
%
\begin{align}
\widehat{\widehat{C}}_S^{\text{ran}}(\mathcal{W},\mathcal{V}) & = \max_{p_X}\left(\min_{W\in \widehat{\widehat{\mathcal{W}}}} I(p_X;W) - \max_{V\in \widehat{\widehat{\mathcal{V}}}}I(p_X;V)\right)\label{eq:theo3}
\end{align}
\end{theo}
\begin{proof}[Proof of Theorem \ref{th:degraded}]
See Appendix \ref{proof:th-degraded}.
\end{proof}
The secrecy capacity $\widehat{\widehat{C}}_S^{\text{ran}}(\mathcal{W},\mathcal{V})$ depends on the row convex closures $\widehat{\widehat{\mathcal{W}}}$ and $\widehat{\widehat{\mathcal{V}}}$.
%
%
\subsection{Capacity Formulas without Side Information at the Jammer, or where the Jammer only knows the Messages }
\begin{cor}\label{th:noknowledge-or-only-the-messages}
Let an \AVWC{} $(\mathcal{W},\mathcal{V})$ be given. If there exists a best channel to the eavesdropper and if the jammer does not possess non-causal side information, or if there exists a best channel to the eavesdropper and the jammer possesses non-causal side information of the messages, then the \CR{} assisted code secrecy capacity under the maximum error criterion is given by
\begin{align}
\widehat{C}_S^{\text{ran}}(\mathcal{W},\mathcal{V}) & =\max_{p_\Psi , \rho(X|\Psi)}\left(\min_{W\in\widehat{\mathcal{W}}}I(p_{\Psi};\rho W) - \max_{V\in\widehat{\mathcal{V}}}I(p_{\Psi};\rho V)\right),
\end{align}
%
If the \AVWC{} is additionally strongly degraded, then the \CR{} assisted code secrecy capacity under the maximum error criterion simplifies to
\begin{align}
\widehat{C}_S^{\text{ran}}(\mathcal{W},\mathcal{V}) & = \max_{p_X}\left(\min_{W\in \widehat{\mathcal{W}}} I(p_X;W) - \max_{V\in \widehat{\mathcal{V}}}I(p_X;V)\right).\label{eq:theonon}
\end{align}
%
\end{cor}
\begin{proof}[Proof of Corollary \ref{th:noknowledge-or-only-the-messages}]
By simple modifications in Lemma \ref{cor:boundprefix}, as well as in the converses, it is easy to see that the theorem holds. 
\end{proof}
The secrecy capacity $\widehat{C}_S^{\text{ran}}(\mathcal{W},\mathcal{V})$ depends on the convex closures $\widehat{\mathcal{W}}$ and $\widehat{\mathcal{V}}$.
\begin{remark}[Input and State Constraints]

The extension of the results to the case of input and state constraints is not straight forward. 
While the modifications in the sense of \cite{Goldfeld2016} might be possible and may lead to a single letter random code secrecy capacity, the restrictions on the jammer's strategy are very strict. 
In \cite{Goldfeld2016}, the jammer is restricted to a type constrained jamming strategy. 
In \cite{Chen2021}, the authors considered deterministic wiretap codes for the \AVWC{} with input and state peak constraints.  
They provided a single letter formula for achievable secrecy rates. 
The converse for the general case is still open. 
In \cite{Janda2015}, a general multi letter formula for the achievable random code secrecy rate with input and state peak constraints is presented. 
The converse for the general case remains an open problem.
\end{remark}

\begin{remark}[From Random to Deterministic - Not Elimination]
In \cite{Ahlswede1978}, Ahlswede proposes the Elimination of Correlation technique to reduce the amount of \CR{} to only $n^2$. 
He then uses a prefix code to inform the receiver which realization of the randomness is used. 
This leads to the following dichotomy result: 
The deterministic code capacity (under the average error criterion) equals its random code capacity, or is equal to zero if the \AVC{} is symmetrizable. 
Note that this technique cannot be used in our system model. 
If a prefix code were used to inform the receiver which deterministic code is used, the jammer  would obtain this information as well and we obtain once again the situation of the maximal error criterion for deterministic codes.

The authors of \cite{Boche2019a} present a technique to reduce the amount of \CR{} to only a polynomial order. 

The authors draw  codewords not from the complete set of typical sequences, but from a "suitable" subset. 
The reduction of the amount of \CR{} is meaningful, since in practical implementations \CR{} might be expensive, or just not available. 
Hence, from a system design point of view, it makes sense to reduce the necessary amount of \CR{}. 
In \cite{Sarwate2007}, the authors provide an upper bound on the amount of \CR{} which corresponds to $\mathcal{O}(\log n)$, even when the jammer knows the codewords non-causally, but without secrecy constraints. 
However, deriving deterministic code results or the minimal amount of \CR{} is not the intention of this work. 
Instead, we assume that there exists a sufficient amount of \CR{} to compute fundamental results on the secrecy capacities for different knowledge-scenarios at the jammer.

\end{remark}
\section{Discussion}\label{sec:Discuss}
\subsection{Relation to the secrecy capacity under average error criterion}
In the following, we provide the secrecy capacity formula under the average error criterion, and set the capacity formulas into relation to each other.
\begin{cor}[Common Randomness Assisted Secrecy Capacity under the Average Error Criterion if the Family of Channels to the Illegitimate Receiver is Strongly Degraded, Strongly Noisier, or Strongly Less Capable with Independent States]\label{cor:1}
If for an \AVWC{} the family of channels to the illegitimate receiver $\mathcal{V}$ is strongly degraded, strongly noisier, or strongly less capable with independent states, then the \CR{} assisted secrecy capacity under the average error criterion for the standard \AVWC{} is given by 
\begin{align*}
\widehat{C}_{S,av}^{\text{ran}}(\mathcal{W},\mathcal{V})& = \max_{p_X}\bigg(\min_{W\in \widehat{\mathcal{W}}} I(p_X;W) - \max_{V\in \widehat{\mathcal{V}}}I(p_X;V)\bigg).
\end{align*}
\end{cor}
\begin{cor}
Let an \AVWC{} $(\mathcal{W},\mathcal{V})$ be given. If there exists a best channel to the eavesdropper, then 
\begin{align}
\widehat{C}_{S,av}^{\text{ran}}(\mathcal{W},\mathcal{V})& = \widehat{C}_S^{\text{ran}}(\mathcal{W},\mathcal{V}) \geq \widehat{\widehat{C}}_S^{\text{ran}}(\mathcal{W},\mathcal{V}). 
\end{align}
where $\widehat{C}_{S,av}^{\text{ran}}(\mathcal{W},\mathcal{V})$ denotes the \CR{} assisted code secrecy capacity under the average error criterion.
\end{cor}
\begin{proof}
It is easy to see that $\widehat{\mathcal{W}}  \subset \widehat{\widehat{\mathcal{W}}}$ and $\widehat{\mathcal{V}}  \subset \widehat{\widehat{\mathcal{V}}}$.
\end{proof}
\subsection{Example}
To clarify the fundamental difference between the capacity formulas mentioned above, and to show that the inclusion can be strict, we provide an explicit example. 
First, we define $\mathcal{I}_{(\cdot)}(\cdot)$ as the convex hull of the row of channel matrices as follows.
\begin{defi}[\cite{Ahlswede1978}]
For a given $x\in\mathcal{X}$, let $\mathcal{I}_w(x)$ denote the convex hull of the set $\{W(\cdot|x,s):s\in\mathcal{S}\}$ of probability distributions on $\mathcal{Y}$, i.e., $\mathcal{I}_w(x)= \text{conv}\left(W(\cdot|x,s):s\in\mathcal{S}\right)$.
\end{defi}
\begin{ex}
~
\begin{figure}[h!]
\begin{center}
\resizebox{!}{6cm}{
\begin{tikzpicture}

\definecolor{color0}{rgb}{0.75,0.75,0}

\begin{axis}[
xlabel={$\delta_1$},
ylabel={$\delta_2$},
xmin=0, xmax=1,
ymin=0, ymax=1,
tick align=outside,
tick pos=both,
xmajorgrids,
x grid style={white},
ymajorgrids,
y grid style={white},
axis line style={white},
axis background/.style={fill=white!89.803921568627459!black}
]
\addplot [very thick, blue, dashed, forget plot]
table {%
0.1 0.9
0.2 0.8
};
\addplot [very thick, black, dashed, forget plot]
table {%
0.7 0.3
0.85 0.15
};
\addplot [very thick, red, dashed, forget plot]
table {%
0.8 0.2
0.9 0.1
};
\addplot [very thick, blue, forget plot]
table {%
0.25 0.75
0.3 0.7
};
\addplot [very thick, black, forget plot]
table {%
0.4 0.6
0.45 0.55
};
\addplot [very thick, red, forget plot]
table {%
0.6 0.4
0.65 0.35
};
\node at (axis cs:0.15,0.85)[
  anchor=base west,
  text=black,
  rotate=0.0
]{ $I_w(x_1)$};
\node at (axis cs:0.775,0.225)[
  anchor=base west,
  text=black,
  rotate=0.0
]{ $I_w(x_2)$};
\node at (axis cs:0.75,0.05)[
  anchor=base west,
  text=black,
  rotate=0.0
]{ $I_w(x_3)$};
\node at (axis cs:0.275,0.725)[
  anchor=base west,
  text=black,
  rotate=0.0
]{ $I_v(x_1)$};
\node at (axis cs:0.425,0.575)[
  anchor=base west,
  text=black,
  rotate=0.0
]{ $I_v(x_2)$};
\node at (axis cs:0.625,0.375)[
  anchor=base west,
  text=black,
  rotate=0.0
]{ $I_v(x_3)$};
\end{axis}

\end{tikzpicture}
}
\caption{Difference of capacities if the channel input is known or unknown at the jammer.}
\end{center}
\end{figure}
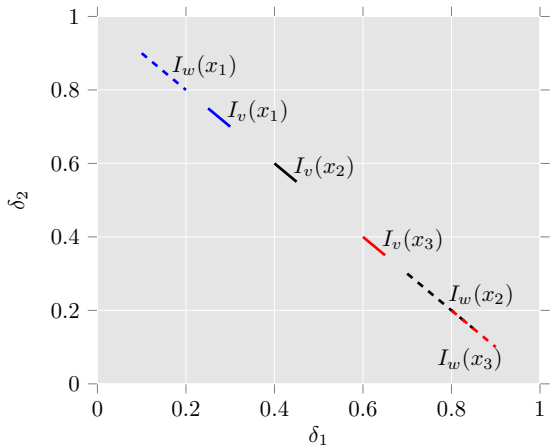
We consider the following example. 
Let the channel matrices be given as follows.
\begin{align*}
w(\cdot|\cdot,s_1) & =\begin{pmatrix}
0.1 & 0.9\\
0.7 & 0.3\\
0.8 & 0.2
\end{pmatrix}, 
& w(\cdot|\cdot,s_2) & =\begin{pmatrix}
0.2 & 0.8\\
0.85 & 0.15\\
0.9 & 0.1
\end{pmatrix}\\
v(\cdot|\cdot,s_1) & =\begin{pmatrix}
0.25 & 0.75\\
0.4 & 0.6\\
0.6 & 0.4
\end{pmatrix} ,
& v(\cdot|\cdot,s_2) & =\begin{pmatrix}
0.3 & 0.7\\
0.45 & 0.55\\
0.65 & 0.35
\end{pmatrix}
\end{align*}
It is easy to see that this channel \AVWC{} fulfills the strongly less capable property. 
We have
\begin{align*}
\widehat{W} & = \alpha w(\cdot|\cdot,s_1) + (1-\alpha)w(\cdot|\cdot,s_2) &= \begin{pmatrix}
		0.2 - 0.1 \alpha & 0.8 + 0.1 \alpha\\
		0.85 - 0.15 \alpha & 0.15 + 0.15 \alpha\\
		0.9 - 0.1 \alpha & 0.1 + 0.1 \alpha
		\end{pmatrix},\\
\widehat{V} & = \beta v(\cdot|\cdot,s_1) + (1-\beta) & = \begin{pmatrix}
		0.3 - 0.05 \beta & 0.7 + 0.05 \beta\\
		0.45 - 0.05 \beta & 0.55 + 0.05 \beta\\
		0.65 - 0.05 \beta & 0.35 + 0.05 \beta\\
		\end{pmatrix}.
\end{align*}
The secrecy capacity $\widehat{C}_{S}^{\text{ran}}(\mathcal{W},\mathcal{V})$ of this \AVWC{} can be calculated to $\widehat{C}_{S}^{\text{ran}}(\mathcal{W},\mathcal{V}) \approx 0.3$ bits per channel use, $p_X(0)=p_X(2)=0.5$, $p_X(1)=0$, $\alpha=0.5$, $\beta \approx 1$.
In contrast to that, one can easily see that the channels 
\begin{align*}
\widehat{\widehat{W}} & = \begin{pmatrix}
0.2 & 0.8 \\
0.8 & 0.2 \\
0.8 & 0.2
\end{pmatrix} & \widehat{\widehat{V}} & = \begin{pmatrix}
0.25 & 0.75 \\
0.4 & 0.6 \\
0.65 & 0.35
\end{pmatrix}
\end{align*}
correspond to the worst and the best channels to Bob and Eve, respectively, if the channel input is non-causally known at the jammer. 
In this case, the secrecy capacity for the \AVWC{} can be calculated to $\widehat{\widehat{C}}_S^{\text{ran}}(\mathcal{W},\mathcal{V}) \approx 0.26$ bits per channel use (which is strictly smaller than $\widehat{C}_{S}^{\text{ran}}(\mathcal{W},\mathcal{V})$), with input distribution $p_X(0)=p_X(1)=0.5$, $p_X(2)=0$. 
The second input symbol is used for the case with non-causal side information at the jammer instead of the third one as for the \AVWC{} without side information.
\end{ex}
\subsection{Summary}
In this work, we derive a single letter formula for the random code secrecy capacity under the maximum error criterion for an active attacker with non-causal side information of the codewords, provided there exists a best channel to the eavesdropper. 
Additionally, we provide a formula for the random code secrecy capacity for the case that the eavesdropping channel is strongly degraded, strongly noisier, or strongly less capable with respect to the main channel. 
We further allow that the messages might also be known at the jammer.
We apply and extend methods of \cite{Boche2019a} and \cite{Wiese2016}. 
We show that the derived secrecy capacities depend on the row convex closures of the sets of channels to Bob and Eve for the general and the strongly degraded cases, respectively, if the input is non-causally known at the jammer and depend on the convex closures of the sets of channels if the channel input is not non-causally known at the jammer. 

We compare our results to the random code secrecy capacity for the cases of maximum error criterion but no non-causal side information at the jammer, maximum error criterion with non-causal side information of the messages at the jammer, and the standard \AVWC . 
In the considered system model, the worst case occurs if the codewords (channel inputs) are non-causally known at the jammer. 
As we have shown, it does not matter if the jammer additionally knows the messages. 
The random code secrecy capacity is determined with respect to the row convex closures of the channel sets. 
In contrast, if the jammer does not know the channel input non-causally, then for the cases of maximum error criterion but without non-causal side information at the jammer, maximum error criterion with non-causal side information of the messages at the jammer, and the case of average error criterion without non-causal side information at the jammer, the random code secrecy capacity is determined with respect to the convex closure of the channel sets. 
We provided an example to illustrate this fundamental difference. 
It is quite obvious that optimizing over a larger set, here the row convex closure compared to the convex closure of the channel sets, may lead to a smaller random code secrecy capacity. 
%
%

From a resource theory point of view, the necessary amount of \CR{} is of interest. 
We do not upper bound the amount of \CR{}. 
To ensure that codewords occur in sufficiently many codebooks in order to confuse the jammer, we give a lower bound on the amount of \CR{}. 
This \CR{} is known at the eavesdropper, and hence cannot be used as key to achieve a secure transmission. 
Secrecy is achieved by wiretap coding.
%
%
%
\appendices
\section{Exchangeability of maximization orders}
\begin{lem}\label{lem:maxeqmax}
Let the sequence $(a_{i,j})_{\substack{i\in \mathcal{A},\\j \in \mathcal{B}}}$, $a_{i,j}\in \mathbb{R}$ be given, where $\mathcal{A},\mathcal{B}\subset \mathbb{N}$ are finite sets. 
Then 
\begin{align*}
\max_{i\in \mathcal{A}} \max_{j\in \mathcal{B}} (a_{i,j})_{\substack{i \in \mathcal{A},\\j \in \mathcal{B}}} & = \max_{j\in \mathcal{B}}\max_{i\in \mathcal{A}}  (a_{i,j})_{\substack{i \in \mathcal{A},\\j \in \mathcal{B}}}.
\end{align*}
\end{lem}
\begin{proof}
Let $\mathcal{J}^\ast$ and $\mathcal{I}^\ast$ be given as
\begin{align*}
\mathcal{J}^\ast & = \left\{\max_{i\in\mathcal{A}}(a_{i,j}):~j\in \mathcal{B}\right\}\\
\mathcal{I}^\ast & = \left\{\max_{j\in\mathcal{B}}(a_{i,j}):~i\in \mathcal{A}\right\}
\end{align*}
Then it is easy to see that
\begin{align*}
\max_{j\in\mathcal{B}}J^\ast  &= \max_{i\in\mathcal{A}}I^\ast
\end{align*}
Intuitively, the result follows when imagining a matrix. 
If the global maximum is unique, then the operations of collecting the maximum in each column in the set $\mathcal{J}^\ast$ and then taking the maximal element of $\mathcal{J}^\ast$ is equivalent to collecting the maximum in each row in the set $\mathcal{I}^\ast$ and then taking the maximal element of $\mathcal{I}^\ast$. 

If the global maximum is not unique, the result remains the same, but the indices $(i,j)\in \mathcal{A}\times\mathcal{B}$ might change. 
\end{proof}
\section{Variation Distance, Markov, Chernoff, and Chernoff-Hoeffding Bounds}
\begin{defi}[Variation Distance]
The variation distance of two distributions $P_1,P_2$ on $\mathcal{X}$ is defined as 
\begin{align}
||P_1-P_2||_V & = \sum_{x\in\mathcal{X}}|P_1(x)-P_2(x)|.
\end{align}
\end{defi}
\begin{lem}[{\cite[Lemma 2.7]{Csiszar2011}}]\label{lem:csiz}
If $||P_1-P_2||_V = \tau \leq \frac{1}{2}$, then
\begin{align*}
|H(P_1)-H(P_2)| & \leq - \tau \log{\frac{\tau}{|\mathcal{X}|}}.
\end{align*}
\end{lem}
We give a reminder on Markov's inequality.
\begin{lem}[Markov's Inequality {\cite[Lemma 83]{Ahlswede2014}}]\label{lem:markov}
Let $X$ be a \RV{} with mean $E[X]=\mu$ and let $a$ be a positive number. 
Then
\begin{align*}
Pr\{X\geq a\} & \leq \frac{\mu}{a}.
\end{align*}
\end{lem}
Chernoff bounds are given as follows. 
\begin{lem}[Chernoff bounds, \cite{Cai2013}, {\cite[Lemma 2]{Boche2019a}}]\label{lem:chernoffbounds}
Let $X_1,X_2,...,X_n$ be \IID{} \RVs{} with values in $\{0,1\}$, with $Pr\{X_i=1\}=p$. 
For all $\epsilon\in (0,1)$ and $p_0<p<p_1$, the following bounds hold 
\begin{align}
Pr\left\{\frac{1}{n}\sum_{i=1}^n X_i > (1+\epsilon)p_1\right\} & < \exp_e\left\{-\frac{\epsilon^2}{8}np_1\right\},\\
Pr\left\{\frac{1}{n}\sum_{i=1}^n X_i < (1-\epsilon)p_0\right\} & < \exp_e\left\{-\frac{3\epsilon^2}{8}np_0\right\}.
\end{align}
\end{lem}
The Chernoff-Hoeffding bound is widely used in the proof. 
Therefore, it shall be stated here.
\begin{lem}[Chernoff-Hoeffding bounds, {\cite[Theorem 1.1]{Dubhashi2009},\cite{Ahlswede2002}}]\label{lem:chernoff}
Let $X_1,X_2,...,X_n$ be \IID{} \RVs{} with values in $[0,b]$, where $b$ is a positive number. 
Further, let $E[X_i]=\mu$, and $0<\epsilon<\frac{1}{2}$. Then
\begin{align}
Pr\left\{\frac{1}{n}\sum_{i=1}^n X_i \not \in [(1\pm\epsilon)\mu]\right\} & \leq 2 \exp_e{\left(-n\frac{\epsilon^2\mu}{3b}\right)}\label{eq:Chernoff},
\end{align}
where $[(1\pm\epsilon)\mu]$ means the interval $[(1-\epsilon)\mu,(1+\epsilon)\mu ]$. 
\end{lem}

\section{Typical Sets}
We summarize some known facts of typicality properties. 
Let $\delta > 0$.
\begin{lem}[Properties of typical sets I, {\cite[Lemma 2.13, Problem 2.5]{Csiszar2011}}]\label{lem:typsetI}
Let $x^n \in \mathcal{T}^n_{p,\delta}$. Then for any $W:\mathcal{X}\to \mathcal{P(Y)}$
\begin{align*}
|\mathcal{T}^n_{pW,2|\mathcal{X}|\delta}| & \leq \exp\{n(H(pW)+f_1(\delta))\},\\
W^n(y^n|x^n) & \leq \exp\{-n(H(W|p)-f_2(\delta))\}\quad \forall y^n \in \mathcal{T}^n_{W,\delta}(x^n),
\end{align*}
for some functions $f_1(\delta),f_2(\delta)>0$ with $\lim_{\delta\to 0}f_1(\delta)=0$ and $\lim_{\delta\to 0}f_2(\delta)=0$.
\end{lem}
\begin{lem}[Properties of typical sets II, {\cite[Lemma III.1.3]{Shields1996}}]\label{lem:typsetII}
For every $p\in\mathcal{P(X)}$, $W:\mathcal{X}\to \mathcal{P(Y)}$ and $x^n\in \mathcal{X}^n$ 
\begin{align*}
p^n(\mathcal{T}^n_{p,\delta}) & \geq 1 - (n+1)^{|\mathcal{X}|} \exp\{-nc\delta^2\},\\
W^n(\mathcal{T}^n_{W,\delta}(x^n)|x^n)& \geq 1 - (n+1)^{|\mathcal{X}||\mathcal{Y}|} \exp\{-nc\delta^2\}. 
\end{align*}
with $c=\frac{1}{2\ln{2}}$. 
Furthermore, there exists an $n_0$ and a $c^\prime>0$, depending on $|\mathcal{X}|,|\mathcal{Y}|$ and $\delta$, such that for all $n>n_0$ for each $p\in\mathcal{P(X)}$ and $W:\mathcal{X}\to \mathcal{P(Y)}$
\begin{align}
p^n(\mathcal{T}^n_{p,\delta}) & \geq 1 - \exp\{-nc^\prime \delta^2\},\\
W^n(\mathcal{T}^n_{W,\delta}(x^n)|x^n)& \geq 1 - \exp\{-nc^\prime \delta^2\}.
\end{align}
\end{lem}
\begin{lem}[Properties of typical sets III, {\cite[Lemma 2.2]{Csiszar2011}}]\label{lem:typsetIII}
Let $\mathcal{P}_0^n(\mathcal{S})$ be the set of all possible types of $n$-length sequences on $\mathcal{S}^n$. 
The cardinality of the set of all possible types of length $n$ is upper bounded by
\begin{align*}
|\mathcal{P}_0^n(\mathcal{S})| & \leq (n+1)^{|\mathcal{S}|}.
\end{align*}
\end{lem}
\begin{lem}[Properties of typical sets IV, {\cite[Lemma 3]{Wyrembelski2010a}}{\cite[Lemma 3]{Bjelakovic2013a}}]\label{lem:typsetIV}
Assume, the distributions $p,\overbar{p} \in \mathcal{P(X)}$ and the two matrices $W,\overbar{W}:\mathcal{X} \to \mathcal{P(Y)}$ are given. For any positive integer $n$ and sufficiently small $\delta>0$,
\begin{align*}
(pW)^n(\mathcal{T}^n_{\overbar{W},\delta}(\overbar{x}^n)) & \leq (n+1)^{|\mathcal{X}||\mathcal{Y}|}\exp\{-n(I(\overbar{p};\overbar{W})-f_3(\delta))\},
\end{align*}
for all $\overbar{x}^n \in \mathcal{T}^n_{\overbar{p},\delta}$ holds, with some $f_3(\delta)>0$ and $\lim_{\delta \to 0}f_3(\delta)=0$. 
Furthermore, there exist an $n_0$ and a $\nu>0$, depending on $|\mathcal{X}|,|\mathcal{Y}|$ and $\delta$, such that for all $n>n_0$,
\begin{align}
(pW)^n(\mathcal{T}^n_{\overbar{W},\delta}(\overbar{x}^n)) & \leq \exp\{-n(I(\overbar{p};\overbar{W}) - \nu))\}.
\end{align}
\end{lem}
\begin{lem}[Properties of typical sets V, {\cite[Lemma 2]{Sarwate2010}}]\label{lem:typicalsetsV}
Let the sequences $x^n\in\mathcal{X}^n$, $s^n\in\mathcal{S}^n$, and  $\delta,\hat{\delta}>0$ be given. 
Further, let $(\Psi,X)$ be distributed according to $p_{\Psi,X}=p_{\Psi}\rho_{X|\Psi}$. 
Define the channel 
\begin{align*}
\theta (s|x) & := \frac{1}{N(x|x^n)}\sum_{i=1}^n\mathds{1}(s_i=s,x_i=x).
\end{align*}
Then,
\begin{align}
Pr\left\{(\Psi^n,x^n,s^n)\notin \mathcal{T}_{p_{\Psi}\times\rho_{X|\Psi}\times\theta,\delta}^n|(\Psi^n,x^n)\in \mathcal{T}_{p_{\Psi}\times\rho_{X|\Psi},\hat{\delta}}^n\right\} & \leq \exp\{-n h(\delta)\},
\end{align}
where $h(\delta)\to 0$ as $\delta\to 0$.
\begin{proof}
Follows for example by {\cite[Lemma 2.10, Lemma 2.12]{Csiszar2011}}.
\end{proof}
\end{lem}
\begin{lem}\label{lem:existenceofaQ}
Let $(\Psi^n,X^n,S^n)\in \mathit{\Psi}^n\times\mathcal{X}^n\times\mathcal{S}^n$ be distributed according to $p_{\Psi}^n\rho_{X|\Psi}^np_{S^n|\Psi^n,X^n}$. 
Let $A$ be defined as the following event. 
\begin{align*}
A &: = \{\nexists \underline{\theta} \in \mathcal{P}_{0}(\mathcal{S}^n|\mathcal{X}^n) :\quad (\Psi^n,X^n,S^n)\in \mathcal{T}_{p_{\Psi}\times\rho_{X|\Psi}\times\underline{\theta},\delta}^n\}.
\end{align*}
Then 
\begin{align*}
Pr\{A\} & \leq \exp\{-n c^\prime \hat{\delta}^2\} + (n+1)^{|\mathcal{X}||\mathcal{S}|}\exp\left\{-n \min_{\theta \in \mathcal{P}_{0}(\mathcal{S}^n|\mathcal{X}^n)}h_{\theta}(\delta)\right\},
\end{align*}
where $h(\delta)\to 0$ if $\delta \to 0$.
\end{lem}
\begin{proof}
\begin{align*}
Pr\{A\} & = Pr\left\{(\Psi^n,X^n)\notin \mathcal{T}_{p_{\Psi}\times\rho_{X|\Psi},\hat{\delta}}^n\right\}Pr\left\{A|(\Psi^n,X^n)\notin \mathcal{T}_{p_{\Psi}\times\rho_{X|\Psi},\hat{\delta}}^n\right\}\\
& \quad + Pr\left\{(\Psi^n,X^n)\in \mathcal{T}_{p_{\Psi}\times\rho_{X|\Psi},\hat{\delta}}^n\right\}Pr\left\{A|(\Psi^n,X^n)\in \mathcal{T}_{p_{\Psi}\times\rho_{X|\Psi},\hat{\delta}}^n\right\}\\
&\overset{(a)}{\leq} Pr\left\{(\Psi^n,X^n)\notin \mathcal{T}_{p_{\Psi}\times\rho_{X|\Psi},\hat{\delta}}^n\right\}\\
&\qquad + \smashoperator[lr]{\sum_{(x^n,s^n)\in \mathcal{X}^n\times\mathcal{S}^n}} p(x^n,s^n) Pr\left\{(\Psi^n,x^n,s^n)\notin \mathcal{T}_{p_{\Psi}\times\rho_{X|\Psi}\times\underline{\theta},\delta}^n|(\Psi^n,x^n)\in \mathcal{T}_{p_{\Psi}\times\rho_{X|\Psi},\hat{\delta}}^n\right\}\\
&\overset{(b)}{\leq} \exp\{-n c^\prime \hat{\delta}^2\}+ \sum_{\substack{\hat{p}\in\mathcal{P}_0(\mathcal{X}^n)\\ \underline{\theta}\in\mathcal{P}_0(\mathcal{S}^n|\mathcal{X}^n)\\(x^n,s^n)\in \mathcal{T}_{\hat{p}\times \underline{\theta}}^n}}\hat{p}(x^n)\underline{\theta}(s^n|x^n)Pr\{(\Psi^n,x^n,s^n)\notin \mathcal{T}_{p_{\Psi}\times\rho_{X|\Psi}\times\underline{\theta},\delta}^n|(\Psi^n,x^n)\in \mathcal{T}_{p_{\Psi}\times\rho_{X|\Psi},\hat{\delta}}^n\}\\
&\overset{(c)}{\leq} \exp\{-n c^\prime \hat{\delta}^2\}+ \sum_{\substack{\hat{p}\in\mathcal{P}_0(\mathcal{X}^n)\\ \underline{\theta}\in\mathcal{P}_0(\mathcal{S}^n|\mathcal{X}^n)}}\exp\{-n h_{\underline{\theta}}(\delta)\}\\
&\overset{(d)}{\leq} \exp\{-n c^\prime \hat{\delta}^2\}+ (n+1)^{|\mathcal{X}||\mathcal{S}|}\exp\left\{-n \min_{\theta \in \mathcal{P}_{0}(\mathcal{S}^n|\mathcal{X}^n)}h_{\theta}(\delta)\right\}.
\end{align*} 
Here, $(a)$ follows by upper bounding $Pr\left\{A|(\Psi^n,X^n)\notin \mathcal{T}_{p_{\Psi}\times\rho_{X|\Psi},\hat{\delta}}^n\right\}$ and $Pr\left\{(\Psi^n,X^n)\in \mathcal{T}_{p_{\Psi}\times\rho_{X|\Psi},\hat{\delta}}^n\right\}$ by $1$. 
Note that $\underline{\theta}$ in $(a)$ is dependent on the sequences $(x^n, s^n)$ according to Lemma \ref{lem:typicalsetsV}, and hence different for different conditional types of $(x^n, s^n)$.
$(b)$ follows because of Lemma \ref{lem:typsetII}, $(c)$ follows because of Lemma \ref{lem:typicalsetsV}, and $(d)$ follows by type counting.
\end{proof}
\section{}\label{proof:cor-boundprefix}
\begin{lem}\label{cor:boundprefix}
For any conditional type $\underline{\theta}\in\mathcal{P}_{0}(\mathcal{S}^n|\mathcal{X}^n)$, define the probability measure $p_{\Psi}\times \rho \times \underline{\theta}$ as 
\begin{align*}
(p_{\Psi}\times \rho \times \underline{\theta})(\psi , x,s) & = p_{\Psi}(\psi) \rho (x|\psi)\underline{\theta}(s|x). 
\end{align*}
Let $\delta>0$ and let $p_{\overline{\Psi X S}}$ be a type fulfilling $p_{\overline{\Psi}}=p_{\Psi}$ and
\begin{align}
||p_{\Psi X S} - p_{\overline{\Psi X S}}||_V &\leq \delta. \label{eq:boundprefixvardistance}
\end{align} 
Moreover, let $\Psi^{\prime n}$ be uniformly distributed on $\mathcal{T}_{p_{\Psi}}^n$. 
Then there exist an $n_0$ and a $\nu$, depending on $|\mathcal{X}|,|\mathcal{Y}|$,$|\mathit{\Psi}|,|\mathcal{S}|$ and $\delta$, such that for all $n>n_0$ we have for any $(x^n,s^n)\in\mathcal{T}_{p_{\overline{X S}}}^n$,
\begin{align*}
E\left[W^{n}\left(\left(\smashoperator[r]{\bigcup_{\underline{\theta}\in \mathcal{P}_0(\mathcal{S}^{n}|\mathcal{X}^n)}}\mathcal{T}^n_{\rho  W_{\underline{\theta}},\delta}(\Psi^{\prime n})\right)\bigg\vert x^{n},s^{n}\right)\right] & \leq \exp\left\{-n\left(\min_{\underline{\theta}\in \mathcal{P}_0(\mathcal{S}^{n}|\mathcal{X}^n)}I(p_{\Psi};\rho W_{\underline{\theta}})- \nu \right)  \right\}\\
& \leq \exp\left\{-n\left(\min_{\underline{\theta}\in \mathcal{P}(\mathcal{S}|\mathcal{X})}I(p_{\Psi};\rho W_{\underline{\theta}})- \nu \right)  \right\}.
\end{align*} 
\end{lem}

\begin{proof}[Proof of Lemma \ref{cor:boundprefix}]
We divide the proof into two steps. 
First we provide an upper bound, and show then secondly that this upper bound holds for arbitrary sequences of the same type.

Let $(\Psi^n,X^{n},S^{n})$ be uniformly distributed according to $ p_{\Psi}^n\times \rho^n \times \theta^n$ and independent of $\Psi^{\prime n}$. 
First, we have 
\begin{align*}
&E\left[W^{n}\left(\left(\smashoperator[r]{\bigcup_{\underline{\theta}\in \mathcal{P}_0(\mathcal{S}^{n}|\mathcal{X}^n)}}\mathcal{T}^n_{\rho W_{\underline{\theta}},\delta}(\Psi^{\prime n})\right)\bigg\vert X^{n},S^{n}\right)\right]\\
& \overset{(a)}{\leq} \sum_{\underline{\theta}\in \mathcal{P}_0(\mathcal{S}^{n}|\mathcal{X}^n)}E\left[W^{n}\left(\mathcal{T}^n_{\rho W_{\underline{\theta}},\delta}(\Psi^{\prime n})\bigg\vert X^{n},S^{n}\right)\right]\\
& \overset{(b)}{=} \sum_{\underline{\theta}\in \mathcal{P}_0(\mathcal{S}^{n}|\mathcal{X}^n)}\sum_{\psi^{\prime n} \in \mathit{\Psi}^n}p_{\Psi^n}(\psi^{\prime n})\smashoperator[lr]{\sum_{(\psi^n x^{n} s^{n})\in \mathit{\Psi}^n\times\mathcal{X}^{n}\times \mathcal{S}^{n}}}p^n_{\Psi X S}(\psi^n, x^{n},s^{n})W^{n}\left(\mathcal{T}^n_{\rho W_{\underline{\theta}},\delta}\left(\Psi^{\prime n}\right)\bigg\vert x^{n},s^{n}\right)\\
& \overset{(c)}{=} \sum_{\underline{\theta}\in \mathcal{P}_0(\mathcal{S}^{n}|\mathcal{X}^n)}\sum_{\psi^{\prime n} \in \mathit{\Psi}^n}p_{\Psi^n}(\psi^{\prime n})\smashoperator[lr]{\sum_{(\psi^n x^{n} s^{n})\in \mathit{\Psi}^n\times\mathcal{X}^{n}\times \mathcal{S}^{n}}}p^n_{\Psi}(\psi^n)\rho^n(x^{n}|\psi^n)\theta^n(s^{n}|x^{n})W^{n}\left(\mathcal{T}^n_{\rho W_{\underline{\theta}},\delta}\left(\psi^{\prime n}\right)\bigg\vert x^{n},s^{n}\right)\\
& \overset{(d)}{=} \sum_{\underline{\theta}\in \mathcal{P}_0(\mathcal{S}^{n}|\mathcal{X}^n)}\sum_{\psi^{\prime n} \in\mathit{\Psi}^n}p_{\Psi^n}(\psi^{\prime n})(p_{\Psi} \rho W_{\theta})^{n}\left(\mathcal{T}^n_{\rho W_{\underline{\theta}},\delta}(\psi^{\prime n})\right)\\
& \overset{(e)}{\leq} \sum_{\underline{\theta}\in \mathcal{P}_0(\mathcal{S}^{n}|\mathcal{X}^n)}\exp{\left\{-n\left(I(p_\Psi;\rho W_{\underline{\theta}}) -\hat{\nu} \right)\right\}}\sum_{\psi^{\prime n} \in \mathit{\Psi}^n}p_{\Psi^n}(\psi^{\prime n})\\
& \overset{(f)}{\leq} (n+1)^{|\mathcal{X}||\mathcal{S}|}\exp{\left\{-n\left(\min_{\underline{\theta}\in \mathcal{P}_0(\mathcal{S}^{n}|\mathcal{X}^n)}I(p_\Psi;\rho W_{\underline{\theta}}) -\hat{\nu} \right)\right\}}\\
& \overset{(g)} {\leq}\exp{\left\{-n\left(\min_{\underline{\theta}\in \mathcal{P}_0(\mathcal{S}^{n}|\mathcal{X}^n)}I(p_\Psi;\rho W_{\underline{\theta}}) - \nu \right)\right\}}
\end{align*}
Here, $(a)$ follows by the union bound. $(b)$ follows by evaluating the expectation. $(c)$ follows by assumption that $(\Psi^n,X^{n},S^{n})$ is uniformly distributed according to $ p_{\Psi}^n\times \rho^n \times \underline{\theta}^n$ and independent of $\Psi^{\prime n}$. $(d)$ follows by expressing the probability function  ${\sum_{(\psi^n x^{n} s^{n})\in \mathit{\Psi}^n\times\mathcal{X}^{n}\times \mathcal{S}^{n}}}p^n_{\Psi}(\psi^n)\rho^n(x^{n}|\psi^n)\theta^n(s^{n}|x^{n})W^{n}\left(\left(\mathcal{T}^n_{\rho W_{\underline{\theta}},\delta}(\psi^{\prime n}\right)\bigg\vert x^{n},s^{n}\right)$ as the output probability function $(p_{\Psi} \rho W_{\theta})^{n}\left(\mathcal{T}^n_{\rho W_{\underline{\theta}},\delta}(\psi^{\prime n})\right)$. $(e)$ follows by Lemma \ref{lem:typsetIV}, $(f)$, and $(g)$ follow by Lemma \ref{lem:typsetIII}.

Next, assume that $(\Psi^n,X^{n},S^{n})$ is uniformly distributed on $\mathcal{T}_{p_{\overline{\Psi X S}}}^n$. 
We will show that the above inequality also holds in this case up to small terms. 
Due to (\ref{eq:boundprefixvardistance}) and Lemma \ref{lem:csiz}, we have 
\begin{align*}
H(p_{\overline{\Psi X S}}) & \geq H(p_{\Psi X S}) + \delta \log \frac{\delta}{|\mathit{\Psi}||\mathcal{X}||\mathcal{S}|}\\
& = : H(p_{\Psi X S}) + \delta^\prime.
\end{align*}
Furthermore, because of (\ref{eq:boundprefixvardistance}), we have $\mathcal{T}_{p_{\overline{\Psi X S}}}^n\subset \mathcal{T}_{p_{\Psi X S},\delta}^n$.
Hence, for any nonnegative function $f(\psi^n,x^n,s^n)$, we have
\begin{align*}
E[f(\Psi^n,X^n,S^x)] & =\sum_{(\psi^n,x^n,s^n)\in \mathcal{T}_{p_{\overline{\Psi X S}}}^n}p_{\overline{\Psi X S}}^n(\psi^n,x^n,s^n)f(\psi^n,x^n,s^n)\\
& = \frac{1}{|\mathcal{T}_{p_{\overline{\Psi X S}}}^n|}\sum_{(\psi^n,x^n,s^n)\in \mathcal{T}_{p_{\overline{\Psi X S}}}^n}f(\psi^n,x^n,s^n)\\
& \leq (n+1)^{|\mathit{\Psi}||\mathcal{X}||\mathcal{S}|} \exp\{-nH(p_{\overline{\Psi X S}})\}\sum_{(\psi^n,x^n,s^n)\in \mathcal{T}_{p_{\overline{\Psi X S}}}^n}f(\psi^n,x^n,s^n)\\
& \leq (n+1)^{|\mathit{\Psi}||\mathcal{X}||\mathcal{S}|} \exp\{-n(H(p_{\Psi X S})-\delta^\prime)\}\sum_{(\psi^n,x^n,s^n)\in \mathcal{T}_{p_{\Psi X S},\delta}^n}f(\psi^n,x^n,s^n)\\
& \leq (n+1)^{|\mathit{\Psi}||\mathcal{X}||\mathcal{S}|} \exp\{n\delta^{\prime\prime}\}\sum_{(\psi^n,x^n,s^n)\in \mathcal{T}_{p_{\Psi X S},\delta}^n}p^n_{\Psi}(\psi^n)\rho^n(x^{n}|\psi^n)\theta^n(s^{n}|x^{n})f(\psi^n,x^n,s^n)\\
& \leq (n+1)^{|\mathit{\Psi}||\mathcal{X}||\mathcal{S}|} \exp\{n\delta^{\prime\prime}\}\sum_{(\psi^n,x^n,s^n)\in \mathit{\Psi}^n\times\mathcal{X}^{n}\times \mathcal{S}^{n}}p^n_{\Psi}(\psi^n)\rho^n(x^{n}|\psi^n)\theta^n(s^{n}|x^{n})f(\psi^n,x^n,s^n).
\end{align*} 
With 
\begin{align*}
f(\psi^n,x^n,s^n) & = \sum_{\underline{\theta}\in \mathcal{P}_0(\mathcal{S}^{n}|\mathcal{X}^n)}\sum_{\psi^{\prime n} \in \mathit{\Psi}^n}p_{\Psi^n}(\psi^{\prime n})W^{n}\left(\left(\mathcal{T}^n_{\rho W_{\underline{\theta}},\delta}(\Psi^{\prime n}\right)\bigg\vert x^{n},s^{n}\right), 
\end{align*}
this shows
\begin{align*}
E\left[W^{n}\left(\left(\smashoperator[r]{\bigcup_{\underline{\theta}\in \mathcal{P}_0(\mathcal{S}^{n}|\mathcal{X}^n)}}\mathcal{T}^n_{\rho W_{\underline{\theta}},\delta}(\Psi^{\prime n})\right)\bigg\vert X^{n},S^{n}\right)\right] & \leq \exp\left\{-n\left(\min_{\underline{\theta}\in \mathcal{P}_0(\mathcal{S}^{n}|\mathcal{X}^n)}I(p_{\Psi};\rho W_{\underline{\theta}})- \nu \right)  \right\}.
\end{align*} 

Secondly, for an arbitrary permutation of the index set $\{1,2,...,n\}$ we have per definition
\begin{align*}
\pi\left(\mathcal{T}^n_{\rho W_{\underline{\theta}},\delta}(\psi^{\prime n})\right) & := \left\{\pi(y^n)\in \mathcal{Y}^n:\quad \left|\frac{1}{n}N(a,b|\psi^{\prime n},y^n) - \rho W_{\underline{\theta}}(b|a)\frac{1}{n}N(a|\psi^{\prime n})\right|\leq \delta , \forall a\in\mathit{\Psi}, b\in\mathcal{Y} \right\}\\
& = \left\{y^n\in \mathcal{Y}^n:\quad \left|\frac{1}{n}N(a,b|\psi^{\prime n},\pi^{-1}(y^n)) - \rho  W_{\underline{\theta}}(b|a)\frac{1}{n}N(a|\psi^{\prime n})\right|\leq \delta , \forall a\in\mathit{\Psi} ,b\in\mathcal{Y} \right\}\\
& = \left\{y^n\in \mathcal{Y}^n:\quad \left|\frac{1}{n}N(a,b|\pi(\psi^{\prime n}),y^n) - \rho  W_{\underline{\theta}}(b|a)\frac{1}{n}N(a|\pi(\psi^{\prime n}))\right|\leq \delta , \forall a\in\mathit{\Psi}, b\in\mathcal{Y} \right\}\\
& =:\mathcal{T}^n_{\rho W_{\underline{\theta}},\delta}\left(\pi(\psi^{\prime n})\right).
\end{align*}
Therefore, for a $(\tilde{x}^n,\tilde{s}^n)$ with $(\psi^n,\tilde{x}^n,\tilde{s}^n)\in \mathcal{T}_{p_{\Psi XS}}^n$ and an arbitrary permutation $\pi$, we have
\begin{align*}
E_{\Psi^{\prime n}}\left[W^n\left(\left(\smashoperator[r]{\bigcup_{\underline{\theta}\in \mathcal{P}_0(\mathcal{S}^{n}|\mathcal{X}^n)}}\mathcal{T}^n_{\rho W_{\underline{\theta}},\delta}(\Psi^{\prime n})\right)\bigg\vert \tilde{x}^n,\tilde{s}^n\right)\right] & =\sum_{\psi^{\prime n} \in \mathcal{T}_{P}^n}p_{\Psi^n}(\psi^{\prime n})W^n\left(\left(\smashoperator[r]{\bigcup_{\underline{\theta}\in \mathcal{P}_0(\mathcal{S}^{n}|\mathcal{X}^n)}}\mathcal{T}^n_{\rho W_{\underline{\theta}},\delta}(\Psi^{\prime n})\right)\bigg\vert \tilde{x}^n,\tilde{s}^n\right)\\
& = \sum_{\psi^{\prime n} \in \mathcal{T}_{p}^n}p_{\Psi^n}(\psi^{\prime n})W^n\left(\left(\smashoperator[r]{\bigcup_{\underline{\theta}\in \mathcal{P}_0(\mathcal{S}^{n}|\mathcal{X}^n)}}\pi\left(\mathcal{T}^n_{\rho W_{\underline{\theta}},\delta}(\psi^{\prime n})\right)\right)\bigg\vert \pi(\tilde{x}^n,\tilde{s}^n)\right)\\
& = \sum_{\psi^{\prime n} \in \mathcal{T}_{p}^n}p_{\Psi^n}(\psi^{\prime n})W^n\left(\left(\smashoperator[r]{\bigcup_{\underline{\theta}\in \mathcal{P}_0(\mathcal{S}^{n}|\mathcal{X}^n)}}\mathcal{T}^n_{\rho W_{\underline{\theta}},\delta}\left(\pi(\psi^{\prime n})\right)\right)\bigg\vert \pi(\tilde{x}^n,\tilde{s}^n)\right)\\
&\overset{(a)}{=} \sum_{\psi^{\prime n} \in \mathcal{T}_{p}^n}p_{\Psi^n}(\psi^{\prime n})W^n\left(\left(\smashoperator[r]{\bigcup_{\underline{\theta}\in \mathcal{P}_0(\mathcal{S}^{n}|\mathcal{X}^n)}}\mathcal{T}^n_{\rho W_{\underline{\theta}},\delta}\left(\psi^{\prime n}\right)\right)\bigg\vert \pi(\tilde{x}^n,\tilde{s}^n)\right)\\
& = E_{\Psi^{\prime n}}\left[W^n\left(\left(\smashoperator[r]{\bigcup_{\underline{\theta}\in \mathcal{P}_0(\mathcal{S}^{n}|\mathcal{X}^n)}}\mathcal{T}^n_{\rho W_{\underline{\theta}},\delta}(\Psi^{\prime n})\right)\bigg\vert \pi(\tilde{x}^n,\tilde{s}^n)\right)\right],
\end{align*}
where $(a)$ follows because we sum up over all $\psi^{\prime n}$ with the same type\footnote{Types are permutation invariant.} (hence, $p_{\Psi^n}(\psi^{\prime n})$ is identical for all $\psi^{\prime n}$ of the same type).

Hence, we can rewrite the expectation as 
\begin{align*}
& E\left[W^{n}\left(\left(\smashoperator[r]{\bigcup_{\underline{\theta}\in \mathcal{P}_0(\mathcal{S}^{n}|\mathcal{X}^n)}}\mathcal{T}^n_{\rho W_{\underline{\theta}},\delta}(\Psi^{\prime n})\right)\bigg\vert X^{n},S^{n}\right)\right]\\
& = \sum_{(\psi^n,x^n,s^n)\in \mathcal{T}_{p_{\Psi XS}}^n}p_{\Psi^nX^nS^n}(\psi^n, x^n,s^n)E_{\Psi^{\prime n}}\left[W\left(\left(\smashoperator[r]{\bigcup_{\underline{\theta}\in \mathcal{P}_0(\mathcal{S}^{n}|\mathcal{X}^n)}}\mathcal{T}^n_{\rho W_{\underline{\theta}},\delta}(\Psi^{\prime n})\right)\bigg\vert x^n,s^n\right)\right] \\
& = E\left[W\left(\left(\smashoperator[r]{\bigcup_{\underline{\theta}\in \mathcal{P}_0(\mathcal{S}^{n}|\mathcal{X}^n)}}\mathcal{T}^n_{\rho W_{\underline{\theta}},\delta}(\Psi^{\prime n})\right)\bigg\vert \tilde{x}^n,\tilde{s}^n\right)\right],
\end{align*}
for all $(\psi^n,\tilde{x}^n,\tilde{s}^n)\in \mathcal{T}_{p_{\Psi XS}}^n$.

\end{proof}
\section{Proof of Lemma \ref{rem:theFunc}}\label{proof:theFunc}
\begin{proof}[Proof of Lemma \ref{rem:theFunc}]
We consider both, the error probability and the information leakage. 
Let the maximum error probability and the information leakage, respectively, be given as 
\begin{align*}
&\hat{e}(\mathcal{K}_n) :=  \max_{\substack{f\in\mathcal{F}}}\max_{j\in\mathcal{J}_n} \sum_{x^n\in \mathcal{X}^n}E(x^n|j)W^n(\mathcal{D}_{j}^c|x^n,f(x^n)),\\
&\lim_{n\to \infty}\max_{\substack{f\in\mathcal{F}}}\max_{u\in \mathcal{U}_n}I(p_{J_n};E_uV_{f}^n) = 0
\end{align*}
Using the same $(n,J_n)$ deterministic wiretap code $\mathcal{K}_n$, fulfilling the above criteria  and considering now the maximization over $\theta \in \mathcal{P}(\mathcal{S}^n|\mathcal{X}^n)$ we can express the maximum error probability of transmitting one codeword as 
\begin{align*}
\max_{j\in\mathcal{J}_n}\sum_{x^n\in \mathcal{X}^n}E(x^n|j)W_{\theta}^n(\mathcal{D}_{j}^c|x^n)&=\max_{j\in\mathcal{J}_n}\sum_{x^n\in \mathcal{X}^n}\sum_{s^n\in \mathcal{S}^n}E(x^n|j)\theta(s^n|x^n)W^n(\mathcal{D}_{j}^c|x^n,s^n),
\end{align*}
and hence we have

\begin{align*}
\max_{j\in\mathcal{J}_n}\sum_{x^n\in \mathcal{X}^n}\sum_{s^n\in \mathcal{S}^n}E(x^n|j)\theta(s^n|x^n)W^n(\mathcal{D}_{j}^c|x^n,s^n)&\leq \max_{\substack{f\in\mathcal{F}}}\max_{j\in\mathcal{J}_n}\sum_{x^n\in \mathcal{X}^n}E(x^n|j)W^n(\mathcal{D}_{j}^c|x^n,f(x^n))\\
&\leq\max_{\theta\in \mathcal{P}(\mathcal{S}^n|\mathcal{X}^n)}\max_{j\in\mathcal{J}_n} \sum_{x^n\in \mathcal{X}^n}E(x^n|j)W_{\theta}^n(\mathcal{D}_{j}^c|x^n)\\
&=\hat{e}(\mathcal{K}_n)
\end{align*}
Since the mutual information is convex (row convex) with respect to the channel for fixed input/ input distribution, the optimal jamming strategy with respect to the reliability constraint is achieved at the boundary of the probability polytope,i.e., is deterministic, {\cite[Proposition 2.4.1]{Bertsekas2009}}. 
Hence, even though the set of stochastic jamming strategies is larger than the set of deterministic jamming strategies, both will lead to the same error expression.

Since 
\begin{align*}
E_uV_{f}^n  &=\sum_{x^n\in \mathcal{X}^n}E_u(x^n|j)V^n(z^n|x^n,f(x^n)),\\
V_{\theta}^n  &= \sum_{s^n\in \mathcal{S}^n}\theta(s^n|x^n)V^n(z^n|x^n,s^n),\\
E_uV_{\theta}^n  &=\sum_{x^n\in \mathcal{X}^n}E_u(x^n|j)\sum_{s^n\in \mathcal{S}^n}\theta(s^n|x^n)V^n(z^n|x^n,s^n), 
\end{align*}
for the leakage we can show that 
\begin{align*}
\max_{\substack{f\in\mathcal{F}}}I(p_{J_n};E_uV_{f}^n) &= \max_{\theta\in \mathcal{P}(\mathcal{S}^n|\mathcal{X}^n)}I(p_{J_n};E_uV_{\theta^n}^n)
\end{align*}

because the mutual information is convex in $V^n(z^n|x^n,s^n)$ for fixed input distribution. 
Hence, taking convex combinations of $V^n(z^n|x^n,s^n)$ does not increase the leakage term. 
Using Jensen's inequality and the fact that each value of $I(p_{J_n};E_uV_{f}^n)$ can also be achieved by $I(p_{J_n};E_uV_{\theta^n}^n)$, since the deterministic mappings $\mathcal{F}$ are a subset of the stochastic mappings $\mathcal{P}(\mathcal{S}^n|\mathcal{X}^n)$, $\mathcal{F} \subset \mathcal{P}(\mathcal{S}^n|\mathcal{X}^n)$, the equality is established, \cite{MolavianJazi2009}.
\end{proof}
\section{Proof of Theorem \ref{th:general}}\label{proof:th-general}

The extension from the standard \AVWC{} to the case where the jammer knows additionally the channel input is not trivial. 
When using standard proof techniques from the \AVWC , the jammer might be able to locate a channel input $x^n$ to a specific deterministic wiretap codebook $\mathcal{K}_n$. 
This automatically leads to the consideration of the deterministic code secrecy capacity of an \AVWC{} under the maximum error criterion. 
Even without secrecy constraints, this problem remains unsolved, \cite{Ahlswede1978,Noetzel2016}. 
To ensure that the confusion at the jammer with respect to the used codebook is sufficiently high, even if the channel input $x^n$ is non-causally known, we fulfill an additional requirement in contrast to the standard \AVWC .
The used codewords $x^n$ occur in multiple codebooks $\mathcal{K}_{n,\mathcal{U}_n}$, where $\mathcal{U}_n$ is the set of codebooks containing $x^n$ as codeword.  

We use random coding arguments as in \cite{Boche2019a} and generate random sets of deterministic wiretap codebooks. 
Note that we have to take into account that the jammer possesses non-causal knowledge about the channel input (and we allow knowledge of the messages, since we consider the maximum error), which results in a different error probability. 
For the prefixing we follow {\cite[Lemma 4 and its proof]{Csiszar1978}}, or {\cite[p.97, Addition of prefix channel]{Bloch2011a}} with slight modifications.
In the original system model (Figure \ref{fig:SysMod}) the jammer knows the channel input $X_u^n$. 
If we concatenate a channel with the \AVWC , and call the prefix variable $\Psi_u^n$, then the jammer does not know the channel input $\Psi_u^n$ of the concatenated channel but an intermediate variable $X^{n}$, which is in fact the channel input of the original channel. 
However, we adopt the codebook generation and decoding regions according to the concatenated channels $\rho W$ and $\rho V$, respectively, with
\begin{align*}
\rho W & = \sum_{x\in\mathcal{X}}\rho(x|\psi)W(y|x,s)\\
\rho V & = \sum_{x\in\mathcal{X}}\rho(x|\psi)V(z|x,s).
\end{align*} 
%
For the secrecy  analysis, we have to show that the leakage to the eavesdropper vanishes asymptotically. 
For the leakage analysis, we consider the mutual information $I(p_{J_n};E_u\rho V_{\theta^\ast}^n)$ 
Last, we show that the probability of obtaining codes for which both the decoding error probability and the leakage vanish asymptotically approaches one. 
For the converse, we modify the standard converse proof for the \WTC .

\subsection{Codebook Generation}

We assume that for all $u \in \mathcal{U}_n$, $p_U(u)=\frac{1}{|\mathcal{U}_n|}$.
Let $p \in \mathcal{P}(\Psi)$ be given. 
Partition the set of typical sequences $\mathcal{T}_{p,\delta}^n$ into disjoint subsets $C_{(j,l)}$ of size $|C_{(j,l)}|= \frac{|\mathcal{T}_{p,\delta}^n|}{|\mathcal{J}_n||\mathcal{L}_n|}$.
Here $j\in \mathcal{J}_n=\{1,2,\dots,J_n\}$ and $l\in \mathcal{L}_n=\{1,2,\dots,L_n\}$ correspond to the secure and confusing messages, respectively. 
We have $J_n\cdot L_n = \exp{\{nR\}}$, and the transmission rate $R$ will determined later.
Let the random variable $\Psi_{ujl}^n$ denote the codeword for the message pair $(j,l)\in\mathcal{J}_n\times \mathcal{L}_n$, if the \CR{} has the realization $U=u$. 
The codewords $\Psi_{ujl}^n$ and $\Psi_{u(jl)^\prime}^n$ are independent of each other for all $(j,l)\neq (j,l)^\prime$.
Let $\hat{\chi} :=\{\Psi_{ujl}^n: j\in\mathcal{J}_n, l\in\mathcal{L}_n, u\in \mathcal{U}_n\}$ be the family of \RV , representing the random codewords.
We start by generating a deterministic wiretap code for each $u\in \mathcal{U}_n$ (still random in terms of random coding arguments).  
To indicate that each codebook at this point is a random variable, we add the argument $\hat{\chi}$.
For each codebook $\mathcal{K}_{n,u}(\hat{\chi})$, we draw $J_n\cdot L_n$ codewords $\Psi_{ujl}^n$ uniformly from the subsets $C_{(j,l)}$.
For each $\Psi_{ujl}^n$ we generate the conditional typical set $\mathcal{T}_{\rho,\delta}^n(\Psi_{ujl}^n)$ and choose randomly $X^{n}$ uniformly distributed over $\mathcal{T}_{\rho,\delta}^n(\Psi_{ujl}^n)$ as the channel input.

%
\subsection{Decoding regions}
Let $\hat{\mathcal{D}}_{ujl}^\prime (\hat{\chi})$ be given as 
\begin{align*}
\hat{\mathcal{D}}_{ujl}^\prime(\hat{\chi}) & = \bigcup_{\underline{\theta} \in \mathcal{P}_0(\mathcal{S}^{n}|\mathcal{X}^n)}\mathcal{T}_{{\rho W_{\underline{\theta}}},\delta}^n(\Psi^n_{ujl}).
\end{align*}
with\footnote{Note that $\underline{\theta}(s|x),x\in\mathcal{X},s\in\mathcal{S}$ is a single letter distribution  on the set of all possible conditional types of $s^n$ given $x^n$.} $(\rho W_{\underline{\theta}})(y|\psi)=\sum_{\substack{x\in\mathcal{X}\\s\in\mathcal{S}}}\rho(x|\psi) \underline{\theta}(s|x) W)(y|x,s)$.

Then, we can define the decoding sets $\hat{\mathcal{D}}_{ujl}(\hat{\chi})$ as follows.
\begin{align}
\hat{\mathcal{D}}_{ujl}(\hat{\chi}) & =\hat{\mathcal{D}}_{ujl}^\prime (\hat{\chi})\bigcap \left(\bigcup_{\substack{(jl)^\prime\in \mathcal{J}_n\times \mathcal{L}_n\\(jl)^\prime\neq (jl)}}\hat{\mathcal{D}}_{u(jl)^\prime}^\prime (\hat{\chi})\right)^c
\end{align}
\subsection{Codebook properties for reliability}\label{par:properties_for_reliability}
%
%
As already mentioned, we have to make sure, that every codeword occurs in multiple codebooks. 
By generating the codebooks $\mathcal{K}_{n,u}(\hat{\chi})$, $u\in \mathcal{U}_n$ as above, there are   at most
\begin{align*}
\frac{|\mathcal{T}_{p,\delta}^n|}{J_n\cdot L_n} & = \exp{\{n(H(\Psi)- R +\epsilon_1(n))\}}
\end{align*}
nonoverlapping codebooks in the worst case, where $R$ corresponds to the code rate of a code with $J_n\cdot L_n$ messages. 
Intuitively, to ensure the occurrence of each codeword in $k$ codebooks (on average), we should use an amount of \CR{} which corresponds roughly to 
\begin{align*}
|\mathcal{U}_n|& \geq k \exp{\{n(H(\Psi)- R +\epsilon_1(n))\}}.
\end{align*}
Later, we will derive a lower bound on the amount of \CR{}, explicitly.
We follow and extend the ideas of \cite{Bjelakovic2013,Bjelakovic2013a} and \cite{Boche2019a}.
Here, in contrast to the classical \DMC , we have three error terms:
\begin{itemize}
\item given the received sequence $Y^{n}$, we do not find sequences $\Psi_{ujl}^n$ and a channel input $X^{n}\in \mathcal{T}_{\rho,\delta}^n(\Psi_{ujl}^n)$, such that $Y^{n}$ is conditional typical given $\Psi_{ujl}^n$ and $X^{n}\in \mathcal{T}_{\rho,\delta}^n(\Psi_{ujl}^n)$,
\item given the received sequence $Y^{n}$ which is conditional typical given the codeword $\Psi_{ujl}^n$ and the channel input $X^{n}\in \mathcal{T}_{\rho,\delta}^n(\Psi_{ujl}^n)$, we find another codeword $\Psi_{u(jl)^\prime}^n$ and channel input $X^{\prime n}\in \mathcal{T}_{\rho,\delta}^n(\Psi_{u(jl)^\prime}^n)$, such that $Y^{n}$ is conditional typical given $\Psi_{u(jl)^\prime}^n$ and $X^{\prime n}\in \mathcal{T}_{\rho,\delta}^n(\Psi_{u(jl)^\prime}^n)$,
\item  given the received sequence $Y^{n}$, there exist \CR{} realizations $u$, such that for some messages $(j,l)\in\mathcal{J}_n\times\mathcal{L}_n$, the codeword $\Psi_{ujl}^n=\psi^n$, the channel input $X^{n}\in \mathcal{T}_{\rho,\delta}^n(\Psi_{ujl}^n)$, $X^{n}=x^{n}$, and the state sequence $S^n=s^n$, the probability of $Y^{n}\in \hat{\mathcal{D}}_{ujl}^c(\hat{\chi})$ is lower bounded by some $\lambda$.
\end{itemize}
Since we apply random codes, we do actually not know which codebook realizations (in terms of random coding arguments) lead to a good error performance. 
But we know that the error probability vanishes averaged over a set of codebooks. 
Since the codewords occur in multiple codebooks, we have to take care of the situation that the codewords perform well in some codebooks, but not so well in others.

First, let us fix a pair $(j,l)\in \mathcal{J}_n\times\mathcal{L}_n$.
Randomly pick and fix the sequences $\psi^n \in \mathcal{C}_{(j,l)}$, $x^{n} \in \mathcal{T}_{\rho,\delta}^n(\psi^n)$ and $s^{n}\in\mathcal{S}^{n}$. 
The probability, that $\exists \underline{\theta}\in\mathcal{P}_0(\mathcal{S}^n|\mathcal{X}^n): (\psi^n,x^n,s^n)\in\mathcal{T}_{p_{\Psi}\times\rho\times\underline{\theta},\delta}$ is close to one according to Lemma \ref{lem:existenceofaQ}. 
For now, assume that $\exists \underline{\theta}\in\mathcal{P}_0(\mathcal{S}^n|\mathcal{X}^n): (\psi^n,x^n,s^n)\in\mathcal{T}_{p_{\Psi}\times\rho\times\underline{\theta},\delta}$
We have to show that if the sequence $\psi^n$ is a codeword (occurring in multiple codebooks), then the state sequence is bad only for few codebooks, such that averaged over all codebooks, the error probability still vanishes.
This has to hold for all pairs $(j,l)$, sequences $\psi^n \in \mathcal{C}_{(j,l)}$, $x^{n} \in \mathcal{T}_{\rho,\delta}^n(\psi^n)$, and $s^{n}\in\mathcal{S}^{n}$ for which there exists $\underline{\theta}\in\mathcal{P}_0(\mathcal{S}^n|\mathcal{X}^n): (\psi^n,x^n,s^n)\in\mathcal{T}_{p_{\Psi}\times\rho\times\underline{\theta},\delta}$. 
%
%
We now can define the sets $\mathcal{U}(j,l,\psi^n,x^{n},\hat{\chi})$ and $\mathcal{U}_0(j,l,\psi^n,x^{n},s^{n}, \hat{\chi})$ as
\begin{align*}
\mathcal{U}(j,l,\psi^n,x^{n},\hat{\chi})&:=\left\{u:\quad \Psi_{ujl}^{n}=\psi^{n},~ X^{n}=x^{n} \right\},\\
\mathcal{U}_0(j,l,\psi^n,x^{n},s^{n},\hat{\chi})&:=\left\{u:\quad \Psi_{ujl}^{n}=\psi^{n},~ X^{n}=x^{n}, \text{ and }  W^{n}(\hat{\mathcal{D}}_{ujl}^c(\hat{\chi})|x^{n},s^{n})>\lambda\right\}.
\end{align*}

Here, $\mathcal{U}(j,l,\psi^n,x^{n},\hat{\chi})$ denotes the set of all codebooks, for which the sequence $\psi^n$ is the codeword for the message pair $(j,l)$ and $x^{n}$ is the corresponding channel input, and $\mathcal{U}_0(j,l,\psi^n,x^{n},s^{n},\hat{\chi})$ is the set of all codebooks,for which the sequence $\psi^n$ is the codeword for the message pair $(j,l)$, $x^{n}$ is the corresponding channel input, and the error bound $\lambda$ is not met.


We can define the binary random variable $B(u,j,l,\psi^n,x^{n},\hat{\chi})$ as 
\begin{align}
B(u,j,l,\psi^n,x^{n},\hat{\chi}) & = \begin{cases}
1&\text{ if} \quad u\in \mathcal{U}(j,l,\psi^n,x^{n},\hat{\chi})\\
0&\text{ else.}  
\end{cases}\\
Pr\{B(u,j,l,\psi^n,x^{n},\hat{\chi})=1\} & = Pr\{\Psi_{ujl}^n=\psi^n\} Pr\{X^{n}=x^{n}|\Psi_{ujl}^n=\psi^n\}\\
& = \frac{1}{| \mathcal{C}_{(j,l)}|} \frac{1}{|\mathcal{T}_{\rho,\delta}^n(\psi^n)|}, \quad \forall u\in \mathcal{U}_n, \quad \forall (j,l) \in \mathcal{J}_n\times \mathcal{L}_n.
\end{align}

It indicates whether the sequences $\psi^n$ and $x^{n}$ are the prefix variable and the channel input realizations  for the codebook realization $u$ and the message pair $(j,l)$.
By the Chernoff bound we obtain
\begin{align*}
&Pr\left\{|\mathcal{U}(j,l,\psi^n,x^{n},\hat{\chi})|\leq (1-\epsilon_2)|\mathcal{U}_n|Pr\{B(u,j,l,\psi^n,x^{n},\hat{\chi})=1\}\right\} \\
&= Pr\left\{\sum_{u\in\mathcal{U}_n}B(u,j,l,\psi^n,x^{n},\hat{\chi})\leq (1-\epsilon_2)|\mathcal{U}_n|Pr\{B(u,j,l,\psi^n,x^{n},\hat{\chi})=1\}\right\}\\
& \leq \exp_e{\left\{-\frac{3\epsilon_2^2 n|\mathcal{U}_n|J_n\cdot L_n \exp\{-n(H(X|\Psi)+\delta)\}}{8|\mathcal{T}^n_{p,\delta}|}\right\}}\\
& \leq \exp_e{\left\{-\frac{3}{8}\epsilon_2^2 n|\mathcal{U}_n| \exp\{-n(H(X,\Psi)-R +\tilde{\delta})\}\right\}}.
\end{align*}
Next, we will upper bound the probability that $|\mathcal{U}_0(j,l,\psi^n,x^{n},s^{n},\hat{\chi})|$ exceeds its expected value. 
We define the binary random variable $\tilde{B}(j,l,\psi^n,x^{n},s^{n},u,\lambda,\hat{\chi})$ as 
\begin{align}
\tilde{B}(j,l,\psi^n,x^{n},s^{n},u,\lambda,\hat{\chi}) & = \begin{cases}
1&\text{ if} \quad u\in \mathcal{U}_0(j,l,\psi^n,x^{n},s^{n},\hat{\chi})\\
0&\text{ else.}  
\end{cases}\\
Pr\left\{\tilde{B}(j,l,\psi^n,x^{n},s^{n},u,\lambda,\hat{\chi})=1\right\}& =Pr\{B(u,j,l,\psi^n,x^{n},\hat{\chi})=1\}\cdot\\
&\qquad \cdot Pr\left\{W^{n}(\hat{\mathcal{D}}_{ujl}^c(\hat{\chi})|x^{n},s^{n})>\lambda|B(u,j,l,\psi^n,x^{n},\hat{\chi})=1\right\}.
\end{align}

It indicates whether the sequences $\psi^n$ and $x^{n}$ are the prefix variable and the channel input realizations  for the codebook realization $u$ and the message pair $(j,l)$, and the error bound $\lambda$ is not met.

We consider the case that the error bound is not met for a fixed $u\in \mathcal{U}_n$.  
By the Markov inequality Lemma \ref{lem:markov} and by Lemma \ref{cor:boundprefix} we have
\begin{align*}
&Pr\left\{W^{n}(\hat{\mathcal{D}}_{ujl}^c(\hat{\chi})|x^{n},s^{n})>\lambda|B(u,j,l,\psi^n,x^{n},\hat{\chi})=1\right\}\\
& \overset{(a)}{\leq} \frac{E\left[W^{n}(\hat{\mathcal{D}}_{ujl}^c(\hat{\chi})|x^{n},s^{n})|B(u,j,l,\psi^n,x^{n},\hat{\chi})=1\right]}{\lambda}\\
& \leq \frac{E\left[\left(W^{n}\left(\hat{\mathcal{D}}_{ujl}^{\prime c}(\hat{\chi}) |x^{n},s^{n}\right)  + W^{n}\left(\bigcup_{\substack{(j,l)^\prime \in \mathcal{J}_n\times\mathcal{L}_n\\(j,l)^\prime \neq (j,l)}}\hat{\mathcal{D}}_{u(jl)^\prime}^\prime(\hat{\chi})|x^{n},s^{n}\right)\right)\Big| B(u,j,l,\psi^n,x^{n},\hat{\chi})=1\right]}{\lambda}\\
& \overset{(b)}{\leq}\frac{\exp\{-n c^\prime {\delta^\prime}^2\}}{\lambda} + \frac{\sum_{\substack{(j,l)^\prime \in \mathcal{J}_n\times\mathcal{L}_n\\(j,l)^\prime \neq (j,l)}}E\left[\left( W^{n}\left(\hat{\mathcal{D}}_{u(jl)^{\prime}}^\prime(\hat{\chi})|x^{n},s^{n}\right)\right)\Big| B(u,j,l,\psi^n,x^{n},\hat{\chi})=1\right]}{\lambda}\\
& \leq \frac{\exp\{-n c^\prime {\delta^\prime}^2\}}{\lambda} + \frac{\sum_{\substack{(j,l)^\prime \in \mathcal{J}_n\times\mathcal{L}_n\\(j,l)^\prime \neq (j,l)}}E\left[W^{n}\left(\left(\smashoperator[r]{\bigcup_{\underline{\theta}\in \mathcal{P}_0(\mathcal{S}^{n}|\mathcal{X}^n)}}\mathcal{T}^n_{\rho \underline{\theta} W,\delta}(\Psi^{n}_{u(jl)^\prime})\right)\bigg\vert x^{n},s^{n}\right)\Big| B(u,j,l,\psi^n,x^{n},\hat{\chi})=1\right]}{\lambda}\\
& \overset{(c)}{\leq} \frac{\exp\{-n c^\prime {\delta^\prime}^2\}}{\lambda} + \sum_{\substack{(j,l)^\prime \in \mathcal{J}_n\times\mathcal{L}_n\\(j,l)^\prime \neq (j,l)}}\frac{\exp{\left\{-n\left(\min_{W\in \widehat{\widehat{\mathcal{W}}}}I(p_\Psi;\rho W) -\nu \right)\right\}}}{\lambda}\\
& \leq \frac{\exp\{-n c^\prime {\delta^\prime}^2\}}{\lambda} + \frac{\exp{\left\{-n\left(\min_{W\in \widehat{\widehat{\mathcal{W}}}}I(p_\Psi;\rho W)  - R -\nu \right)\right\}}}{\lambda}.
\end{align*}
Here, $(a)$ follows by the Markov inequality (Lemma \ref{lem:markov}), $(b)$ follows by Lemma \ref{lem:typsetII} and the union bound, and $(c)$ follows by Lemma \ref{cor:boundprefix} and the fact that $\Psi_{u(jl)^{\prime }}^n$ and $\Psi_{ujl}^n$ are independent of each other.

Then, identifying $p_1$ in Lemma \ref{lem:chernoffbounds} as
\begin{align*}
p_1 & = \frac{\exp\{-n c^\prime {\delta^\prime}^2\}}{\lambda} + \frac{\exp{\left\{-n\left(\min_{W\in \widehat{\widehat{\mathcal{W}}}}I(p_\Psi;\rho W)  - R -\nu \right)\right\}}}{\lambda},
\end{align*}
we can bound the probability that $|\mathcal{U}_0(j,l,\psi^n,x^{n},s^{n},\hat{\chi})|$ exceeds a certain value as
\begin{align*}
&Pr\left\{|\mathcal{U}_0(j,l,\psi^n,x^{n},s^{n},\hat{\chi})|\geq (1+\epsilon_2)|\mathcal{U}_n|Pr\{B(u,j,l,\psi^n,x^{n},\hat{\chi})=1\}p_1\right\}\\
& = Pr\left\{\sum_{u\in\mathcal{U}_n}\tilde{B}(j,l,\psi^n,x^{n},s^{n},u,\lambda,\hat{\chi})\geq (1+\epsilon_2)|\mathcal{U}_n|Pr\{B(u,j,l,\psi^n,x^{n},\hat{\chi})=1\}p_1\right\}\\
& \leq \exp_e{\left\{-\frac{\epsilon_2^2 n|\mathcal{U}_n| \exp\{-n(H(X,\Psi)-R +\tilde{\delta})\} \left(\frac{\exp\{-n c^\prime \delta^\prime\}}{\lambda} + \frac{\exp{\left\{-n\left(\min_{W\in \widehat{\widehat{\mathcal{W}}}}I(p_\Psi;\rho W)  - R -\nu\right)\right\}}}{\lambda} \right)}{8}\right\}}.
\end{align*}
Hence for all $|\mathcal{U}_n|$ fulfilling 
\begin{align*}
|\mathcal{U}_n| & > \exp{\{n(H(X,\Psi)-R +\tilde{\delta}) \}} \left(\frac{\exp\{-n c^\prime \delta^\prime\}}{\lambda} + \frac{\exp{\left\{-n\left(\min_{W\in \widehat{\widehat{\mathcal{W}}}}I(p_\Psi;\rho W)  - R -\nu \right)\right\}}}{\lambda} \right)^{-1}
\end{align*}
the probabilities that codewords do not occur in at least $1-\epsilon_2$ times the expected number of codebooks and that codewords occur in more than $1+\epsilon_2$ times the expected number of codebooks for which the error bound is not met, vanish super exponentially fast.

The above described events have to hold for all $(j,l)\in \mathcal{J}_n\times \mathcal{L}_n$, $\psi^n\in \mathcal{C}_{(j,l)}$ $x^{n}\in \mathcal{T}_{\rho,\delta}^n(\psi^n)$ and $s^{n}\in \mathcal{S}^{n}$,
for which there exists $\underline{\theta}\in\mathcal{P}_0(\mathcal{S}^n|\mathcal{X}^n): (\psi^n,x^n,s^n)\in\mathcal{T}_{p_{\Psi}\times\rho\times\underline{\theta},\delta}$. 
Hence, 
\begin{align*}
&Pr\left\{\bigcap_{(j,l)\in\mathcal{J}_n\times \mathcal{L}_n}\qquad\smashoperator[lr]{\bigcap_{\substack{\psi^n\in\mathcal{C}_{(j,l)}\\x^{n}\in\mathcal{T}_{\rho,\delta}^n(\psi^n)\\s^{n}\in\mathcal{S}^{n}\\\exists\underline{\theta}\in\mathcal{P}_0(\mathcal{S}^n|\mathcal{X}^n): (\psi^n,x^n,s^n)\in\mathcal{T}_{p_{\Psi}\times\rho\times\underline{\theta},\delta}}}}\left\{|\mathcal{U}_0(j,l,\psi^n,x^{n},s^{n},\hat{\chi})|\leq (1+\epsilon_2)|\mathcal{U}_n|\cdot Pr\{B(u,j,l,\psi^n,x^{n},\hat{\chi})=1\} p_1\right\}\right\}\\
&=  1 -Pr\left\{\left(\bigcap_{(j,l)\in\mathcal{J}_n\times \mathcal{L}_n}\qquad\smashoperator[lr]{\bigcap_{\substack{\psi^n\in\mathcal{C}_{(j,l)}\\x^{n}\in\mathcal{T}_{\rho,\delta}^n(\psi^n)\\s^{n}\in\mathcal{S}^{n}\\\exists\underline{\theta}\in\mathcal{P}_0(\mathcal{S}^n|\mathcal{X}^n): (\psi^n,x^n,s^n)\in\mathcal{T}_{p_{\Psi}\times\rho\times\underline{\theta},\delta}}}}\left\{|\mathcal{U}_0(j,l,\psi^n,x^{n},s^{n},\hat{\chi})|\leq (1+\epsilon_2)|\mathcal{U}_n|\cdot Pr\{B(u,j,l,\psi^n,x^{n},\hat{\chi})=1\} p_1\right\}\right)^c\right\}\\
& \overset{(a)}{\geq} 1 -  |\mathcal{J}_n||\mathcal{L}_n|\frac{|\mathcal{T}_{p,\delta}^n|}{|\mathcal{J}_n||\mathcal{L}_n|} |\mathcal{T}_{\rho,\delta}^n||\mathcal{S}^{n}|\exp_e\left\{-\frac{\epsilon_2^2 n |\mathcal{U}_n|\exp{\{-n(H(X,\Psi) - R + \epsilon_1(n)) \}}}{8}\right. \\
&\left. \qquad \qquad \qquad \qquad \qquad \qquad \left(\frac{\exp\{-n c^\prime \delta^\prime\}}{\lambda} + \frac{\exp{\left\{-n\left(\min_{W\in \widehat{\widehat{\mathcal{W}}}}I(p_{\Psi};\rho W)  - R -\nu \right)\right\}}}{\lambda} \right)\right\}\\
&= 1 -  |\mathcal{T}_{p,\delta}^n||\mathcal{T}_{\rho,\delta}^n||\mathcal{S}^{n}|\exp_e\left\{-\frac{\epsilon_2^2 n |\mathcal{U}_n|\exp{\{-n(H(X,\Psi) - R + \epsilon_1(n)) \}}}{8}\right. \\
&\left.   \qquad \qquad \qquad \qquad \qquad \qquad \left(\frac{\exp\{-n c^\prime \delta^\prime\}}{\lambda} + \frac{\exp{\left\{-n\left(\min_{W\in \widehat{\widehat{\mathcal{W}}}}I(p_{\Psi};\rho W)  - R -\nu \right)\right\}}}{\lambda} \right)\right\}
\end{align*}
and
\begin{align*}
&Pr\left\{\bigcap_{(j,l)\in\mathcal{J}_n\times \mathcal{L}_n}\bigcap_{\psi^n\in\mathcal{C}_{(j,l)}}\bigcap_{x^{n}\in\mathcal{T}_{\rho,\delta}^n(\psi^n)}\bigcap_{s^{n}\in\mathcal{S}^{n}}\left\{|\mathcal{U}(j,l,\psi^n,x^{n},\hat{\chi})|\leq (1-\epsilon_2)|\mathcal{U}_n|Pr\{B(u,j,l,\psi^n,x^{n},\hat{\chi})=1\}\right\}\right\}\\
&= 1 - Pr\left\{\left(\bigcap_{(j,l)\in\mathcal{J}_n\times \mathcal{L}_n}\bigcap_{\psi^n\in\mathcal{C}_{(j,l)}}\bigcap_{x^{n}\in\mathcal{T}_{\rho,\delta}^n(\psi^n)}\bigcap_{s^{n}\in\mathcal{S}^{n}}\right.\right.\\
&\qquad \qquad \qquad \qquad \qquad \qquad \left.\left.\left\{|\mathcal{U}(j,l,\psi^n,x^{n},\hat{\chi})|\leq (1-\epsilon_2)|\mathcal{U}_n|Pr\{B(u,j,l,\psi^n,x^{n},\hat{\chi})=1\}\right\}\right)^c\right\}\\
&\overset{(b)}{\geq} 1 -  |\mathcal{J}_n||\mathcal{L}_n|\frac{|\mathcal{T}_{p,\delta}^n|}{|\mathcal{J}_n||\mathcal{L}_n|}|\mathcal{T}_{\rho,\delta}^n||\mathcal{S}^{n}| \exp_e{\left\{-\frac{3\epsilon_2^2 n|\mathcal{U}_n|J_n\cdot L_n}{8|\mathcal{T}^n_{p,\delta}||\mathcal{T}_{\rho,\delta}^n|}\right\}}\\
&= 1 -  |\mathcal{T}_{p,\delta}^n||\mathcal{T}_{\rho,\delta}^n||\mathcal{S}^{n}|  \exp_e{\left\{-\frac{3\epsilon_2^2 n|\mathcal{U}_n|J_n\cdot L_n}{8|\mathcal{T}^n_{p,\delta}||\mathcal{T}_{\rho,\delta}^n|}\right\}}.
\end{align*}
Here, $(a)$ and $(b)$ follow by the union bound and summing over all $(j,l)\in \mathcal{J}_n\times \mathcal{L}_n$, $\psi^n\in \mathcal{C}_{(j,l)}$ $x^{n}\in \mathcal{T}_{\rho,\delta}^n(\psi^n)$ and $s^{n}\in \mathcal{S}^{n}$.

Furthermore, we bound the probability that the amount of sequences $(\psi^n,x^n,f(x^n))$ for which there does not exist a $\underline{\theta}\in\mathcal{P}_0(\mathcal{S}^n|\mathcal{X}^n): (\psi^n,x^n,f(x^n))\in\mathcal{T}_{p_{\Psi}\times\rho\times\underline{\theta},\delta}$ is not $\epsilon_3$ close to its expected value, vanishes super exponentially fast. 
More explicitly, for any $(j,l)\in\mathcal{J}_n\times\mathcal{L}_n$ we have
\begin{align*}
&Pr\left\{\left|\{(\psi^n,x^n,f(x^n)): \quad \nexists\underline{\theta}\in\mathcal{P}_0(\mathcal{S}^n|\mathcal{X}^n): (\psi^n,x^n,f(x^n))\in\mathcal{T}_{p_{\Psi}\times\rho\times\underline{\theta},\delta}\}\right|\geq \right.\\
&\left. \quad(1+\epsilon_3)|\mathcal{C}_{(j,l)}||\mathcal{T}_{\rho,\delta}^n(\psi^n)|(n+1)^{|\mathcal{X}||\mathcal{S}|}\exp\{-n \min_{\theta \in \mathcal{P}_{0}(\mathcal{S}^n|\mathcal{X}^n)}h_{\theta}(\delta)\}\right\}\\
&\leq \exp_e \left\{-\frac{\epsilon_3^2 n |\mathcal{C}_{(j,l)}||\mathcal{T}_{\rho,\delta}^n(\psi^n)|(n+1)^{|\mathcal{X}||\mathcal{S}|}\exp\{-n \min_{\theta \in \mathcal{P}_{0}(\mathcal{S}^n|\mathcal{X}^n)}h_{\theta}(\delta)\}}{8}\right\}\\
& \leq \exp_e \left\{-\frac{\epsilon_3^2 n \exp\{n(H(\Psi X) - R + \epsilon_1(n))\} (n+1)^{|\mathcal{X}||\mathcal{S}|}\exp\{-n \min_{\theta \in \mathcal{P}_{0}(\mathcal{S}^n|\mathcal{X}^n)}h_{\theta}(\delta)\}}{8}\right\}\\
& = \exp_e \left\{-\frac{\epsilon_3^2 n \exp\{n(H(\Psi X) - R + \epsilon_1(n)-\hat{\lambda})\} }{8}\right\},
\end{align*}
where the last inequality vanishes super exponentially fast in $n$.
\subsection{Codebook realization}
Now, let $\mathcal{K}_n^{\text{ran}}$ be a codebook realization of $\mathcal{K}_n^{\text{ran}}(\hat{\chi})$, fulfilling the aforementioned properties (codewords occur in sufficiently many (deterministic) codebooks, indexed by the realization of the \CR{}, and are bad only for few), with  $\mathcal{D}_{ujl}^\prime$ as 
\begin{align*}
\hat{\mathcal{D}}_{ujl}^\prime & = \bigcup_{\underline{\theta} \in \mathcal{P}_0(\mathcal{S}^{n}|\mathcal{X}^n)}\mathcal{T}_{\rho W_{\underline{\theta}},\delta}^n(\psi^n_{ujl}).
\end{align*}
with\footnote{Note that $\underline{\theta}(s|x),x\in\mathcal{X},s\in\mathcal{S}$ is a single letter distribution on the set of all possible conditional types of $s^n$ given $x^n$.} $(\rho W_{\underline{\theta}})(y|\psi)=\sum_{\substack{x\in\mathcal{X}\\s\in\mathcal{S}}}\rho(x\psi) \underline{\theta}(s|x) W)(y|x,s)$

and decoding sets $\mathcal{D}_{ujl}$, being as follows.
\begin{align}
\mathcal{D}_{ujl} & =\mathcal{D}_{ujl}^\prime \bigcap \left(\bigcup_{\substack{(jl)^{\prime}\in\mathcal{J}_n\times\mathcal{L}_n\\(jl)\neq (jl)^{\prime}}}\mathcal{D}_{u(jl)^{\prime}}^\prime\right)^c
\end{align} 
\subsection{Adaptation of the error criterion}
We will modify the error criterion and require that both the secret message $J$ and the confusing message $L$ should be successfully decoded at Bob. 

Hence, we have 
\begin{align*}
&\max_{j\in \mathcal{J}_n} \max_{l\in \mathcal{L}_n} \max_{f\in\mathcal{F}}\sum_{u\in \mathcal{U}_n}p_U(u)\sum_{\psi^n\in \mathit{\Psi}^n} E_u(\psi^n|j)\sum_{x^{n}\in \mathcal{T}_{\rho,\delta}^n(\psi^n)}\frac{1}{|\mathcal{T}_{\rho,\delta}^n(\psi^n)|}W^{n}(\mathcal{D}_{ujl}^c|x^{n},f(x^{n}))\\
& =\max_{j\in \mathcal{J}_n} \max_{l\in \mathcal{L}_n}\max_{f\in\mathcal{F}}\sum_{\psi^n\in \mathit{\Psi}^n}\sum_{x^{n}\in \mathcal{T}_{\rho,\delta}^n(\psi^n)}\sum_{u\in \mathcal{U}_n}p_U(u) E(\psi^n|j,l,u)\frac{1}{|\mathcal{T}_{\rho,\delta}^n(\psi^n)|}W^{n}(\mathcal{D}_{ujl}^c|x^{n},f(x^{n}))\\
& =\max_{j\in \mathcal{J}_n} \max_{l\in \mathcal{L}_n}\max_{f\in\mathcal{F}}\sum_{u\in \mathcal{U}_n}p_U(u)\smashoperator[lr]{\sum_{\substack{\psi^n\in \mathit{\Psi}^n\\x^{n}\in \mathcal{T}_{\rho,\delta}^n(\psi^n)\\\exists\underline{\theta}\in\mathcal{P}_0(\mathcal{S}^n|\mathcal{X}^n): (\psi^n,x^n,f(x^n))\in\mathcal{T}_{p_{\Psi}\times\rho\times\underline{\theta},\delta}}}} E(\psi^n|j,l,u)\frac{1}{|\mathcal{T}_{\rho,\delta}^n(\psi^n)|}W^{n}(\mathcal{D}_{ujl}^c|x^{n},f(x^{n}))\\
& \qquad  +\max_{j\in \mathcal{J}_n} \max_{l\in \mathcal{L}_n}\max_{f\in\mathcal{F}}\sum_{u\in \mathcal{U}_n}p_U(u)\smashoperator[lr]{\sum_{\substack{\psi^n\in \mathit{\Psi}^n\\x^{n}\in \mathcal{T}_{\rho,\delta}^n(\psi^n)\\\nexists\underline{\theta}\in\mathcal{P}_0(\mathcal{S}^n|\mathcal{X}^n): (\psi^n,x^n,f(x^n))\in\mathcal{T}_{p_{\Psi}\times\rho\times\underline{\theta},\delta}}}} E(\psi^n|j,l,u)\frac{1}{|\mathcal{T}_{\rho,\delta}^n(\psi^n)|}W^{n}(\mathcal{D}_{ujl}^c|x^{n},f(x^{n}))\\
& \leq \max_{j\in \mathcal{J}_n} \max_{l\in \mathcal{L}_n}\max_{f\in\mathcal{F}}\sum_{u\in \mathcal{U}_n}p_U(u)\smashoperator[lr]{\sum_{\substack{\psi^n\in \mathit{\Psi}^n\\x^{n}\in \mathcal{T}_{\rho,\delta}^n(\psi^n)\\\exists\underline{\theta}\in\mathcal{P}_0(\mathcal{S}^n|\mathcal{X}^n): (\psi^n,x^n,f(x^n))\in\mathcal{T}_{p_{\Psi}\times\rho\times\underline{\theta},\delta}}}} E(\psi^n|j,l,u)\frac{1}{|\mathcal{T}_{\rho,\delta}^n(\psi^n)|}W^{n}(\mathcal{D}_{ujl}^c|x^{n},f(x^{n}))\\
& \qquad +\max_{j\in \mathcal{J}_n} \max_{l\in \mathcal{L}_n}\max_{f\in\mathcal{F}}\smashoperator[lr]{\sum_{\substack{\psi^n\in \mathcal{C}_{(j,l)}\\x^{n}\in \mathcal{T}_{\rho,\delta}^n(\psi^n)\\\nexists\underline{\theta}\in\mathcal{P}_0(\mathcal{S}^n|\mathcal{X}^n): (\psi^n,x^n,f(x^n))\in\mathcal{T}_{p_{\Psi}\times\rho\times\underline{\theta},\delta}}}}\quad \frac{1}{|\mathcal{C}_{(j,l)}|}\frac{1}{|\mathcal{T}_{\rho,\delta}^n(\psi^n)|}\\
&\overset{(a)}{\leq} \max_{j\in \mathcal{J}_n} \max_{l\in \mathcal{L}_n}\max_{f\in\mathcal{F}}\sum_{u\in \mathcal{U}_n}p_U(u)\smashoperator[lr]{\sum_{\substack{\psi^n\in \mathit{\Psi}^n\\x^{n}\in \mathcal{T}_{\rho,\delta}^n(\psi^n)\\\exists\underline{\theta}\in\mathcal{P}_0(\mathcal{S}^n|\mathcal{X}^n): (\psi^n,x^n,f(x^n))\in\mathcal{T}_{p_{\Psi}\times\rho\times\underline{\theta},\delta}}}} E(\psi^n|j,l,u)\frac{1}{|\mathcal{T}_{\rho,\delta}^n(\psi^n)|}W^{n}(\mathcal{D}_{ujl}^c|x^{n},f(x^{n}))\\
& \qquad  + \max_{j\in \mathcal{J}_n} \max_{l\in \mathcal{L}_n}\frac{(1+\epsilon_3)|\mathcal{C}_{(j,l)}||\mathcal{T}_{\rho,\delta}^n(\psi^n)|(n+1)^{|\mathcal{X}||\mathcal{S}|}\exp\{-n \min_{\theta \in \mathcal{P}_{0}(\mathcal{S}^n|\mathcal{X}^n)}h_{\theta}(\delta)\}}{|\mathcal{C}_{(j,l)}||\mathcal{T}_{\rho,\delta}^n(\psi^n)|}\\
& = \max_{j\in \mathcal{J}_n} \max_{l\in \mathcal{L}_n}\max_{f\in\mathcal{F}}\sum_{u\in \mathcal{U}_n}p_U(u)\smashoperator[lr]{\sum_{\substack{\psi^n\in \mathit{\Psi}^n\\x^{n}\in \mathcal{T}_{\rho,\delta}^n(\psi^n)\\\exists\underline{\theta}\in\mathcal{P}_0(\mathcal{S}^n|\mathcal{X}^n): (\psi^n,x^n,f(x^n))\in\mathcal{T}_{p_{\Psi}\times\rho\times\underline{\theta},\delta}}}} E(\psi^n|j,l,u)\frac{1}{|\mathcal{T}_{\rho,\delta}^n(\psi^n)|}W^{n}(\mathcal{D}_{ujl}^c|x^{n},f(x^{n}))\\
& \qquad  + (1+\epsilon_3)(n+1)^{|\mathcal{X}||\mathcal{S}|}\exp\{-n \min_{\theta \in \mathcal{P}_{0}(\mathcal{S}^n|\mathcal{X}^n)}h_{\theta}(\delta)\}\\
& \leq \max_{j\in \mathcal{J}_n} \max_{l\in \mathcal{L}_n} \max_{f\in\mathcal{F}}\smashoperator[lr]{\sum_{\substack{\psi^n\in \mathit{\Psi}^n\\x^{n}\in \mathcal{T}_{\rho,\delta}^n(\psi^n)\\\exists\underline{\theta}\in\mathcal{P}_0(\mathcal{S}^n|\mathcal{X}^n): (\psi^n,x^n,f(x^n))\in\mathcal{T}_{p_{\Psi}\times\rho\times\underline{\theta},\delta}}}}\quad\sum_{u\in \mathcal{U}_n}p_{UJL\Psi^n X^{n}}(u,j,l,\psi^n,x^{n})W^{n}(\mathcal{D}_{u,j,l}^c|x^{n},f(x^{n}))+ \exp\{-n\hat{\lambda}\}\\
& = \max_{j\in \mathcal{J}_n}\max_{l\in \mathcal{L}_n} \max_{f\in\mathcal{F}}\smashoperator[lr]{\sum_{\substack{\psi^n\in \mathit{\Psi}^n\\x^{n}\in \mathcal{T}_{\rho,\delta}^n(\psi^n)\\\exists\underline{\theta}\in\mathcal{P}_0(\mathcal{S}^n|\mathcal{X}^n): (\psi^n,x^n,f(x^n))\in\mathcal{T}_{p_{\Psi}\times\rho\times\underline{\theta},\delta}}}}\quad\sum_{u\in \mathcal{U}_n}p_{U|JL\Psi^n X^{n}}(u|j,l,\psi^n,x^{n})p_{J L \Psi^n X^{n}}(j,l,\psi^n,x^{n})W^n(\mathcal{D}_{ujl}^c|x^n,f(x^n)) + \exp\{-n\hat{\lambda}\}\\
&\leq \max_{j\in \mathcal{J}_n}\max_{l\in \mathcal{L}_n}\max_{f\in\mathcal{F}}\quad\smashoperator[lr]{\max_{\substack{\psi^n\in \mathit{\Psi}^n\\x^{n}\in \mathcal{T}_{\rho,\delta}^n(\psi^n)\\\exists\underline{\theta}\in\mathcal{P}_0(\mathcal{S}^n|\mathcal{X}^n): (\psi^n,x^n,f(x^n))\in\mathcal{T}_{p_{\Psi}\times\rho\times\underline{\theta},\delta}}}}\quad ~\sum_{u\in \mathcal{U}_n}p_{U|\Psi^n X^{n} JL}(u|\psi^n,x^{n},j,l)W^{n}(\mathcal{D}_{ujl}^c|x^{n},f(x^{n}))+ \exp\{-n\hat{\lambda}\}\\
&\leq \max_{j\in \mathcal{J}_n}\max_{l\in \mathcal{L}_n}\quad\smashoperator[lr]{\max_{\substack{\psi^n\in \mathit{\Psi}^n\\x^{n}\in \mathcal{T}_{\rho,\delta}^n(\psi^n)\\s^{n}\in \mathcal{S}^{n}\\\exists\underline{\theta}\in\mathcal{P}_0(\mathcal{S}^n|\mathcal{X}^n): (\psi^n,x^n,s^n)\in\mathcal{T}_{p_{\Psi}\times\rho\times\underline{\theta},\delta}}}}\quad ~\sum_{u\in \mathcal{U}_n}p_{U|\Psi^n X^{n} J L}(u|\psi^n,x^{n},j,l)W^{n}(\mathcal{D}_{ujl}^c|x^{n},s^{n}) + \exp\{-n\hat{\lambda}\}\\
&:=\hat{\hat{e}}(\mathcal{K}_n^{\text{ran}})
\end{align*}

We first split the error probability into two terms with respect to sequences $(\psi^n,x^n,f(x^n))$. In the first term, there exists a $\underline{\theta}\in\mathcal{P}_0(\mathcal{S}^n|\mathcal{X}^n): (\psi^n,x^n,f(x^n))\in\mathcal{T}_{p_{\Psi}\times\rho\times\underline{\theta},\delta}$, in the second term there does not exist such a $\underline{\theta}\in\mathcal{P}_0(\mathcal{S}^n|\mathcal{X}^n)$. 
Here, we have implicitly shown in Appendix \ref{par:properties_for_reliability}, that $(a)$ follows with probability 1, where Lemma \ref{lem:existenceofaQ} is applied.

Secondly, we consider the maximization over all terms ($(\psi^n,x^n,s^n)$).
Our motivation to do so is to reduce the size of the space, over which should be optimized. 
The family $\mathcal{F}=\{f: \mathcal{X}^{n}\to \mathcal{S}^{n}\}$ consists of $|\mathcal{F}|=|\mathcal{S}^{n}|^{|\mathcal{X}^{n}|}$ elements, hence it grows doubly exponentially with $n$. 
By considering the maximum with respect to $x^{n}$, it is sufficient to consider the state sequence $s^{n}$ maximizing the error probability. 
Hence, we can reduce the space size used for optimization to $\mathcal{X}^{n} \times \mathcal{S}^{n}$, which grows only exponentially in $n$.
\subsection{Error Analysis}
For the error probability we can overall conclude
\begin{align*}
&\hat{\hat{e}}(\mathcal{K}_n^{\text{ran}}) =  \exp\{-n\hat{\lambda}\} +  \max_{j\in \mathcal{J}_n}\max_{l\in \mathcal{L}_n}\quad\smashoperator[lr]{\max_{\substack{\psi^n\in \mathit{\Psi}^n\\x^{n}\in \mathcal{T}_{\rho,\delta}^n(\psi^n)\\s^{n}\in \mathcal{S}^{n}\\\exists\underline{\theta}\in\mathcal{P}_0(\mathcal{S}^n|\mathcal{X}^n): (\psi^n,x^n,s^n)\in\mathcal{T}_{p_{\Psi}\times\rho\times\underline{\theta},\delta}}}}\quad ~\sum_{u\in \mathcal{U}(j,l,x^{n},\psi^n)}p_{U|\Psi^n X^{n} J L}(u|\psi^n,x^{n},j,l)W^{n}(\mathcal{D}_{ujl}^c|x^{n},s^{n})\\
&= \exp\{-n\hat{\lambda}\} + \max_{j\in \mathcal{J}_n}\max_{l\in \mathcal{L}_n}\quad\smashoperator[lr]{\max_{\substack{\psi^n\in \mathit{\Psi}^n\\x^{n}\in \mathcal{T}_{\rho,\delta}^n(\psi^n)\\s^{n}\in \mathcal{S}^{n}\\\exists\underline{\theta}\in\mathcal{P}_0(\mathcal{S}^n|\mathcal{X}^n): (\psi^n,x^n,s^n)\in\mathcal{T}_{p_{\Psi}\times\rho\times\underline{\theta},\delta}}}}\quad ~\left(\smashoperator[r]{\sum_{u\in \mathcal{U}_0^c(j,l,\psi^n,x^{n},s^{n})}}p_{U|\Psi^n X^{n} J L}(u|\psi^n,x^{n},j,l)W^{n}(\mathcal{D}_{ujl}^c|x^{n},s^{n})\right.\\
&\qquad \qquad \qquad \qquad \left.+\smashoperator[r]{\sum_{u\in \mathcal{U}_0(j,\psi^n,l,x^{n},s^{n})}}p_{U|\Psi^n X^{n} J L}(u|\psi^n,x^{n},j,l)W^{n}(\mathcal{D}_{ujl}^c|x^{n},s^{n})\right)\\
&\leq \exp\{-n\hat{\lambda}\} + \lambda + \max_{j\in \mathcal{J}_n}\max_{l\in \mathcal{L}_n}\max_{\psi^n\in\mathit{\Psi}^n}\max_{x^{n}\in\mathcal{T}_{\rho,\delta}^n(\psi^n)}\max_{s^{n}\in \mathcal{S}^{n}}\smashoperator[r]{\sum_{u\in \mathcal{U}_0(j,\psi^n,l,x^{n},s^{n})}}p_{U|\Psi^n X^{n} J L}(u|\psi^n,x^{n},j,l)W^{n}(\mathcal{D}_{ujl}^c|x^{n},s^{n})\\
&\leq  \exp\{-n\hat{\lambda}\} + \lambda +\max_{j\in \mathcal{J}_n}\max_{l\in \mathcal{L}_n}\quad\smashoperator[lr]{\max_{\substack{\psi^n\in \mathit{\Psi}^n\\x^{n}\in \mathcal{T}_{\rho,\delta}^n(\psi^n)\\s^{n}\in \mathcal{S}^{n}\\\exists\underline{\theta}\in\mathcal{P}_0(\mathcal{S}^n|\mathcal{X}^n): (\psi^n,x^n,s^n)\in\mathcal{T}_{p_{\Psi}\times\rho\times\underline{\theta},\delta}}}}\quad ~\smashoperator[r]{\sum_{u\in \mathcal{U}_0(j,\psi^n,l,x^{n},s^{n})}}p_{U|\Psi^n X^{n} J L}(u|\psi^n,x^{n},j,l)\\
&=\exp\{-n\hat{\lambda}\} + \lambda+\max_{j\in \mathcal{J}_n}\max_{l\in \mathcal{L}_n}\quad\smashoperator[lr]{\max_{\substack{\psi^n\in \mathit{\Psi}^n\\x^{n}\in \mathcal{T}_{\rho,\delta}^n(\psi^n)\\s^{n}\in \mathcal{S}^{n}\\\exists\underline{\theta}\in\mathcal{P}_0(\mathcal{S}^n|\mathcal{X}^n): (\psi^n,x^n,s^n)\in\mathcal{T}_{p_{\Psi}\times\rho\times\underline{\theta},\delta}}}}\quad ~\sum_{u\in \mathcal{U}_0(j,\psi^n,l,x^{n},s^{n})}\frac{p_{U \Psi^n X^{n} J L}(u,\psi^n,x^{n},j,l)}{p_{\Psi^n X^{n} J L}(\psi^n,x^{n},j,l)}\\
&=\exp\{-n\hat{\lambda}\} + \lambda+\max_{j\in \mathcal{J}_n}\max_{l\in \mathcal{L}_n}\quad\smashoperator[lr]{\max_{\substack{\psi^n\in \mathit{\Psi}^n\\x^{n}\in \mathcal{T}_{\rho,\delta}^n(\psi^n)\\s^{n}\in \mathcal{S}^{n}\\\exists\underline{\theta}\in\mathcal{P}_0(\mathcal{S}^n|\mathcal{X}^n): (\psi^n,x^n,s^n)\in\mathcal{T}_{p_{\Psi}\times\rho\times\underline{\theta},\delta}}}}\quad ~\frac{\sum_{u\in \mathcal{U}_0(j,\psi^n,l,x^{n},s^{n})}p_{U \Psi^n X^{n} J L}(u,\psi^n,x^{n},j,l)}{\sum_{u^\prime \in\mathcal{U}(j,l,\psi^n,x^{n})}p_{U \Psi^n X^{n} J L}(u^\prime,\psi^n,x^{n},j,l)} \\
&=\exp\{-n\hat{\lambda}\} + \lambda+\max_{j\in \mathcal{J}_n}\max_{l\in \mathcal{L}_n}\quad\smashoperator[lr]{\max_{\substack{\psi^n\in \mathit{\Psi}^n\\x^{n}\in \mathcal{T}_{\rho,\delta}^n(\psi^n)\\s^{n}\in \mathcal{S}^{n}\\\exists\underline{\theta}\in\mathcal{P}_0(\mathcal{S}^n|\mathcal{X}^n): (\psi^n,x^n,s^n)\in\mathcal{T}_{p_{\Psi}\times\rho\times\underline{\theta},\delta}}}}\quad ~\frac{\sum_{u\in \mathcal{U}_0(j,\psi^n,l,x^{n},s^{n})}p_{U}(u)p_{\Psi^n|UJL}(\psi^n|u,j,l)p_{X^{n}|\Psi^n}(x^{n}|\psi^n)}{\sum_{u^\prime \in\mathcal{U}(j,l,\psi^n,x^{n})}p_{U}(u^\prime)p_{\Psi^n|UJL}(\psi^n|u^\prime,j,l)p_{X^{n}|\Psi^n}(x^{n}|\psi^n)} \\
&= \exp\{-n\hat{\lambda}\} + \lambda+\max_{j\in \mathcal{J}_n}\max_{l\in \mathcal{L}_n}\quad\smashoperator[lr]{\max_{\substack{\psi^n\in \mathit{\Psi}^n\\x^{n}\in \mathcal{T}_{\rho,\delta}^n(\psi^n)\\s^{n}\in \mathcal{S}^{n}\\\exists\underline{\theta}\in\mathcal{P}_0(\mathcal{S}^n|\mathcal{X}^n): (\psi^n,x^n,s^n)\in\mathcal{T}_{p_{\Psi}\times\rho\times\underline{\theta},\delta}}}}\quad ~\frac{\sum_{u\in \mathcal{U}_0(j,\psi^n,l,x^{n},s^{n})}p_{U}(u)}{\sum_{u^\prime \in\mathcal{U}(j,l,\psi^n,x^{n})}p_{U}(u^\prime)} \\
& =\exp\{-n\hat{\lambda}\} + \lambda + \max_{j\in \mathcal{J}_n}\max_{l\in \mathcal{L}_n}\quad\smashoperator[lr]{\max_{\substack{\psi^n\in \mathit{\Psi}^n\\x^{n}\in \mathcal{T}_{\rho,\delta}^n(\psi^n)\\s^{n}\in \mathcal{S}^{n}\\\exists\underline{\theta}\in\mathcal{P}_0(\mathcal{S}^n|\mathcal{X}^n): (\psi^n,x^n,s^n)\in\mathcal{T}_{p_{\Psi}\times\rho\times\underline{\theta},\delta}}}}\quad ~ \frac{|\mathcal{U}_0(j,l,\psi^n,x^{n},s^{n})|}{|\mathcal{U}(j,l,\psi^n,x^{n})|}.
\end{align*}
In Appendix \ref{par:properties_for_reliability}, we have implicitly shown, that the probability
\begin{align*}
&Pr\left\{\frac{|\mathcal{U}_0(j,l,\psi^n,x^{n},s^{n})|}{|\mathcal{U}(j,l,\psi^n,x^{n})|} \geq \frac{(1+\epsilon_2)|\mathcal{U}_n|Pr\{B(u,j,l,\psi^n,x^{n},\hat{\chi})=1\} p_1}{(1-\epsilon_2)|\mathcal{U}_n|Pr\{B(u,j,l,\psi^n,x^{n},\hat{\chi})=1\}}\right\}
\end{align*}
vanishes super exponentially fast.
Hence, with probability 1, we can upper bound $\hat{\hat{e}}(\mathcal{K}_n^{\text{ran}})$ as 
\begin{align*}
&\hat{\hat{e}}(\mathcal{K}_n^{\text{ran}})  \leq  \exp\{-n\hat{\lambda}\} + \lambda + \max_{j\in \mathcal{J}_n}\max_{l\in \mathcal{L}_n}\quad\smashoperator[lr]{\max_{\substack{\psi^n\in \mathit{\Psi}^n\\x^{n}\in \mathcal{T}_{\rho,\delta}^n(\psi^n)\\s^{n}\in \mathcal{S}^{n}\\\exists\underline{\theta}\in\mathcal{P}_0(\mathcal{S}^n|\mathcal{X}^n): (\psi^n,x^n,s^n)\in\mathcal{T}_{p_{\Psi}\times\rho\times\underline{\theta},\delta}}}}\quad ~ \frac{(1+\epsilon_2)|\mathcal{U}_n|Pr\{B(u,j,l,\psi^n,x^{n},\hat{\chi})=1\} p_1}{(1-\epsilon_2)|\mathcal{U}_n|Pr\{B(u,j,l,\psi^n,x^{n},\hat{\chi})=1\}} \\
& = \exp\{-n\hat{\lambda}\} + \lambda + \frac{1+\epsilon_2}{1-\epsilon_2}\left(\frac{\exp\{-n c^\prime \delta^\prime\}}{\lambda} + \frac{\exp{\left\{-n\left(\min_{W\in \widehat{\widehat{\mathcal{W}}}}I(p_{\Psi};\rho W)  - R -\nu \right)\right\}}}{\lambda} \right)
\end{align*}
We choose 
\begin{align*}
R & \leq \min_{W\in \widehat{\widehat{\mathcal{W}}}}I(p_{\Psi};\rho W)-\nu \\
\lambda&=\exp\{-n\frac{\tau}{2}\},\\
\tau & < \min\left\{c^\prime \delta^\prime,\min_{W\in \widehat{\widehat{\mathcal{W}}}}I(p_{\Psi};\rho W)  - R -\nu \right\}
\end{align*}
and have shown an exponential vanishing error probability.
%
\subsection{Codebook properties for secure communication}
We have to show that the leakage to the eavesdropper vanishes asymptotically. 
Therefore, we make use of the fact that there exists a best channel to the eavesdropper and the fact that the probability that the implied probability distributions are not in an $\epsilon$ region around the expected typical ones can be upper bounded using Chernoff bounds. 
Then we apply Lemma \ref{lem:csiz}.
If the variation distance of the channel output probability distribution and the conditional channel output probability distribution can be upper bounded, then the leakage can be upper bounded as well. 
To upper bound the variation distance, the triangle inequality will be used in combination  with properties of typical sequences.  
Note that the existence of a best channel to the eavesdropper is crucial at this point to reduce the jammer's possible choices of jamming sequence from double exponentially many to exactly one, for the case of a best channel to the eavesdropper. 

Notice that in contrast to the error analysis we do not average with respect to the \CR{} when considering the leakage. 
In other words, the leakage has to vanish for all $u\in \mathcal{U}_n$, hence we will omit indexing on $u$. 
Operationally, that means the eavesdropper may have access to the \CR . 
It is sufficient to consider the best channel to the eavesdropper, invoked by $\theta^{\ast,n}\in\mathcal{P}^n(\mathcal{S}|\mathcal{X})$, since fulfilling the secrecy requirement for the best channel to the eavesdropper implies that the secrecy requirement is fulfilled for all other channels to the eavesdropper by the data processing inequality, as well.

For a fixed $u\in \mathcal{U}_n$, we have
\begin{align*}
I(p_{J_n};E_uV_{\theta^{\ast,n}}^{n}) & = H(p_{J_n}E_uV_{\theta^{\ast,n}}^{n}) - H(E_uV_{\theta^{\ast,n}}^{n}|p_{J_n})\quad( = H(Z^{n}_{\theta^{\ast,n}}) - H(Z^{n}_{\theta^{\ast,n}}|J))\\
& = \frac{1}{J_n}\sum_{j\in\mathcal{J}_n}(H(p_{J_n}E_uV_{\theta^{\ast,n}}^{n}) - H(E_uV_{\theta^{\ast,n}}^{n}|j))\\
& =  \frac{1}{J_n}\sum_{j\in\mathcal{J}_n} \left( H\left(\frac{1}{J_n}\sum_{j\in\mathcal{J}_n}\sum_{\psi^n\in\mathit{\Psi}^n}\sum_{x^{n}\in \mathcal{T}_{\rho,\delta}(\psi^n)}p_u(\psi^n|j)\rho(x^{n}|\psi^n)V_{\theta^{\ast,n}}(\cdot|x^{n})\right)\right.\\
&\left. \qquad \qquad - H\left(\sum_{\psi^n\in\mathit{\Psi}^n}\sum_{x^{n}\in \mathcal{T}_{\rho,\delta}(\psi^n)}p_u(\psi^n|j)\rho(x^{n}|\psi^n)V_{\theta^{\ast,n}}(\cdot|x^{n})\right)\right)\\
& = \frac{1}{J_n}\sum_{j\in\mathcal{J}_n}\left( H\left(\rho\bar{V}_{\theta^{\ast,n}}(\cdot)\right) - H\left(\rho\hat{V}_{\theta^{\ast,n}}(\cdot|j)\right)\right),
\end{align*}
where we define
\begin{align*}
\frac{1}{J_n}\sum_{j\in\mathcal{J}_n}\sum_{\psi^n\in\mathit{\Psi}^n}\sum_{x^{n}\in \mathcal{T}_{\rho,\delta}(\psi^n)}p_u(\psi^n|j)\frac{1}{|\mathcal{T}_{\rho,\delta}^n(\psi^n)|}V_{\theta^{\ast,n}}(\cdot|x^{n}) & = \rho\bar{V}_{\theta^{\ast,n}}(\cdot)\\
\sum_{\psi^n\in\mathit{\Psi}^n}\sum_{x^{n}\in \mathcal{T}_{\rho,\delta}(\psi^n)}p_u(\psi^n|j)\frac{1}{|\mathcal{T}_{\rho,\delta}^n(\psi^n)|}V_{\theta^{\ast,n}}(\cdot|x^{n}) & = \rho\hat{V}_{\theta^{\ast,n}}(\cdot|j).
\end{align*}
Now, if we can show that 
\begin{align*}
|| \rho\bar{V}_{\theta^{\ast,n}}(\cdot) - \rho\hat{V}_{\theta^{\ast,n}}(\cdot|j)||_V & \leq \epsilon_3 \leq \frac{1}{2} 
\end{align*}
then we can apply Lemma \ref{lem:csiz} and obtain
\begin{align*}
|  H(\rho\bar{V}_{\theta^{\ast,n}}(\cdot)) - H(\rho\hat{V}_{\theta^{\ast,n}}(\cdot|j))| & \leq - \epsilon_3 \log{\frac{\epsilon_3 }{|\mathcal{Z}|^n}}
\end{align*}
We extend \cite{Wiese2016} to prove that the secrecy requirement is fulfilled. 
For some $\Omega(Z^n)$ that will be defined later in this section, we have
\begin{align}
|| \rho\bar{V}_{\theta^{\ast,n}}(\cdot) - \rho\hat{V}_{\theta^{\ast,n}}(\cdot|j)||_V & \leq ||\rho\hat{V}_{\theta^{\ast,n}}(\cdot|j) - \Omega(\cdot)||_V + || \Omega(\cdot) - \rho\bar{V}_{\theta^{\ast,n}}(\cdot) ||_V . \label{eq:vdistance1prefix}
\end{align}
We will concentrate on the first term, since 
\begin{align*}
|| \Omega(\cdot) - \rho\bar{V}_{\theta^{\ast,n}}(\cdot) ||_V & = ||\frac{1}{J_n}\sum_{j\in\mathcal{J}_n} \left(\rho\hat{V}_{\theta^{\ast,n}}(\cdot|j) - \Omega(\cdot)\right)||_V\\
& \leq \frac{1}{J_n}\sum_{j\in\mathcal{J}_n} ||\rho\hat{V}_{\theta^{\ast,n}}(\cdot|j) - \Omega(\cdot)||_V .
\end{align*}
Let $(\psi^n,s^{n})$ have type $p_0\in \mathcal{P}_0^n(\Psi^n \mathcal{S}^{n})$, with 
\begin{align*}
p_0(\psi^n,s^{n}) & = p^n(\psi^n)\sum_{x^{n}\in \mathcal{T}_{\rho,\delta}(\psi^n)}\frac{1}{|\mathcal{T}_{\rho,\delta}(\psi^n)|} \theta^{\ast,n}(s^{n}|x^{n})\\
& = p^n(\psi^n)\sum_{x^{n}\in \mathcal{T}_{\rho,\delta}(\psi^n)}\frac{1}{|\mathcal{T}_{\rho,\delta}(\psi^n)|} \prod_{i=1}^n\theta^{\ast}_i(s_i|x_i)\\
& \overset{(a)}{=} p^n(\psi^n)\sum_{x^{n}\in \mathcal{T}_{\rho,\delta}(\psi^n)}\frac{1}{|\mathcal{T}_{\rho,\delta}(\psi^n)|} \prod_{i=1}^n\theta^{\ast}(s_i|x_i),
\end{align*}
where $(a)$ follows because of Definition \ref{def:bestchannel}. This effectively transforms the channel $V_{\theta^{\ast, n}}(z^n|x^n)$ to a \DMC{} with transition probability $V_{\theta^{\ast, n}}(z^n|x^n)=\prod_{i=1}^n \sum_{s\in\mathcal{S}}\theta^\ast(s_i|x_i)V(z_i|x_i,s_i)$.
We define the set $\mathbf{\varepsilon}_1(\psi^n)$ and $\tilde{\Omega}(z^{n})$ as
\begin{align}
\mathbf{\varepsilon}_1(\psi^n) &= \mathcal{T}_{\rho V_{\theta^\ast},\delta}^n(\psi^n),\\
\tilde{\Omega}(z^{n}) & = \mathbb{E}_{\Psi^n}[\rho V_{\theta^{\ast,n}}^{n}(z^{n}|\Psi^n)\mathds{1}_{\mathbf{\varepsilon}_1(\Psi^n)}(z^{n})],
\end{align}
%
%
where we take the expectation over all $\psi^n \in \mathcal{T}_{p,\delta}^n$, and $\rho V_{\theta^{\ast,n}}^n(z^{n}|\psi^n)$ is defined as 
\begin{align*}
\rho V_{\theta^{\ast,n}}^{n}(z^{n}|\psi^n) & = \sum_{x^{n}\in\mathcal{T}_{\rho ,\delta}^{n}(\psi^n)}\frac{1}{|\mathcal{T}_{\rho,\delta}(\psi^n)|} \sum_{s^{n}\in\mathcal{S}^{n}}\theta^{\ast ,n} (s^{n}|x^{n}) V^{n}(z^{n}|x^{n},s^{n})
\end{align*}
%
%
%
%
Further, we define the set 
\begin{align}
\mathbf{\varepsilon}_2&:= \left\{z^{n} \in \mathcal{T}_{Z_{\theta^{\ast,n}},2|\mathcal{X}||\Psi|\delta}: \tilde{\Omega}(z^{n}) \geq \exp\{-nc^\prime \delta^2\} \exp\{-n(H(Z_{\theta^{\ast}})+f_1(\delta))\} \right\}\label{eq:lowerboundEpsilon2prefix},
\end{align}
with
\begin{align*}
|\mathcal{T}_{Z_{\theta^{\ast,n}},2|\mathcal{X}||\Psi|\delta}| & \leq \exp\{n(H(Z_{\theta^{\ast}})+f_1(\delta))\},\\
\epsilon_n & = \exp\{-nc^\prime \delta^2\}.
\end{align*}
where these bounds follow by Lemmas \ref{lem:typsetI} and \ref{lem:typsetII}, respectively.
%
We set
\begin{align}
\Omega(z^{n}) & = \tilde{\Omega}(z^{n})\mathds{1}_{\mathbf{\varepsilon}_2}(z^{n}).
\end{align}
By definition, $\Omega(z^{n})\geq \epsilon_n \exp\{-n(H(Z_{\theta^{\ast}})+f_1(\delta))\}$, for all $z^{n}\in \mathbf{\varepsilon}_2$, else $\Omega(z^{n})=0$. 
Note, that when summing up over all $z^{n}\in \mathbf{\varepsilon}_2$ we get
\begin{align*}
\sum_{z^{n}\in \mathbf{\varepsilon}_2}\Omega(z^{n}) & = \Omega(\mathbf{\varepsilon}_2)\\
& =  \tilde{\Omega}(\mathbf{\varepsilon}_2)\\
&=\tilde{\Omega}\left(\mathcal{T}_{Z_{\theta^{\ast,n}},2|\mathcal{X}||\Psi|\delta}\right)-\tilde{\Omega}\left(\mathcal{T}_{Z_{\theta^{\ast,n}},2|\mathcal{X}||\Psi|\delta}\setminus \mathbf{\varepsilon}_2\right)\\
& \geq 1 - 2 \epsilon_n,
\end{align*}
where the inequality follows by the properties of typical sets and sequences, Lemma \ref{lem:typsetII}, hence by $\tilde{\Omega}\left(\mathcal{T}_{Z_{\theta^{\ast,n}},2|\mathcal{X}||\Psi|\delta}\right)\geq 1-\epsilon_n$, and $\tilde{\Omega}\left(\mathcal{T}_{Z_{\theta^{\ast,n}},2|\mathcal{X}||\Psi|\delta}\setminus \mathbf{\varepsilon}_2\right)\leq \epsilon_n$.
Similar to \cite{Wiese2016} we obtain a modification of $\rho V^{n}_{\theta^{\ast,n}}$ as
\begin{align}
Q_{\theta^{\ast,n}}(z^{n}|\psi^n) & := \rho V_{\theta^{\ast,n}}^{n}(z^{n}|\psi^n)\mathds{1}_{\mathbf{\varepsilon}_1(\psi^n)}(z^{n})\mathds{1}_{\mathbf{\varepsilon}_2}(z^{n}),
\end{align}
and can define the event 
\begin{align}
\iota_1(j,z^{n}) &:= \left\{\frac{1}{L_n}\sum_{l=1}^{L_n}Q_{\theta^{\ast,n}}(\Psi_{jl}^n|z^{n})\in [(1\pm\epsilon_n)\Omega(z^{n})]\right\}
\end{align}
\begin{lem}\label{lem:vdistanceaprefix}
For $\tau_a >0$, the probability that $\iota_1(j,z^{n})$ is not fulfilled can be upper bounded as
\begin{align}
Pr\{\iota_1(j,z^{n})^c\} & \leq 2 \exp_e\left\{-\frac{1}{3}\exp\{n\tau_a\}\right\} \label{eq:vdistancea1prefix}
\end{align}
\end{lem}
\begin{proof}
We will apply a Chernoff-Hoeffding bound, Lemma \ref{lem:chernoff}. 
\begin{align*}
Pr\left\{\frac{1}{L_n}\sum_{l=1}^{L_n}Q_{\theta^{\ast,n}}(z^{n}|\Psi_{jl}^n)\notin [(1\pm\epsilon_n)\Omega(z^{n})]\right\}\leq 2 \exp_e{\left(-L_n\frac{\epsilon_n^2\Omega(z^{n})}{3b_n}\right)}.
\end{align*}
We can plug in the bounds for $Q_{\theta^{\ast,n}}(\Psi_{jl}^n,z^{n})$ according to $\varepsilon_1(\psi^n)$, and $\Omega(z^{n})$ according to $\varepsilon_2$, 
\begin{align*}
Q_{\theta^{\ast,n}}(z^{n}|\Psi_{jl}^n) & \leq \exp\{-n(H(Z_{\theta^{\ast}}|\Psi)-f_2(\delta))\},\\
\Omega(z^{n}) & \geq \epsilon_n \exp\{-n(H(Z_{\theta^{\ast}})+f_1(\delta))\},
\end{align*}
and obtain for the exponent
\begin{align*}
-L_n\frac{\epsilon_n^2\Omega(z^{n})}{3b_n} & \leq -\frac{1}{3}L_n \epsilon_n^3 \exp\{-n(H(Z_{\theta^{\ast}})+f_1(\delta))\}\exp\{n(H(Z_{\theta^{\ast}}|\Psi)-f_2(\delta))\}\\
& = -\frac{1}{3}L_n \exp\{-n(H(Z_{\theta^{\ast}})-H(Z_{\theta^{\ast}}|\Psi)+f_1(\delta)+f_2(\delta)) + 3c^\prime \delta^2\}\\
& = -\frac{1}{3}L_n \exp\{-n(I(Z_{\theta^{\ast}};\Psi)+f_1(\delta)+f_2(\delta)) + 3c^\prime \delta^2\}.
\end{align*}
If we choose $L_n$ to be
\begin{align*}
L_n & \geq \exp\{n(I(Z_{\theta^{\ast}};\Psi)+f_1(\delta)+f_2(\delta) + 3c^\prime \delta^2 + \tau_a)\},\\
&\lim_{\delta\to 0}f_1(\delta) = \lim_{\delta\to 0} f_2(\delta) = \lim_{\delta\to 0} 3c^\prime \delta^2 = 0, 
\end{align*}
then the probability that $\iota_1(j,z^{n})$ is not fulfilled vanishes doubly exponentially fast.
\end{proof}
We define the event $\iota_0$ as the event that $\iota_1(j,z^{n})$ holds for all $j\in \mathcal{J}_n$, $z^{n} \in \mathcal{Z}^{n}$, and $u\in \mathcal{U}_n$ 
\begin{align}
\iota_0 := \bigcap_{j\in \mathcal{J}_n} \bigcap_{z^{n} \in \mathcal{Z}^{n}}\bigcap_{u\in \mathcal{U}_n}\iota_1(j,z^n).
\end{align}
We can bound the probability of $\iota_0$ from below as 
\begin{align*}
Pr\{\iota_0\} & = 1 - Pr\{\iota_0^c\}\\
			  &	= 1 -  Pr\left\{\bigcup_{j\in \mathcal{J}_n} \bigcup_{z^{n} \in \mathcal{Z}^{n}}\bigcup_{u\in \mathcal{U}_n}\iota_1^c(j,z^n)\right\}\\
			  & \geq 1 - 2 J_n |\mathcal{Z}|^{n}|\mathcal{U}_n| \exp_e\{-\frac{1}{3}\exp\{n\tau_a\}\}.
\end{align*}
Since $\mathcal{J}_n$, $|\mathcal{Z}|^{n}$, and $|\mathcal{U}_n|$ grow only exponentially fast in $n$, but $Pr\left\{\iota_1^c(j,z^{n})\right\}$ vanishes doubly exponentially fast in $n$, the probability that $\iota_0$ holds, approaches one.
\paragraph{Leakage analysis}
Let $\mathcal{K}_n^\text{ran}$ be a realization of the random \CR{} assisted code $\mathcal{K}_n^\text{ran}(\hat{\chi})$, fulfilling the required properties for guaranteeing secrecy. 
Furthermore, let $\psi_{jl}^n$ be the codeword realization for the message pair $(j,l)\in\mathcal{J}_n\times\mathcal{L}_n$ for the \CR{} assisted code $\mathcal{K}_n^\text{ran}$ for a specific realization of $u\in\mathcal{U}_n$. 
Keep in mind that the leakage has to vanish for all $u\in\mathcal{U}_n$, and that we omit the indexing on $u$ as before.
We can bound the first term in equation (\ref{eq:vdistance1prefix}) for any $j\in\mathcal{J}_n$ as 
\begin{align}
\left\lVert\rho \hat{V}_{\theta^{\ast,n}}(\cdot|j) - \Omega(\cdot)\right\rVert_V & \leq \left\lVert\frac{1}{L_n}\sum_{l=1}^{L_n}Q_{\theta^{\ast,n}}(\cdot|\psi_{jl}^n) - \Omega(\cdot)\right\rVert_V \label{eq:vdistanceaprefix} \\
&\quad+ \left\lVert \frac{1}{L_n}\sum_{l=1}^{L_n}\rho V_{\theta^{\ast,n}}^{n}(\cdot|\psi_{jl}^n)\mathds{1}_{\mathbf{\varepsilon}_1(\psi_{jl}^n)}(\cdot)(\mathds{1}_{\mathcal{Z}^{n}}(\cdot)- \mathds{1}_{\mathbf{\varepsilon}_2}(\cdot)) \right\rVert_V  \label{eq:vdistancebprefix}\\
& \quad + \left\lVert \frac{1}{L_n}\sum_{l=1}^{L_n}\rho V_{\theta^{\ast,n}}^{n}(\cdot|\psi_{jl}^n)(\mathds{1}_{\mathcal{Z}^{n}}(\cdot)- \mathds{1}_{\mathbf{\varepsilon}_1(\psi_{jl}^n)}(\cdot)) \right\rVert_V.\label{eq:vdistancecprefix}
\end{align}
In the following, we bound the right hand side of (\ref{eq:vdistanceaprefix}), and the terms in (\ref{eq:vdistancebprefix}), (\ref{eq:vdistancecprefix}), individually.

The right hand side of (\ref{eq:vdistanceaprefix}) can be bounded by the result of Lemma \ref{lem:vdistanceaprefix} to  
\begin{align*}
\left\lVert\frac{1}{L_n}\sum_{l=1}^{L_n}Q_{\theta^{\ast,n}}(\cdot|\psi_{jl}^n) - \Omega(\cdot)\right\rVert_V & = \sum_{z^{n}\in \mathcal{Z}^{n}}\left|\frac{1}{L_n}\sum_{l=1}^{L_n}Q_{\theta^{\ast,n}}(z^{n}|\psi_{jl}^n) - \Omega(z^{n})\right|\\
& \leq \sum_{z^{n}\in \mathcal{Z}^{n}}\epsilon_n\Omega(z^{n})\\
& \leq \epsilon_n
\end{align*}
For (\ref{eq:vdistancebprefix}), we obtain 
\begin{align*}
&\left\lVert \frac{1}{L_n}\sum_{l=1}^{L_n}\rho V_{\theta^{\ast,n}}^{n}(\cdot|\psi_{jl}^n)\mathds{1}_{\mathbf{\varepsilon}_1(\psi_{jl}^n)}(\cdot)(\mathds{1}_{\mathcal{Z}^{n}}(\cdot)- \mathds{1}_{\mathbf{\varepsilon}_2}(\cdot)) \right\rVert_V \\
& = \sum_{z^{n}\in \mathcal{Z}^{n}}\left|\frac{1}{L_n}\sum_{l=1}^{L_n}\rho V_{\theta^{\ast,n}}^{n}(z^{n}|\psi_{jl}^n)\mathds{1}_{\mathbf{\varepsilon}_1(\psi_{jl}^n)}(z^n)(\mathds{1}_{\mathcal{Z}^{n}}(z^{n})- \mathds{1}_{\mathbf{\varepsilon}_2}(z^{n})) \right|\\
& = \frac{1}{L_n}\sum_{l=1}^{L_n}\sum_{z^{n}\in \mathcal{Z}^{n}}\rho V_{\theta^{\ast,n}}^{n}(z^{n}|\psi_{jl}^n)\mathds{1}_{\mathbf{\varepsilon}_1(\psi_{jl}^n)}(z^{n})\mathds{1}_{\mathcal{Z}^{n}}(z^{n}) - \sum_{z^{n}\in \mathcal{Z}^{n}}\frac{1}{L_n}\sum_{l=1}^{L_n}\rho V_{\theta^{\ast,n}}^{n}(z^{n}|\psi_{jl}^n)\mathds{1}_{\mathbf{\varepsilon}_1(\psi_{jl}^n)}(z^{n})\mathds{1}_{\mathbf{\varepsilon}_2}(z^{n})\\
& \leq 1 - \sum_{z^{n}\in \mathcal{Z}^{n}}\frac{1}{L_n}\sum_{l=1}^{L_n}Q_{\theta^{\ast,n}}(z^{n}|\psi_{j,l}^n)\\
& \leq 1 - \sum_{z^{n}\in \mathcal{Z}^{n}}(1 - \epsilon_n)\Omega(z^{n})\\
& \leq 1 - (1 - \epsilon_n) (1 - 2\epsilon_n) \\
& \leq  3\epsilon_n - 2\epsilon_n^2\\
& \leq  3\epsilon_n.
\end{align*}
For (\ref{eq:vdistancecprefix}), we obtain
\begin{align*}
&\left\lVert \frac{1}{L_n}\sum_{l=1}^{L_n}\rho V_{\theta^{\ast,n}}^{n}(\cdot|\psi_{jl}^n)(\mathds{1}_{\mathcal{Z}^{n}}(\cdot)- \mathds{1}_{\mathbf{\varepsilon}_1(\psi_{jl}^n)}(\cdot)) \right\rVert_V  \overset{(a)}{=} \frac{1}{L_n}\sum_{l=1}^{L_n}\rho V_{\theta^{\ast,n}}^{n}(\mathbf{\varepsilon}^c_1(\psi_{jl}^n)|\psi_{jl}^n)\\
& \overset{(b)}{=} \frac{1}{L_n}\sum_{l\in\mathcal{L}_n}\rho V_{\theta^{\ast,n}}^{n}(\mathcal{T}_{\rho V_{\theta^{\ast,n},\delta}}^c(\psi_{jl}^n)|\psi_{jl}^n)\\
& \overset{(c)}{\leq} \frac{1}{L_n}\sum_{l\in\mathcal{L}_n}\exp\{-nc^\prime \delta^2\} \\
&\overset{(d)}{=} \epsilon_n.
\end{align*}
Here, $(a)$ follows by summing up only over $z^{n}\in\varepsilon_1^c(\cdot)$. 
$(b)$ follows by the definition of $\varepsilon_1(\psi_{jl}^n)$. 
$(c)$ follows since the probability of not obtaining a conditional typical $z^{n}$ can be upper bounded. 
$(d)$ follows since the upper bound in $(c)$ is valid for all $\psi_{jl}^n$.

Therefore, for (\ref{eq:vdistance1prefix}) we obtain 
\begin{align*}
|| \rho \bar{V}_{\theta^{\ast,n}}(Z^{n}) - \rho\hat{V}_{\theta^{\ast,n}}(Z^{n}|j)||_V &\leq 10 \epsilon_n\\
I(p_{J_n};E_u V_{\theta^{\ast,n}}^{n}) & \leq 10 n\epsilon_n \log\left(|\mathcal{Z}|\right)  - 10 \epsilon_n\log\left(10 \epsilon_n \right),
\end{align*}
which vanishes as $n$ goes to infinity because $\epsilon_n$ vanishes exponentially in $n$.
\subsection{Existence of codes fulfilling both the error and the secrecy requirement}
It remains to show that there exist codes fulfilling the error requirement and the secrecy requirement simultaneously. 

Therefore, we define the following event. 

\begin{align*}
\tilde{\iota} &:= \left\{\hat{\hat{e}}(\mathcal{K}_n^{\text{ran}}) \leq \lambda + \frac{1+\epsilon_2}{1-\epsilon_2}\left(\frac{\exp\{-n c^\prime \delta^\prime\}}{\lambda} + \frac{\exp{\left\{-n\left(\min_{W\in \widehat{\widehat{\mathcal{W}}}}I(p_{\Psi};\rho W)  - R -\nu \right)\right\}}}{\lambda} \right)\right\} \\
\hat{\iota} & := \iota_0 \cap \tilde{\iota}
\end{align*}

Here, we can apply the union bound and obtain 
\begin{align*}
Pr\{\hat{\iota}\} & = 1 - Pr\{\hat{\iota}^c\} = 1 - Pr\{\iota_0^c \cup \tilde{\iota}^c\}\geq 1 - Pr\{\iota_0^c\} - Pr\{\tilde{\iota}^c\},
\end{align*}
where both, $Pr\{\iota_0^c\}$ and $Pr\{\tilde{\iota}^c\}$ vanish super exponentially fast. 
Hence, there exist codes fulfilling the aforementioned criteria simultaneously. 
%
Finally, we get the achievable \CR{} assisted code secrecy rate as 
\begin{align*}
\widehat{\widehat{R}}_S^{ran}& \leq \max_{\Psi\leftrightarrow X \leftrightarrow (Y,Z)}\left(\min_{\theta\in \mathcal{P}(\mathcal{S}|\mathcal{X})} I(\Psi;Y_{\theta}) - \max_{\theta\in \mathcal{P}(\mathcal{S}|\mathcal{X})}I(\Psi;Z_{\theta})\right)\\
& = \max_{p_\Psi\rho(X|\Psi)}\left(\min_{W\in\widehat{\widehat{\mathcal{W}}}}I(p_{\Psi};\rho W) - \min_{V\in\widehat{\widehat{\mathcal{V}}}}I(P_{\Psi};\rho V)\right).
\end{align*}
%
\subsection{Converse}
What remains is to show the converse.  

We adopt the standard converse of the \WTC . 
As usual, we assumed strong secrecy in the achievability part and show in the converse, that even with weak secrecy the upper and lower bounds match.

Let $nR_L\geq\max_{u\in\mathcal{U}}I(J;Z_{\theta^\ast}^n|U=u)$.
We consider a sequence $(\mathcal{K}_n^{\text{ran}})_{n=1}^{\infty}$ of $(n,J_n,\mathcal{U}_n,p_U)$ wiretap codes for which $e(\mathcal{K}_n^{\text{ran}})=0$ and $R_L\leq\epsilon$ for an $\epsilon>0$, as $n\to \infty$.
\begin{align*}
nR_s & = H(J)\\
	 & \overset{(a)}{\leq} \min_{\theta\in \mathcal{P}(\mathcal{S}^n|\mathcal{X}^n)} I(J;Y_{\theta}^n|U) + 1 + \hat{\epsilon}H(J),\\
\rightarrow  nR_s & \leq \frac{1}{1-\hat{\epsilon}}	\left(\min_{\theta\in \mathcal{P}(\mathcal{S}^n|\mathcal{X}^n)} I(J;Y_{\theta}^n|U) - I(J;Z_{\theta^{\ast}}^n|U) + \max_{u\in\mathcal{U}}I(J;Z_{\theta^{\ast}}^n|U=u) + 1\right)\\
	 & \overset{(b)}{\leq} \frac{1}{1-\hat{\epsilon}}	\left(\min_{\theta\in \mathcal{P}(\mathcal{S}^n|\mathcal{X}^n)} I(J;Y_{\theta}^n|U) - I(J;Z_{\theta^{\ast}}^n|U) + nR_L + 1\right) \\
	 & \overset{(c)}{\leq} \frac{1}{1-\hat{\epsilon}}	\left(\min_{\theta\in \mathcal{P}(\mathcal{S}^n|\mathcal{X}^n)} I(J;Y_{\theta}^n|U) - I(J;Z_{\theta^{\ast}}^n|U) + n\epsilon +1 \right) \\
	 & \overset{(d)}{=} \frac{1}{1-\hat{\epsilon}}	\left(\min_{\theta\in \mathcal{P}(\mathcal{S}^n|\mathcal{X}^n)} I(J,U;Y_{\theta}^n|U) - I(J,U;Z_{\theta^{\ast}}^n|U) + n\epsilon +1 \right) \\
	 & \overset{(e)}{=} \frac{1}{1-\hat{\epsilon}}	\left(\min_{\theta\in \mathcal{P}(\mathcal{S}^n|\mathcal{X}^n)} I(\tilde{\Psi}^n;Y_{\theta}^n|U) - I(\tilde{\Psi}^n;Z_{\theta^{\ast}}^n|U) + n\epsilon +1 \right) \\
	 & \overset{(f)}{\leq} \frac{1}{1-\hat{\epsilon}}	\left(\max_{u\in\mathcal{U}}\left(\min_{\theta\in \mathcal{P}(\mathcal{S}^n|\mathcal{X}^n)} I(\tilde{\Psi}^n;Y_{\theta}^n|U=u) - I(\tilde{\Psi}^n;Z_{\theta^{\ast}}^n|U=u) + n\epsilon +1 \right)\right)\\
	 & \overset{(g)}{=} \frac{1}{1-\hat{\epsilon}}	\left(\min_{\theta\in \mathcal{P}(\mathcal{S}^n|\mathcal{X}^n)} I(\tilde{\Psi}^n;Y_{\theta}^n) - I(\tilde{\Psi}^n;Z_{\theta^{\ast}}^n) + n\epsilon +1 \right)\\
	 & \overset{(h)}{\leq}\frac{1}{1-\hat{\epsilon}}	\left(\min_{\theta^n\in \mathcal{P}^n(\mathcal{S}|\mathcal{X})} I(\tilde{\Psi}^n;Y_{\theta^n}^n) - I(\tilde{\Psi}^n;Z_{\theta^{\ast}}^n) + n\epsilon +1 \right)\\
	 & = \frac{1}{1-\hat{\epsilon}}	\left(\min_{\theta^n\in \mathcal{P}^n(\mathcal{S}|\mathcal{X})} \sum_{i=1}^n I(\tilde{\Psi}^n;Y_{i,\theta_i}|Y^{i-1}_{\theta^{i-1}}) - \sum_{i=1}^nI(\tilde{\Psi}^n;Z_{i,\theta_i^{\ast}}|Z^{n}_{i+1,\theta_{i+1}^{n,\ast}}) + n\epsilon +1 \right)\\
	 & = \frac{1}{1-\hat{\epsilon}}	\left(\min_{\theta^n\in \mathcal{P}^n(\mathcal{S}|\mathcal{X})} \sum_{i=1}^n \left(I(\tilde{\Psi}^n,Z^{n}_{i+1,\theta_{i+1}^{n,\ast}};Y_{i,\theta_i}|Y^{i-1}_{\theta^{i-1}}) - I(Z^{n}_{i+1,\theta_{i+1}^{n,\ast}};Y_{i,\theta_i}|\tilde{\Psi}^n,Y^{i-1}_{\theta^{i-1}})\right) \right. \\
	 &\left. \qquad \qquad \qquad - \sum_{i=1}^nI(\tilde{\Psi}^n;Z_{i,\theta_i^{\ast}}|Z^{i+1}_{\theta^{i+1,\ast}}) + n\epsilon +1 \right)\\
	  & = \frac{1}{1-\hat{\epsilon}}	\left(\min_{\theta^n\in \mathcal{P}^n(\mathcal{S}|\mathcal{X})} \sum_{i=1}^n \left(I(\tilde{\Psi}^n,Z^{n}_{i+1,\theta_{i+1}^{n,\ast}};Y_{i,\theta_i}|Y^{i-1}_{\theta^{i-1}}) - I(Z^{n}_{i+1,\theta_{i+1}^{n,\ast}};Y_{i,\theta_i}|\tilde{\Psi}^n,Y^{i-1}_{\theta^{i-1}})\right)\right.\\
	  &\qquad \qquad \left. - \sum_{i=1}^n\left(I(\tilde{\Psi}^n,Y^{i-1}_{\theta^{i-1}};Z_{i,\theta_i^{\ast}}|Z^{i+1}_{\theta^{i+1,\ast}}) + I(Y^{i-1}_{\theta^{i-1}};Z_{i,\theta_i^{\ast}}|\tilde{\Psi}^n,Z^{i+1}_{\theta^{i+1,\ast}})\right) + n\epsilon +1 \right)\\
	  & \overset{(i)}{=}  \frac{1}{1-\hat{\epsilon}}	\left(\min_{\theta^n\in \mathcal{P}^n(\mathcal{S}|\mathcal{X})} \sum_{i=1}^n I(\tilde{\Psi}^n,Z^{n}_{i+1,\theta_{i+1}^{n,\ast}};Y_{i,\theta_i}|Y^{i-1}_{\theta^{i-1}})  - \sum_{i=1}^nI(\tilde{\Psi}^n,Y^{i-1}_{\theta^{i-1}};Z_{i,\theta_i^{\ast}}|Z^{i+1}_{\theta^{i+1,\ast}})  + n\epsilon +1 \right)\\
	  & = \frac{1}{1-\hat{\epsilon}}	\left(\min_{\theta^n\in \mathcal{P}^n(\mathcal{S}|\mathcal{X})} \sum_{i=1}^n \left(I(Z^{n}_{i+1,\theta_{i+1}^{n,\ast}};Y_{i,\theta_i}|Y^{i-1}_{\theta^{i-1}}) +I(\tilde{\Psi}^n;Y_{i,\theta_i}|Y^{i-1}_{\theta^{i-1}},Z^{n}_{i+1,\theta_{i+1}^{n,\ast}})\right.\right.\\
	  &\qquad \qquad \left.\left.  - I(Y^{i-1}_{\theta^{i-1}};Z_{i,\theta_i^{\ast}}|Z^{i+1}_{\theta^{i+1,\ast}}) -I(\tilde{\Psi}^n;Z_{i,\theta_i^{\ast}}|Z^{i+1}_{\theta^{i+1,\ast}},Y^{i-1}_{\theta^{i-1}})\right) + n\epsilon +1 \right)\\
	  & \overset{(j)}{=}  \frac{1}{1-\hat{\epsilon}}	\left(\min_{\theta^n\in \mathcal{P}^n(\mathcal{S}|\mathcal{X})} \sum_{i=1}^n \left(I(\tilde{\Psi}^n;Y_{i,\theta_i}|Y^{i-1}_{\theta^{i-1}},Z^{n}_{i+1,\theta_{i+1}^{n,\ast}}) -I(\tilde{\Psi}^n;Z_{i,\theta_i^{\ast}}|Z^{i+1}_{\theta^{i+1,\ast}},Y^{i-1}_{\theta^{i-1}})\right) + n\epsilon +1 \right)\\
	   & \overset{(k)}{=}  \frac{1}{1-\hat{\epsilon}}	\left(\min_{\theta^n\in \mathcal{P}^n(\mathcal{S}|\mathcal{X})} \sum_{i=1}^n \left(I(\tilde{\Psi}^n;Y_{i,\theta_i}|V_i) -I(\tilde{\Psi}^n;Z_{i,\theta_i^{\ast}}|V_i)\right) + n\epsilon +1 \right)\\
	   & = \frac{1}{1-\hat{\epsilon}}	\left(\min_{\theta^n\in \mathcal{P}^n(\mathcal{S}|\mathcal{X})} \sum_{i=1}^n \left(I(\tilde{\Psi}^n,V_i;Y_{i,\theta_i}|V_i) -I(\tilde{\Psi}^n,V_i;Z_{i,\theta_i^{\ast}}|V_i)\right) + n\epsilon +1 \right)\\
	   & \overset{(l)}{=} \frac{1}{1-\hat{\epsilon}}	\left(\min_{\theta^n\in \mathcal{P}^n(\mathcal{S}|\mathcal{X})} \sum_{i=1}^n \left(I(\Psi^{\prime}_i;Y_{i,\theta_i}|V_i) -I(\Psi^{\prime}_i;Z_{i,\theta_i^{\ast}}|V_i)\right) + n\epsilon +1 \right)\\
	   & \overset{(m)}{=} \frac{1}{1-\hat{\epsilon}}	\left(\min_{\theta^n\in \mathcal{P}^n(\mathcal{S}|\mathcal{X})} n (I(\Psi^{\prime}_Q;Y_{Q,\theta_Q}|V_Q) -I(\Psi^{\prime}_Q;Z_{Q,\theta_Q^{\ast}}|V_Q,Q)) + n\epsilon +1 \right)\\
	   & = \frac{1}{1-\hat{\epsilon}}	\left(\min_{\theta^n\in \mathcal{P}^n(\mathcal{S}|\mathcal{X})} n (I(\Psi^{\prime};Y_{\theta}|V) -I(\Psi^{\prime};Z_{\theta^{\ast}}|V)) + n\epsilon +1 \right)\\
	   & \leq \frac{1}{1-\hat{\epsilon}}	\left(\min_{\theta^n\in \mathcal{P}^n(\mathcal{S}|\mathcal{X})} n \max_{V=v}(I(\Psi^{\prime};Y_{\theta}|V=v) -I(\Psi^{\prime};Z_{\theta^{\ast}}|V=v)) + n\epsilon +1 \right)\\
	 & = \frac{1}{1-\hat{\epsilon}} 	\left(n\min_{\theta\in \mathcal{P}(\mathcal{S}|\mathcal{X})} I(\Psi^{\prime};Y_{\theta}) - nI(\Psi^{\prime};Z_{\theta^{\ast}}) + n\epsilon +1 \right)\\
	 & \leq \frac{1}{1-\hat{\epsilon}}	\left(\max_{\Psi \leftrightarrow X\leftrightarrow(Y_\theta ,Z_{\theta^{\ast}})}\left(n\min_{\theta\in \mathcal{P}(\mathcal{S}|\mathcal{X})} I(\Psi^{\prime};Y_{\theta}) - nI(\Psi^{\prime};Z_{\theta^{\ast}})\right) + n\epsilon +1\right)\\
\Rightarrow R_s & \leq \frac{1}{1-\hat{\epsilon}}	\left(\max_{\Psi^{\prime} \leftrightarrow X\leftrightarrow(Y,Z)}\left(\min_{\theta\in \mathcal{P}(\mathcal{S}|\mathcal{X})} I(\Psi^{\prime};Y_{\theta}) - I(\Psi^{\prime};Z_{\theta^{\ast}})\right)  +\frac{1}{n} +\epsilon \right) 
\end{align*}

Here, $(a)$ follows by Fano's inequality, where $\hat{\epsilon}$ approaches zero as $n\to \infty$, $(b)$ follows by the definition of the leakage to the eavesdropper, $(c)$ follows because the leakage to the eavesdropper vanishes with $n$. 
Now, $(d)$ follows because $J$ and $U$ are independent, $(e)$ by defining $\tilde{\Psi} = (J,U)$, $(f)$ follows naturally. 
$(g)$ follows because $\tilde{\Psi} \leftrightarrow X^n\leftrightarrow (Y_{\theta}^n,Z_{\theta^{\ast}}^n)$ forms a conditional Markov chain, given $u\in \mathcal{U}$. 
To see this we evaluate the following term.
\begin{align*}
p_{\tilde{\Psi},X^n,Y_{\theta}^n,Z_{\theta^{\ast}}^n|U}(\cdot|u) & = p_{\tilde{\Psi}|U}(\cdot|u) p_{X^n|\tilde{\Psi},U}(\cdot|\cdot,u) p_{Y_{\theta}^n,Z_{\theta^{\ast}}^n|X^n,\tilde{\Psi},U}(\cdot|\cdot,u)\\
&\overset{(i)}{=} p_{\tilde{\Psi}|U}(\cdot|u) p_{X^n|\tilde{\Psi},U}(\cdot|\cdot,u) p_{Y_{\theta}^n,Z_{\theta^{\ast}}^n|X^n,U}(\cdot|\cdot,u)
\end{align*} 
$(i)$ follows because $X^n$ and $(Y^n_\theta,Z_{\theta^{\ast}}^n)$ are connected through a memoryless channel. 
Remember that when upper bounding the capacity, only the marginals are of interest. 
Then, we can invoke the same marginals property and can describe the input output relation between $X^n$ and $(Y^n_\theta,Z_{\theta^{\ast}}^n)$ by the channels $W^n_\theta(y^n|x^n)$, $V^n_{\theta^{\ast}}(z^n|x^n)$. 
Finally, $(h)$ follows since $\min_{\theta\in \mathcal{P}(\mathcal{S}^n|\mathcal{X}^n)} I(\tilde{\Psi}^n;Y_{\theta}^n)\leq \min_{\theta^n\in \mathcal{P}^n(\mathcal{S}|\mathcal{X})} I(\tilde{\Psi}^n;Y_{\theta^n}^n)$, with $\theta^n(s^n|x^n)=\prod_{i=1}^n\theta_i(s_i|x_i)$. 
$(i)$ and $(j)$ follow because of Csiszar's Sum Identity, $(k)$ follows by identifying $V_i=(Z^{i+1}_{\theta^{i+1,\ast}},Y^{i-1}_{\theta^{i-1}})$, $(l)$ by identifying $\Psi^{\prime}_i = (\tilde{\Psi}^n,V_i)$, and $(m)$ follows by introducing a uniformly distributed time sharing variable $Q$.

\section{Proof of Theorem \ref{th:degraded}}\label{proof:th-degraded}
~
\subsection{Achievability}
Since strongly degraded implies strongly less capable, we use the same approach as in \cite{Bloch2011a}.
We have 
\begin{align*}
I(X;Y_{\theta})&\geq I(X;Z_{\theta^\ast}),\\
I(\Psi;Y_{\theta}) & = I(\Psi,X;Y_{\theta}) - I(X;Y_{\theta}|\Psi)\\
& = I(X;Y_{\theta})+I(\Psi;Y_{\theta}|X)-I(X;Y_{\theta}|\Psi)\\
& =I(X;Y_{\theta})-I(X;Y_{\theta}|\Psi),\\
I(\Psi;Z_{\theta^\ast})& =  I(X;Z_{\theta^\ast})-I(X;Z_{\theta^\ast}|\Psi),\\
I(\Psi;Y_{\theta})- I(\Psi;Z_{\theta^\ast}) & =I(X;Y_{\theta})-I(X;Z_{\theta^\ast})+I(X;Z_{\theta^\ast}|\Psi)-I(X;Y_{\theta}|\Psi),
\end{align*}
where we can upper bound
\begin{align*}
I(X;Z_{\theta^\ast}|\Psi)-I(X;Y_{\theta}|\Psi) & \leq \max_{p_{\Psi X}}(I(X;Z_{\theta^\ast}|\Psi)-I(X;Y_{\theta}|\Psi))\\
& = \max_{p_{\Psi X}}\left(\sum_{\psi\in\Psi} p_{\Psi}(\psi)I(X;Z_{\theta^\ast}|\Psi=\psi)-I(X;Y_{\theta}|\Psi=\psi)\right)\\
& = \max_{p_X}(I(X;Z_{\theta^\ast})-I(X;Y_\theta))\\
& \leq 0.
\end{align*}
Hence, in total we obtain the following 
\begin{align*}
\max_{p_{\Psi},\rho_{X|\Psi}}(I(\Psi;Y_{\theta})- I(\Psi;Z_{\theta^\ast})) & \leq \max_{p_X}(I(X;Y_{\theta})-I(X;Z_{\theta^\ast})),
\end{align*}
with equality if we choose  $\Psi=X$ as the channel input.

\section{Nomenclature}\label{ap:Nomenclature}
\begin{longtable}{|p{.3\textwidth}|p{.66\textwidth}|}
\hline
Symbols & Meaning\\
\hhline{==}
$\log{(\cdot)}$ & Logarithm to base $2$, $\log_2{(\cdot)}$, unless stated otherwise.\\
\hline
$\exp{\{\cdot \}}$, $\exp_e{\{\cdot \}}$ &  $2^{\{\cdot \}}$, $e^{\{\cdot \}}$.\\
\hline
$X$ , $x$ & The random variable $X$ and its realization $x$.\\
\hline
$\mathcal{U}$ & The set $\mathcal{U}$, sets are denoted by calligraphic letters. \\
\hline
$|\mathcal{U}|$ & The cardinality of a set $\mathcal{U}$.\\
\hline
$\mathcal{P(U)}$ & The set of all probability measures on a set $\mathcal{U}$.\\
\hline
$p^n(x^n)$ & For $p\in \mathcal{P(U)}$ we define $p^n\in \mathcal{P}(\mathcal{U}^n)$ as $p^n(x^n)=\prod_i^n p(x_i)$.\\
\hline
$pW$ , $pW(y)$ & Induced output probability function by $p_X$ and the channel $W(y|x)$, $pW(y)=\sum_{x\in\mathcal{X}}p(x)W(y|x)$.\\
\hline
$H(X)$, $H(p_X)$ & Entropy of the \RV{} $X$, written in terms of the involved \RV{} or the involved probability function $p_X$.\\
\hline
$H(W|p)$ & The conditional Entropy of $Y$ given $X$, $H(W|p)= - \sum_{x,y} p(x)W(y|x) \log W(y|x)$.\\
\hline 
$I(p;W)$, $I(X;Y)$ & Mutual information between channel input and channel output, written in terms of the involved probability functions or the involved \RV{}.\\
\hline
$N(a|s^n)$ & Number of occurrences of the symbol $a$ in the sequence $s^n$.\\
\hline
$\mathcal{P}_0^n(\mathcal{S})$ & The set of all possible types of sequences of length $n$.\\
\hline 
$\mathcal{T}_{p,\delta}^n \subset \mathcal{X}^n$ & For a $p\in \mathcal{P}(\mathcal{X})$ and $\delta >0$, this denotes the $\delta$-typical set. \\
\hline
$\mathcal{T}_{W,\delta}^n(x^n) \subset \mathcal{Y}^n$ & For a $W\in \mathcal{P}(\mathcal{Y}|\mathcal{X})$ and a $\delta > 0$ this denotes the $\delta$-conditionally typical set, given the sequence $x^n$.\\
\hline
$\mathcal{J}_n$, $\mathcal{L}_n$ & Secure and confusing message sets.\\
\hline 
$\Psi_{j,l,u}$ , $\psi_{j,l,u}$ & Codeword (\RV{} and realization) for the messages $j\in\mathcal{J}_n$ and $l\in\mathcal{L}_n$ with \CR{} realization $u\in\mathcal{U}_n$.\\ 
\hline
$\mathcal{X}, \mathcal{S},\mathcal{Y}, \mathcal{Z}$ & Channel input set, channel state set, channel output set at Bob, channel output set at Eve. All are finite sets.\\
\hline 
$\rho^n(x^n|\psi_{j,l,u}^n)$ & Mapping from codeword to channel input.\\ 
\hline
$W^n(y^n|x^n,s^n)$, $V^n(z^n|x^n,s^n)$ & \DMCs{} from Alice to Bob and Alice to Eve, here $s^n$ is the channel state, $x^n$ is the channel input, and $y^n$ and $z^n$ are the received sequences at Bob and Eve, respectively.\\
\hline 
$\mathcal{U}_n$ & Common source of randomness, shared between Alice, Bob and Eve.\\
\hline
$\mathcal{W}=\{ (W_s:\mathcal{X}\to \mathcal{P}(\mathcal{Y})) : ~s\in \mathcal{S}\}$ & The family of channels to the legitimate receiver.\\
\hline
$\mathcal{V}=\{ (V_s:\mathcal{X}\to \mathcal{P}(\mathcal{Z})) : ~s\in \mathcal{S}\}$ & The family of channels to the illegitimate receiver.\\
\hline
$(\mathcal{W},\mathcal{V})$ &  The \AVWC .\\
\hline
$\mathcal{K}_n$ & An $(n,J_n)$ deterministic wiretap-code $\mathcal{K}_n$.\\
\hline
$E: \mathcal{J}_n\to \mathcal{P}(\mathcal{X}^n)$ & A stochastic encoder for an $(n,J_n)$ deterministic wiretap-code $\mathcal{K}_n$.\\
\hline
$\mathcal{D}_{j}$, $\mathcal{D}_{j,u}$, $\mathcal{D}_{jlu}$, $j\in \mathcal{J}_n$, $l\in \mathcal{L}_n$, $u\in \mathcal{U}_n$ & Mutually disjoint decoding sets for an $(n,J_n)$ deterministic wiretap-code $\mathcal{K}_n$, an $(n,J_n,\mathcal{U}_n,p_{U})$ \CR{} assisted wiretap code  $\mathcal{K}_n^{\text{ran}}$, and an $(n,J_n,\mathcal{U}_n,p_{U})$ \CR{} assisted wiretap code  $\mathcal{K}_n^{\text{ran}}$ with the requirement that confusing message should also be decoded at Bob.\\
\hline 
$EW_{s^n}^n:\mathcal{J}_n \to \mathcal{P}(\mathcal{Y}^n)$ & Channel from the secure messages to Bob, $EW_{s^n}^n(y^n|j) = \sum_{x^n\in\mathcal{X}^n}E(x^n|j)W^n(y^n|x^n,s^n)$.\\
\hline
$e(\mathcal{K}_n)$ & The maximum error probability for the \AVWC{} for an $(n,J_n)$ deterministic wiretap-code $\mathcal{K}_n$.\\
\hline
$\mathcal{F}:\mathcal{X}^n \to \mathcal{S}^n$ & Set of all deterministic functions, mapping from the channel inputs to the channel states. Equivalently the set of all deterministic jamming strategies.\\
\hline
$\hat{e}(\mathcal{K}_n)$ & Maximum error probability  of $(n,J_n)$ deterministic wiretap-code $\mathcal{K}_n$ for an \AVWC{} if the jammer has non-causal knowledge about the channel input $x^n$.\\
\hline 
$\mathcal{K}_n^{\text{ran}}$ & An $(n,J_n,\mathcal{U}_n,p_{U})$ \CR{} assisted wiretap code  $\mathcal{K}_n^{\text{ran}}$\\
\hline
$\mathcal{E}=\{ (E_u: \mathcal{J}_n\to \mathcal{P}(\mathcal{X}^n)):~ u\in \mathcal{U}_n\}$ & Family of stochastic encoders for an $(n,J_n,\mathcal{U}_n,p_{U})$ \CR{} assisted wiretap code  $\mathcal{K}_n^{\text{ran}}$.\\
\hline
$e(\mathcal{K}_n^{\text{ran}})$ & The maximum error probability of an $(n,J_n,\mathcal{U}_n,p_{U})$ \CR{} assisted wiretap code  $\mathcal{K}_n^{\text{ran}}$ averaged over all possible randomly chosen deterministic wiretap codebooks.\\
\hline
$\hat{e}(\mathcal{K}_n^{\text{ran}})$ & Maximum error probability of an $(n,J_n,\mathcal{U}_n,p_{U})$ \CR{} assisted wiretap code  $\mathcal{K}_n^{\text{ran}}$ averaged over all possible randomly chosen deterministic wiretap codebooks if the jammer has non-causal knowledge of the channel input $x^n$.\\
\hline
$\hat{\hat{e}}(\mathcal{K}_n^{\text{ran}})$ & Upper bound of $\hat{e}(\mathcal{K}_n^{\text{ran}})$, results in the consideration of the maxima with respect to $\mathcal{J}_n$, $\mathcal{L}_n$, $\Psi^n$, $\mathcal{T}_{\rho,\delta}^n(\psi^n)$ and $\mathcal{S}^{n}$.\\
\hline
$\mathcal{F}^\prime$ & The family of all deterministic mappings $\mathcal{J}_n\times \mathcal{X}^n \to \mathcal{S}^n $\\
\hline
$\mathcal{F}^{\prime \prime}$ & The family of all deterministic mappings $\mathcal{J}_n\to \mathcal{S}^n $\\
\hline
$R_S$ & An achievable \CR{} assisted secrecy rate for the \AVWC .\\
\hline
$\widehat{\widehat{R}}_S$ & An achievable \CR{} assisted secrecy rate for the \AVWC{} with non-causal knowledge of the channel input at the jammer.\\
\hline 
$\widehat{C}_S^{\text{ran}}(\mathcal{W},\mathcal{V})$ & The \CR{} assisted secrecy capacity of the \AVWC{} $(\mathcal{W},\mathcal{V})$ with maximum error probability criterion, when the jammer has not non-causal knowledge about the channel input (or only knows the messages).\\
\hline
$\widehat{C}_{S,av}^{\text{ran}}(\mathcal{W},\mathcal{V})$ & The \CR{} assisted secrecy capacity of the \AVWC{} $(\mathcal{W},\mathcal{V})$ with average error probability criterion, when the jammer has not non-causal knowledge about the channel input (or only knows the messages).\\
\hline
$\widehat{\widehat{C}}_S^{\text{ran}}(\mathcal{W},\mathcal{V})$ & The \CR{} assisted secrecy capacity of the \AVWC{} $(\mathcal{W},\mathcal{V})$ with maximum error probability criterion if the jammer has non-causal knowledge of the channel input.\\
\hline
$\mathcal{P}(\mathcal{S}^n|\mathcal{X}^n)$ & The set of all stochastic jamming strategies.\\
\hline 
$\widehat{\mathcal{W}}$ & Convex closure of $\mathcal{W}$.\\
\hline 
$\widehat{\widehat{\mathcal{W}}}$ & Row convex closure of $\mathcal{W}$.\\
\hline 
$\min_{W\in\widehat{\widehat{\mathcal{W}}}}I(p;W)$= $\min_{\theta\in\mathcal{P}(\mathcal{S}|\mathcal{X})}I(p;W_{\theta})$ & Worst case mutual information.\\
\hline 
$\theta^{\ast , n} \in \mathcal{P}^n(\mathcal{S}|\mathcal{X})$, $V_{\theta^{\ast,n}}^n$ & Best jamming strategy, leading to a best channel to the eavesdropper.\\
\hline 
$\pi (\cdot)$ & Permutation.\\
\hline 
$C_{(j,l)}$, $j\in \mathcal{J}_n$, $l\in \mathcal{L}_n$ & Disjoint subsets of the typical sequences $\mathcal{T}_{p,\delta}^n$ of size $|C_{(j,l)}|= \frac{|\mathcal{T}_{p,\delta}^n|}{|\mathcal{J}_n||\mathcal{L}_n|}$.\\
\hline 
$\hat{\chi}=\{\Psi_{ujl}^n: j\in\mathcal{J}_n, l\in\mathcal{L}_n, u\in \mathcal{U}_n\}$ &The family of \RV , representing random codewords. Also used as argument, when we use random coding arguments.\\
\hline 
$\mathcal{K}_n^{\text{ran}}(\hat{\chi})$ & Random $(n,J_n,\mathcal{U}_n,p_{U})$ \CR{} assisted code.\\
\hline
$\mathcal{U}(j,l,\psi^n,x^{n},\hat{\chi})$ & The set of all codebooks, for which the sequence $\psi^n$ is the codeword for the message pair $(j,l)$ and $x^{n}$ is the corresponding channel input. \\
\hline
$\mathcal{U}_0(j,l,\psi^n,x^{n},s^{n}, \hat{\chi})$ & The set of all codebooks, for which the sequence $\psi^n$ is the codeword for the message pair $(j,l)$, $x^{n}$ is the corresponding channel input, and the error bound $\lambda$ is not met.\\
\hline 
$B(u,j,l,\psi^n,x^{n},\hat{\chi})$ & Binary \RV , equals $1$ if $u\in\mathcal{U}(j,l,\psi^n,x^{n},\hat{\chi})$.\\
\hline
$\tilde{B}(j,l,\psi^n,x^{n},s^{n},u,\lambda,\hat{\chi})$ & Binary \RV , equals $1$ if $u\in\mathcal{U}_0(j,l,\psi^n,x^{n},s^{n}, \hat{\chi})$.\\
\hline
$\mathbf{\varepsilon}_1(\psi^n)$ & The set of typical output sequences $z^n$ for which the conditional probability of obtaining the sequence $z^n$ given the codeword $\psi^n$ can be upper bounded in terms of the conditional entropy of $Z_{\theta^\ast}$ given $\Psi$. \\
\hline
$\tilde{\Omega}(z^{n})$ & Expectation (with respect to the codeword $\Psi^n$) of the conditional probability of obtaining the sequence $z^n$ given the codeword $\Psi^n$. We consider only those summands in the expectation, for which the sequence $z^n$ is in the set $\mathbf{\varepsilon}_1(\psi^n)$.\\
\hline
$\mathbf{\varepsilon}_2$ & The set of typical output sequences $z^n$ for which$\tilde{\Omega}(z^{n})$ can be lower bounded in terms of the entropy of $Z_{\theta^{\ast}}$.\\
\hline
$\Omega(z^{n})$ & Equals $\tilde{\Omega}(z^{n})$, if $z^n$ is element of $\mathbf{\varepsilon}_2$, otherwise it equals zero. In other words, $\Omega(z^{n})$ equals the expectation (with respect to the codeword $\Psi^n$) of the conditional probability of obtaining the sequence $z^n$ given the codeword $\Psi^n$ under the condition that the conditional probability of obtaining the sequence $z^n$ given the codeword $\psi^n$ can be upper bounded in terms of the conditional entropy of $Z_{\theta^\ast}$ given $\Psi$, and that this expectation can be lower bounded terms of the entropy of $Z_{\theta^{\ast}}$.\\
\hline
$Q_{\theta^{\ast,n}}(z^{n}|\psi^n)$ & The conditional probability of the sequence $z^n$ given $\psi^n$, under the condition that the sequence $z^n$ belongs to $\mathbf{\varepsilon}_1(\psi^n)$ and $\mathbf{\varepsilon}_2$. Equals zero otherwise. \\
\hline
$\iota_1(j,z^{n})$ & Event that the expectation of $Q_{\theta^{\ast,n}}(z^{n}|\Psi^n_{jl})$ with respect to the confusing messages $L_n$ is in an $\epsilon_n$-region of its expected value, $\Omega(z^{n})$.\\
\hline
$\iota_0$ & Event that $\iota_1(j,z^{n})$ holds for all $j\in \mathcal{J}_n$, $z^{n} \in \mathcal{Z}^{n}$, and $u\in \mathcal{U}_n$. \\
\hline
$\tilde{\iota}$ & Event that a realization $\mathcal{K}_n^{\text{ran}}$ of a $\mathcal{K}_n^{\text{ran}}(\hat{\chi})$ fulfills the reliability constraint.\\
\hline
$\hat{\iota}$ & Event that a realization $\mathcal{K}_n^{\text{ran}}$ of a $\mathcal{K}_n^{\text{ran}}(\hat{\chi})$ fulfills the reliability and secrecy constraints, simultaneously.\\
\hline
$f_{(\cdot)}(\delta)$ & Function with $\lim_{\delta\to 0}f_{(\cdot)}(\delta)=0$.\\
\hline
\caption{Notation, Symbols and Meanings}
\label{tab:notation}
\end{longtable}
\bibliography{./Literature2}

\begin{thebibliography}{10}
\providecommand{\url}[1]{#1}
\csname url@samestyle\endcsname
\providecommand{\newblock}{\relax}
\providecommand{\bibinfo}[2]{#2}
\providecommand{\BIBentrySTDinterwordspacing}{\spaceskip=0pt\relax}
\providecommand{\BIBentryALTinterwordstretchfactor}{4}
\providecommand{\BIBentryALTinterwordspacing}{\spaceskip=\fontdimen2\font plus
\BIBentryALTinterwordstretchfactor\fontdimen3\font minus
  \fontdimen4\font\relax}
\providecommand{\BIBforeignlanguage}[2]{{%
\expandafter\ifx\csname l@#1\endcsname\relax
\typeout{** WARNING: IEEEtran.bst: No hyphenation pattern has been}%
\typeout{** loaded for the language `#1'. Using the pattern for}%
\typeout{** the default language instead.}%
\else
\language=\csname l@#1\endcsname
\fi
#2}}
\providecommand{\BIBdecl}{\relax}
\BIBdecl

\bibitem{Wyner1975}
A.~Wyner, ``{The Wire-Tap Channel},'' \emph{Bell Syst. Tech. J.}, vol.~54,
  no.~8, pp. 1355--1387, 1975.

\bibitem{Csiszar1978}
\BIBentryALTinterwordspacing
I.~Csiszar and J.~Korner, ``{Broadcast Channels with Confidential Messages},''
  \emph{IEEE Trans. Inf. Theory}, vol.~24, no.~3, pp. 339--348, May 1978.
  [Online]. Available: \url{http://ieeexplore.ieee.org/document/1055892/}
\BIBentrySTDinterwordspacing

\bibitem{Cheong1978}
\BIBentryALTinterwordspacing
S.~Leung-Yan-Cheong and M.~Hellman, ``{The Gaussian Wire-Tap Channel},''
  \emph{IEEE Trans. Inf. Theory}, vol.~24, no.~4, pp. 451--456, Jul 1978.
  [Online]. Available: \url{http://ieeexplore.ieee.org/document/1055917/}
\BIBentrySTDinterwordspacing

\bibitem{Ozarow1984}
L.~H. Ozarow and A.~D. Wyner, ``{Wire‐Tap Channel II},'' \emph{AT{\&}T Bell
  Lab. Tech. J.}, vol.~63, no.~10, pp. 2135--2157, 1984.

\bibitem{Goldfeld2016a}
\BIBentryALTinterwordspacing
Z.~Goldfeld, P.~Cuff, and H.~H. Permuter, ``{Semantic-Security Capacity for
  Wiretap Channels of Type II},'' \emph{IEEE Trans. Inf. Theory}, vol.~62,
  no.~7, pp. 3863--3879, Jul 2016. [Online]. Available:
  \url{http://ieeexplore.ieee.org/document/7467522/}
\BIBentrySTDinterwordspacing

\bibitem{Nafea2018}
\BIBentryALTinterwordspacing
M.~Nafea and A.~Yener, ``{A New Wiretap Channel Model and Its Strong Secrecy
  Capacity},'' \emph{IEEE Trans. Inf. Theory}, vol.~64, no.~3, pp. 2077--2092,
  Mar 2018. [Online]. Available:
  \url{http://ieeexplore.ieee.org/document/8234688/}
\BIBentrySTDinterwordspacing

\bibitem{Blackwell1960}
\BIBentryALTinterwordspacing
D.~Blackwell, L.~Breiman, and A.~J. Thomasian, ``{The Capacities of Certain
  Channel Classes Under Random Coding},'' \emph{Ann. Math. Stat.}, vol.~31,
  no.~3, pp. 558--567, Sep 1960. [Online]. Available:
  \url{http://www.jstor.org/stable/2237566}
\BIBentrySTDinterwordspacing

\bibitem{Ahlswede1970a}
\BIBentryALTinterwordspacing
R.~Ahlswede, ``{A Note on the Existence of the Weak Capacity for Channels with
  Arbitrarily Varying Channel Probability Functions and Its Relation to
  Shannon's Zero Error Capacity},'' \emph{Ann. Math. Stat.}, vol.~41, no.~3,
  pp. 1027--1033, Jun 1970. [Online]. Available:
  \url{http://projecteuclid.org/euclid.aoms/1177696979
  https://pub.uni-bielefeld.de/record/1781084}
\BIBentrySTDinterwordspacing

\bibitem{Shannon1956}
\BIBentryALTinterwordspacing
C.~Shannon, ``{The Zero Error Capacity of a Noisy Channel},'' \emph{IEEE Trans.
  Inf. Theory}, vol.~2, no.~3, pp. 8--19, Sep 1956. [Online]. Available:
  \url{http://ieeexplore.ieee.org/document/1056798/}
\BIBentrySTDinterwordspacing

\bibitem{Ahlswede1973}
R.~Ahlswede, ``{Channels with Arbitrarily Varying Channel Probability Functions
  in the Presence of Noiseless Feedback},'' \emph{Zeitschrift f{\"{u}}r
  Wahrscheinlichkeitstheorie und Verwandte Gebiete}, vol.~25, no.~3, pp.
  239--252, Sep 1973.

\bibitem{Ahlswede1973a}
\BIBentryALTinterwordspacing
------, ``{A Constructive Proof of the Coding Theorem for Discrete Memoryless
  Channels with Feedback},'' \emph{Trans. Sixth Prague Conf. Inf. Theory, Stat.
  Decis. Funct. Random Process.}, pp. 39--50, 1973. [Online]. Available:
  \url{https://pub.uni-bielefeld.de/record/1780373}
\BIBentrySTDinterwordspacing

\bibitem{Ahlswede1978}
------, ``{Elimination of Correlation in Random Codes for Arbitrarily Varying
  Channels},'' \emph{Z. Wahrsch. Verw. Gebiete}, vol.~44, pp. 159--175, 1978.

\bibitem{Csiszr1981}
\BIBentryALTinterwordspacing
I.~Csiszar and J.~K{\"{o}}rner, ``{On the Capacity of the Arbitrarily Varying
  Channel for Maximum Probability of Error},'' \emph{Zeitschrift f{\"{u}}r
  Wahrscheinlichkeitstheorie und Verwandte Gebiete}, vol.~57, no.~1, pp.
  87--101, 1981. [Online]. Available:
  \url{http://link.springer.com/10.1007/BF00533715}
\BIBentrySTDinterwordspacing

\bibitem{Jahn1981}
J.~H. Jahn, ``{Coding of Arbitrarily Varying Multiuser Channels},'' \emph{IEEE
  Trans. Inf. Theory}, vol.~27, no.~2, pp. 212--226, Mar 1981.

\bibitem{Gubner1990}
\BIBentryALTinterwordspacing
J.~Gubner, ``{On the Deterministic-Code Capacity of the Multiple-Access
  Arbitrarily Varying Channel},'' \emph{IEEE Trans. Inf. Theory}, vol.~36,
  no.~2, pp. 262--275, Mar 1990. [Online]. Available:
  \url{http://ieeexplore.ieee.org/document/52472/}
\BIBentrySTDinterwordspacing

\bibitem{Ericson1985}
\BIBentryALTinterwordspacing
T.~Ericson, ``{Exponential Error Bounds for Random Codes in the Arbitrarily
  Varying Channel},'' \emph{IEEE Trans. Inf. Theory}, vol.~31, no.~1, pp.
  42--48, Jan 1985. [Online]. Available:
  \url{http://ieeexplore.ieee.org/document/1056995/}
\BIBentrySTDinterwordspacing

\bibitem{Ahlswede1991}
R.~Ahlswede and N.~Cai, ``{Two Proofs of Pinsker's Conjecture Concerning
  Arbitrarily Varying Channels},'' \emph{IEEE Trans. Inf. Theory}, vol.~37,
  no.~6, pp. 1647--1649, Nov 1991.

\bibitem{Hughes1997}
B.~L. Hughes, ``{The Smallest List for the Arbitrarily Varying Channel},''
  \emph{IEEE Trans. Inf. Theory}, vol.~43, no.~3, pp. 803--815, May 1997.

\bibitem{Gohari2008}
A.~A. Gohari and V.~Anantharam, ``{An Outer Bound to the Admissible Source
  Region of Broadcast Channels with Arbitrarily Correlated Sources and Channel
  Variations},'' \emph{46th Annu. Allert. Conf. Commun. Control. Comput.}, pp.
  301--308, Sep 2008.

\bibitem{Liang2008}
Y.~Liang, G.~Kramer, and S.~Shamai, ``{Capacity Outer Bounds for Broadcast
  Channels},'' \emph{2008 IEEE Inf. Theory Work. ITW}, pp. 2--4, May 2008.

\bibitem{Wyrembelski2009a}
R.~F. Wyrembelski, I.~Bjelakov{\'{i}}c, and H.~Boche, ``{Coding Strategies for
  Bidirectional Relaying for Arbitrarily Varying Channels},'' \emph{GLOBECOM -
  IEEE Glob. Telecommun. Conf.}, Dec 2009.

\bibitem{Wyrembelski2009}
R.~F. Wyrembelski, I.~Bjelakovi{\'{c}}, and H.~Boche, ``{On the Capacity of
  Bidirectional Relaying with Unknown Varying Channels},'' \emph{CAMSAP 2009 -
  2009 3rd IEEE Int. Work. Comput. Adv. Multi-Sensor Adapt. Process.}, pp.
  269--272, Dec 2009.

\bibitem{Nitinawarat2013}
S.~Nitinawarat, ``{On the Deterministic Code Capacity Region of an Arbitrarily
  Varying Multiple-Access Channel under List Decoding},'' \emph{IEEE Trans.
  Inf. Theory}, vol.~59, no.~5, pp. 2683--2693, Jan 2013.

\bibitem{Schaefer2014}
R.~F. Schaefer and H.~Boche, ``{How much Coordination is needed for Robust
  Broadcasting over Arbitrarily Varying Bidirectional Broadcast Channels},''
  \emph{2014 IEEE Int. Conf. Commun. ICC 2014}, pp. 1872--1877, Jun 2014.

\bibitem{Boche2014}
\BIBentryALTinterwordspacing
H.~Boche and J.~N{\"{o}}tzel, ``{Positivity, Discontinuity, Finite Resources,
  and Nonzero Error for Arbitrarily Varying Quantum Channels},'' \emph{J. Math.
  Phys.}, vol.~55, no.~12, pp. 541--545, Dec 2014. [Online]. Available:
  \url{http://aip.scitation.org/doi/10.1063/1.4902930}
\BIBentrySTDinterwordspacing

\bibitem{Cai2016}
N.~Cai, ``{List Decoding for Arbitrarily Varying Multiple Access Channel
  Revisited: List Configuration and Symmetrizability},'' \emph{IEEE Trans. Inf.
  Theory}, vol.~62, no.~11, Sep 2016.

\bibitem{Kiefer1962}
J.~Kiefer and J.~Wolfowitz, ``{Channels with Arbitrarily Varying Channel
  Probability Functions},'' \emph{Inf. Control}, vol.~5, no.~1, pp. 44--54,
  1962.

\bibitem{Ahlswede1969}
R.~Ahlswede and J.~Wolfowitz, ``{Correlated Decoding for Channels with
  Arbitrarily Varying Channel Probability Functions},'' \emph{Inf. Control},
  vol.~14, no.~5, pp. 457--473, 1969.

\bibitem{Ahlswede1970b}
------, ``{The Capacity of a Channel with Arbitrarily Varying Channel
  Probability Functions and Binary Output Alphabet},'' \emph{Zeitschrift
  f{\"{u}}r Wahrscheinlichkeitstheorie und Verwandte Gebiete}, vol.~15, no.~3,
  pp. 186--194, Sep 1970.

\bibitem{Ahlswede1986}
\BIBentryALTinterwordspacing
R.~Ahlswede, ``{Arbitrarily Varying Channels with States Sequence Known to the
  Sender},'' \emph{IEEE Trans. Inf. Theor.}, vol.~32, no.~5, pp. 621--629, Sep
  1986. [Online]. Available:
  \url{http://dblp.uni-trier.de/db/journals/tit/tit32.html{\#}Ahlswede86a}
\BIBentrySTDinterwordspacing

\bibitem{Ahlswede1991a}
R.~Ahlswede and G.~Simonyi, ``{Reusable Memories in the Light of the Old
  Arbitrarily Varying and a New Outputwise Varying Channel Theory},''
  \emph{IEEE Trans. Inf. Theory}, vol.~37, no.~4, pp. 1143--1150, Jul 1991.

\bibitem{Ahlswede1999}
\BIBentryALTinterwordspacing
R.~Ahlswede and {Ning Cai}, ``{Arbitrarily Varying Multiple-Access Channels.
  II. Correlated Senders' Side Information, Correlated Messages, and Ambiguous
  Transmission},'' \emph{IEEE Trans. Inf. Theory}, vol.~45, no.~2, pp.
  749--756, Mar 1999. [Online]. Available:
  \url{http://ieeexplore.ieee.org/document/749025/}
\BIBentrySTDinterwordspacing

\bibitem{Winshtok2006}
A.~Winshtok and Y.~Steinberg, ``{The Arbitrarily Varying Degraded Broadcast
  Channel with States Known at the Encoder},'' \emph{IEEE Int. Symp. Inf.
  Theory - Proc.}, pp. 2156--2160, Jul 2006.

\bibitem{Sarwate2007}
\BIBentryALTinterwordspacing
A.~D. Sarwate and M.~Gastpar, ``{Channels with Nosy "Noise"},'' in \emph{2007
  IEEE Int. Symp. Inf. Theory}.\hskip 1em plus 0.5em minus 0.4em\relax IEEE,
  Jun 2007, pp. 996--1000. [Online]. Available:
  \url{http://ieeexplore.ieee.org/document/4557354/}
\BIBentrySTDinterwordspacing

\bibitem{Cai2010}
N.~Cai, T.~Chan, and A.~Grant, ``{The arbitrarily varying channel when the
  jammer knows the channel input},'' \emph{IEEE Int. Symp. Inf. Theory -
  Proc.}, pp. 295--299, Jun 2010.

\bibitem{Sarwate2010}
\BIBentryALTinterwordspacing
A.~D. Sarwate, ``{Coding against Myopic Adversaries},'' in \emph{2010 IEEE Inf.
  Theory Work.}\hskip 1em plus 0.5em minus 0.4em\relax IEEE, Aug 2010, pp.
  1--5. [Online]. Available: \url{http://ieeexplore.ieee.org/document/5592896/}
\BIBentrySTDinterwordspacing

\bibitem{Wiese2013}
M.~Wiese and H.~Boche, ``{The Arbitrarily Varying Multiple-Access Channel with
  Conferencing Encoders},'' \emph{IEEE Trans. Inf. Theory}, vol.~59, no.~3, pp.
  1405--1416, Mar 2013.

\bibitem{Dey2013}
\BIBentryALTinterwordspacing
B.~K. Dey, S.~Jaggi, M.~Langberg, and A.~D. Sarwate, ``{Upper Bounds on the
  Capacity of Binary Channels with Causal Adversaries},'' \emph{IEEE Trans.
  Inf. Theory}, vol.~59, no.~6, pp. 3753--3763, Jun 2013. [Online]. Available:
  \url{http://ieeexplore.ieee.org/document/6516725/}
\BIBentrySTDinterwordspacing

\bibitem{Boche2014c}
H.~Boche and R.~F. Schaefer, ``{List Decoding for Arbitrarily Varying Multiple
  Access Channels with Conferencing Encoders},'' \emph{2014 IEEE Int. Conf.
  Commun. ICC 2014}, pp. 1934--1940, Jun 2014.

\bibitem{Schaefer2014a}
R.~F. Schaefer and H.~Boche, ``{List Decoding for Arbitrarily Varying Broadcast
  Channels with Receiver Side Information},'' \emph{IEEE Trans. Inf. Theory},
  vol.~60, no.~8, pp. 4472--4487, May 2014.

\bibitem{Pereg2017}
U.~Pereg and Y.~Steinberg, ``{The Arbitrarily Varying Degraded Broadcast
  Channel with Causal Side Information at the Encoder},'' \emph{IEEE Int. Symp.
  Inf. Theory - Proc.}, pp. 1033--1037, Aug 2017.

\bibitem{Budkuley2017}
\BIBentryALTinterwordspacing
A.~J. Budkuley, B.~K. Dey, and V.~M. Prabhakaran, ``{Communication in the
  Presence of a State-Aware Adversary},'' \emph{IEEE Trans. Inf. Theory},
  vol.~63, no.~11, pp. 7396--7419, Nov 2017. [Online]. Available:
  \url{http://ieeexplore.ieee.org/document/8039279/}
\BIBentrySTDinterwordspacing

\bibitem{Boche2019a}
H.~Boche, M.~Cai, and N.~Cai, ``{Message Transmission over Classical Quantum
  Channels with a Jammer with Side Information: Message Transmission Capacity
  and Resources},'' \emph{IEEE Trans. Inf. Theory}, vol.~65, no.~5, pp.
  2922--2943, May 2019.

\bibitem{Beemer2020}
\BIBentryALTinterwordspacing
A.~Beemer, O.~Kosut, J.~Kliewer, E.~Graves, and P.~Yu, ``{Authentication
  Against a Myopic Adversary},'' in \emph{2019 IEEE Conf. Commun. Netw.
  Secur.}\hskip 1em plus 0.5em minus 0.4em\relax IEEE, Jun 2019, pp. 1--5.
  [Online]. Available: \url{https://ieeexplore.ieee.org/document/8802705/}
\BIBentrySTDinterwordspacing

\bibitem{Ahlswede1971}
R.~Ahlswede, ``{The Capacity of a Channel with Arbitrarily Varying Additive
  Gaussian Channel Probability Functions},'' in \emph{Sixth Prague Conf. Inf.
  Th., Stat. Dec. Fct's Rand. Proc.}\hskip 1em plus 0.5em minus 0.4em\relax
  House Czechosl. Academy of Sc, 1971.

\bibitem{Hughes1987}
\BIBentryALTinterwordspacing
B.~Hughes and P.~Narayan, ``{Gaussian Arbitrarily Varying Channels},''
  \emph{IEEE Trans. Inf. Theory}, vol.~33, no.~2, pp. 267--284, Mar 1987.
  [Online]. Available: \url{http://ieeexplore.ieee.org/document/1057288/}
\BIBentrySTDinterwordspacing

\bibitem{Csiszar1988a}
\BIBentryALTinterwordspacing
I.~Csiszar and P.~Narayan, ``{Arbitrarily Varying Channels with Constrained
  Inputs and States},'' \emph{IEEE Trans. Inf. Theory}, vol.~34, no.~1, pp.
  27--34, Jan 1988. [Online]. Available:
  \url{http://dblp.uni-trier.de/db/journals/tit/tit34.html{\#}CsiszarN88
  https://ieeexplore.ieee.org/document/2598/}
\BIBentrySTDinterwordspacing

\bibitem{Csiszar1988}
\BIBentryALTinterwordspacing
------, ``{The Capacity of the Arbitrarily Varying Channel Revisited:
  Positivity, Constraints},'' \emph{IEEE Trans. Inf. Theory}, vol.~34, no.~2,
  pp. 181--193, Mar 1988. [Online]. Available:
  \url{http://ieeexplore.ieee.org/document/2627/}
\BIBentrySTDinterwordspacing

\bibitem{Csiszar1991}
\BIBentryALTinterwordspacing
------, ``{Capacity of the Gaussian Arbitrarily Varying Channel},'' \emph{IEEE
  Trans. Inf. Theory}, vol.~37, no.~1, pp. 18--26, Jan 1991. [Online].
  Available: \url{http://ieeexplore.ieee.org/document/61125/}
\BIBentrySTDinterwordspacing

\bibitem{Gubner1991a}
\BIBentryALTinterwordspacing
J.~Gubner, ``{State Constraints for the Multiple-Access Arbitrarily Varying
  Channel},'' \emph{IEEE Trans. Inf. Theory}, vol.~37, no.~1, pp. 27--35, Jan
  1991. [Online]. Available: \url{http://ieeexplore.ieee.org/document/61126/}
\BIBentrySTDinterwordspacing

\bibitem{Gubner1992}
\BIBentryALTinterwordspacing
------, ``{On the Capacity Region of the Discrete Additive Multiple-Access
  Arbitrarily Varying Channel},'' \emph{IEEE Trans. Inf. Theory}, vol.~38,
  no.~4, pp. 1344--1347, Jul 1992. [Online]. Available:
  \url{http://ieeexplore.ieee.org/document/144713/}
\BIBentrySTDinterwordspacing

\bibitem{Gubner}
\BIBentryALTinterwordspacing
J.~Gubner and B.~Hughes, ``{Nonconvexity of the Capacity Region of the
  Multiple-Access Arbitrarily Varying Channel Subject to Constraints},'' in
  \emph{Proc. 1994 IEEE Int. Symp. Inf. Theory}.\hskip 1em plus 0.5em minus
  0.4em\relax IEEE, Jul 1994, p.~53. [Online]. Available:
  \url{http://ieeexplore.ieee.org/document/394917/}
\BIBentrySTDinterwordspacing

\bibitem{Bross2003}
S.~I. Bross and S.~Shamai, ``{Capacity and Decoding Rules for the Poisson
  Arbitrarily Varying Channel},'' \emph{IEEE Trans. Inf. Theory}, vol.~49,
  no.~11, pp. 3076--3093, Nov 2003.

\bibitem{Hof2006}
E.~Hof and S.~I. Bross, ``{On the Deterministic-Code Capacity of the Two-User
  Discrete Memoryless Arbitrarily Varying General Broadcast Channel with
  Degraded Message Sets},'' \emph{IEEE Trans. Inf. Theory}, vol.~52, no.~11,
  pp. 5023--5044, Nov 2006.

\bibitem{Wyrembelski2010}
R.~F. Wyrembelski, I.~Bjelakovi{\'{c}}, and H.~Boche, ``{On Arbitrarily Varying
  Bidirectional Broadcast Channels with Constraints on Input and States},''
  \emph{ISITA/ISSSTA 2010 - 2010 Int. Symp. Inf. Theory Its Appl.}, pp.
  410--415, Oct 2010.

\bibitem{Sarwate2012}
\BIBentryALTinterwordspacing
A.~D. Sarwate and M.~Gastpar, ``{List-Decoding for the Arbitrarily Varying
  Channel Under State Constraints},'' \emph{IEEE Trans. Inf. Theory}, vol.~58,
  no.~3, pp. 1372--1384, Mar 2012. [Online]. Available:
  \url{http://ieeexplore.ieee.org/document/6157083/}
\BIBentrySTDinterwordspacing

\bibitem{Sarwate2012a}
\BIBentryALTinterwordspacing
A.~D. Sarwate, ``{An AVC Perspective on Correlated Jamming},'' in \emph{2012
  Int. Conf. Signal Process. Commun.}\hskip 1em plus 0.5em minus 0.4em\relax
  IEEE, Jul 2012, pp. 1--5. [Online]. Available:
  \url{https://ieeexplore.ieee.org/document/6290241}
\BIBentrySTDinterwordspacing

\bibitem{Mirmohseni2014}
\BIBentryALTinterwordspacing
M.~Mirmohseni and P.~Papadimitratos, ``{Active Adversaries from an
  Information-Theoretic Perspective: Data Modification Attacks},'' \emph{2014
  IEEE Int. Symp. Inf. Theory}, pp. 791--795, Jun 2014. [Online]. Available:
  \url{https://ieeexplore.ieee.org/document/6874941}
\BIBentrySTDinterwordspacing

\bibitem{Sarwate2018}
\BIBentryALTinterwordspacing
Y.~Zhang, S.~Vatedka, S.~Jaggi, and A.~D. Sarwate, ``{Quadratically Constrained
  Myopic Adversarial Channels},'' in \emph{2018 IEEE Int. Symp. Inf.
  Theory}.\hskip 1em plus 0.5em minus 0.4em\relax IEEE, Jun 2018, pp. 611--615.
  [Online]. Available: \url{https://ieeexplore.ieee.org/document/8437457/}
\BIBentrySTDinterwordspacing

\bibitem{Hosseinigoki2018}
F.~Hosseinigoki and O.~Kosut, ``{Capacity of the Gaussian Arbitrarily-Varying
  Channel with List Decoding},'' \emph{IEEE Int. Symp. Inf. Theory - Proc.},
  vol. 2018-June, pp. 471--475, Aug 2018.

\bibitem{Pereg2019}
U.~Pereg and Y.~Steinberg, ``{The Arbitrarily Varying Channel under Constraints
  with Side Information at the Encoder},'' \emph{IEEE Trans. Inf. Theory},
  vol.~65, no.~2, pp. 861--887, Feb 2019.

\bibitem{MolavianJazi2009}
\BIBentryALTinterwordspacing
E.~MolavianJazi, M.~Bloch, and J.~N. Laneman, ``{Arbitrary Jamming can Preclude
  Secure Communication},'' in \emph{2009 47th Annu. Allert. Conf. Commun.
  Control. Comput.}\hskip 1em plus 0.5em minus 0.4em\relax IEEE, Sep 2009, pp.
  1069--1075. [Online]. Available:
  \url{http://ieeexplore.ieee.org/document/5394876/}
\BIBentrySTDinterwordspacing

\bibitem{Bjelakovic2013}
\BIBentryALTinterwordspacing
I.~Bjelakovi{\'{c}}, H.~Boche, and J.~Sommerfeld, ``{Capacity Results for
  Arbitrarily Varying Wiretap Channels},'' in \emph{Inf. Theory, Comb. Search
  Theory}, ser. Lecture Notes in Computer Science, H.~Aydinian, F.~Cicalese,
  and C.~Deppe, Eds.\hskip 1em plus 0.5em minus 0.4em\relax Springer Berlin
  Heidelberg, 2013, vol. 7777, pp. 123--144. [Online]. Available:
  \url{http://dx.doi.org/10.1007/978-3-642-36899-8{\_}5
  http://link.springer.com/10.1007/978-3-642-36899-8{\_}5}
\BIBentrySTDinterwordspacing

\bibitem{Boche2015}
\BIBentryALTinterwordspacing
H.~Boche, R.~F. Schaefer, and H.~V. Poor, ``{On the Continuity of the Secrecy
  Capacity of Compound and Arbitrarily Varying Wiretap Channels},'' Dec 2015.
  [Online]. Available: \url{http://ieeexplore.ieee.org/document/7182343/}
\BIBentrySTDinterwordspacing

\bibitem{Noetzel2016}
\BIBentryALTinterwordspacing
J.~Nötzel, M.~Wiese, and H.~Boche, ``{The Arbitrarily Varying Wiretap
  Channel-Secret Randomness, Stability, and Super-Activation},'' \emph{IEEE
  Trans. Inf. Theory}, vol.~62, no.~6, pp. 3504--3531, Jun 2016. [Online].
  Available: \url{http://ieeexplore.ieee.org/document/7447794/}
\BIBentrySTDinterwordspacing

\bibitem{Wiese2016}
\BIBentryALTinterwordspacing
M.~Wiese, J.~Nötzel, and H.~Boche, ``{A Channel under Simultaneous Jamming and
  Eavesdropping Attack-Correlated Random Coding Capacities under Strong Secrecy
  Criteria},'' \emph{IEEE Trans. Inf. Theory}, vol.~62, no.~7, pp. 3844--3862,
  Jul 2016. [Online]. Available:
  \url{http://ieeexplore.ieee.org/document/7467557/}
\BIBentrySTDinterwordspacing

\bibitem{Chen2021a}
\BIBentryALTinterwordspacing
Y.~Chen, D.~He, and Y.~Luo, ``{Strong Secrecy of Arbitrarily Varying Multiple
  Access Channels},'' \emph{IEEE Trans. Inf. Forensics Secur.}, vol.~16, pp.
  3662--3677, Jun 2021. [Online]. Available:
  \url{https://ieeexplore.ieee.org/document/9448110/}
\BIBentrySTDinterwordspacing

\bibitem{Aggarwal2009}
V.~Aggarwal, L.~Lai, A.~R. Calderbank, and H.~V. Poor, ``{Wiretap channel type
  II with an active eavesdropper},'' \emph{IEEE Int. Symp. Inf. Theory -
  Proc.}, pp. 1944--1948, Dec 2009.

\bibitem{Boche2012}
H.~Boche and R.~F. Wyrembelski, ``{Comparison of Different Attack Classes in
  Arbitrarily Varying Wiretap Channels},'' \emph{WIFS 2012 - Proc. 2012 IEEE
  Int. Work. Inf. Forensics Secur.}, pp. 270--275, Dec 2012.

\bibitem{Bjelakovic2013a}
\BIBentryALTinterwordspacing
I.~Bjelakovi{\'{c}}, H.~Boche, and J.~Sommerfeld, ``{Secrecy Results for
  Compound Wiretap Channels},'' \emph{Probl. Inf. Transm.}, vol.~49, no.~1, pp.
  73--98, Jan 2013. [Online]. Available:
  \url{http://dx.doi.org/10.1134/S0032946013010079
  http://link.springer.com/10.1134/S0032946013010079}
\BIBentrySTDinterwordspacing

\bibitem{Boche2013b}
\BIBentryALTinterwordspacing
H.~Boche and R.~F. Schaefer, ``{Capacity Results, Coordination Resources, and
  Super-Activation in Wiretap Channels},'' in \emph{2013 IEEE Int. Symp. Inf.
  Theory}.\hskip 1em plus 0.5em minus 0.4em\relax IEEE, Jul 2013, pp.
  1342--1346. [Online]. Available:
  \url{http://ieeexplore.ieee.org/document/6620445/}
\BIBentrySTDinterwordspacing

\bibitem{S.Mansour2019}
\BIBentryALTinterwordspacing
A.~{S. Mansour}, H.~Boche, and R.~{F. Schaefer}, ``{The Secrecy Capacity of the
  Arbitrarily Varying Wiretap Channel under List Decoding},'' \emph{Adv. Math.
  Commun.}, vol.~13, no.~1, pp. 11--39, 2019. [Online]. Available:
  \url{http://aimsciences.org//article/doi/10.3934/amc.2019002}
\BIBentrySTDinterwordspacing

\bibitem{Goldfeld2020}
\BIBentryALTinterwordspacing
Z.~Goldfeld, P.~Cuff, and H.~H. Permuter, ``{Wiretap Channels with Random
  States Non-Causally Available at the Encoder},'' \emph{IEEE Trans. Inf.
  Theory}, vol.~66, no.~3, pp. 1497--1519, Mar 2020. [Online]. Available:
  \url{https://ieeexplore.ieee.org/document/8894385/}
\BIBentrySTDinterwordspacing

\bibitem{Tahmasbi2020}
\BIBentryALTinterwordspacing
M.~Tahmasbi, M.~R. Bloch, and A.~Yener, ``{Learning an Adversary's Actions for
  Secret Communication},'' \emph{IEEE Trans. Inf. Theory}, vol.~66, no.~3, pp.
  1607--1624, Mar 2020. [Online]. Available:
  \url{https://ieeexplore.ieee.org/document/8836089/}
\BIBentrySTDinterwordspacing

\bibitem{Liang2009}
\BIBentryALTinterwordspacing
Y.~Liang, G.~Kramer, H.~V. Poor, and S.~S. (Shitz), ``{Compound Wiretap
  Channels},'' \emph{EURASIP J. Wirel. Commun. Netw.}, vol. 2009, no.~1, p.
  142374, Dec 2009. [Online]. Available:
  \url{http://jwcn.eurasipjournals.com/content/2009/1/142374
  https://jwcn-eurasipjournals.springeropen.com/articles/10.1155/2009/142374}
\BIBentrySTDinterwordspacing

\bibitem{He2010}
X.~He and A.~Yener, ``{Providing Secrecy when the Eavesdropping Channel is
  Arbitrarily Varying: a Case for Multiple Antennas},'' in \emph{Forty-Eighth
  Annu. Allert. Conf. Allert. House, UIUC, Illinois, USA}, 2010.

\bibitem{Chou2013}
\BIBentryALTinterwordspacing
R.~A. Chou and M.~R. Bloch, ``{Secret-Key Generation with Arbitrarily Varying
  Eavesdropper's Channel},'' in \emph{2013 IEEE Glob. Conf. Signal Inf.
  Process.}\hskip 1em plus 0.5em minus 0.4em\relax IEEE, Dec 2013, pp.
  277--280. [Online]. Available:
  \url{http://ieeexplore.ieee.org/document/6736869/}
\BIBentrySTDinterwordspacing

\bibitem{Janda2014}
\BIBentryALTinterwordspacing
C.~R. Janda, C.~Scheunert, and E.~A. Jorswieck, ``{Wiretap-Channels with
  Constrained Active Attacks},'' in \emph{2014 48th Asilomar Conf. Signals,
  Syst. Comput.}\hskip 1em plus 0.5em minus 0.4em\relax IEEE, Nov 2014, pp.
  1984--1988. [Online]. Available:
  \url{http://ieeexplore.ieee.org/document/7094818/}
\BIBentrySTDinterwordspacing

\bibitem{Janda2015}
\BIBentryALTinterwordspacing
C.~R. Janda, M.~Wiese, J.~Nötzel, H.~Boche, and E.~A. Jorswieck,
  ``{Wiretap-Channels under Constrained Active and Passive Attacks},'' in
  \emph{2015 IEEE Conf. Commun. Netw. Secur.}\hskip 1em plus 0.5em minus
  0.4em\relax IEEE, Sep 2015, pp. 16--21. [Online]. Available:
  \url{https://ieeexplore.ieee.org/document/7346805}
\BIBentrySTDinterwordspacing

\bibitem{Wang2016}
\BIBentryALTinterwordspacing
C.~Wang, ``{On the Capacity of the Binary Adversarial Wiretap Channel},'' in
  \emph{2016 54th Annu. Allert. Conf. Commun. Control. Comput.}\hskip 1em plus
  0.5em minus 0.4em\relax IEEE, Sep 2016, pp. 363--369. [Online]. Available:
  \url{http://ieeexplore.ieee.org/document/7852254/}
\BIBentrySTDinterwordspacing

\bibitem{Goldfeld2016}
\BIBentryALTinterwordspacing
Z.~Goldfeld, P.~Cuff, and H.~H. Permuter, ``{Arbitrarily Varying Wiretap
  Channels with Type Constrained States},'' \emph{IEEE Trans. Inf. Theory},
  vol.~62, no.~12, pp. 7216--7244, Dec 2016. [Online]. Available:
  \url{http://ieeexplore.ieee.org/document/7604072/}
\BIBentrySTDinterwordspacing

\bibitem{Chen2021}
\BIBentryALTinterwordspacing
Y.~Chen, D.~He, C.~Ying, and Y.~Luo, ``{Strong Secrecy of Arbitrarily Varying
  Wiretap Channels with Constraints by Stochastic Code},'' \emph{2021 IEEE Int.
  Symp. Inf. Theory}, pp. 843--848, Jul 2021. [Online]. Available:
  \url{https://ieeexplore.ieee.org/document/9517973/}
\BIBentrySTDinterwordspacing

\bibitem{Csiszar2011}
\BIBentryALTinterwordspacing
I.~Csiszar and J.~Korner, \emph{{Information Theory: Coding Theorems for
  Discrete Memoryless Systems}}, 2nd~ed.\hskip 1em plus 0.5em minus 0.4em\relax
  Cambridge: Cambridge University Press, 2011. [Online]. Available:
  \url{http://ebooks.cambridge.org/ref/id/CBO9780511921889}
\BIBentrySTDinterwordspacing

\bibitem{Bloch2011a}
\BIBentryALTinterwordspacing
M.~R. Bloch and J.~Barros, \emph{{Physical-Layer Security: From Information
  Theory to Security Engineering}}.\hskip 1em plus 0.5em minus 0.4em\relax
  Cambridge: Cambridge University Press, 2011. [Online]. Available:
  \url{http://books.google.de/books?id=ov5jYjrrNCIC}
\BIBentrySTDinterwordspacing

\bibitem{Ahlswede2014}
\BIBentryALTinterwordspacing
R.~Ahlswede, \emph{{Storing and Transmitting Data}}, ser. Foundations in Signal
  Processing, Communications and Networking, A.~Ahlswede, I.~Alth{\"{o}}fer,
  C.~Deppe, and U.~Tamm, Eds.\hskip 1em plus 0.5em minus 0.4em\relax Cham:
  Springer International Publishing, 2014, vol.~10. [Online]. Available:
  \url{http://link.springer.com/10.1007/978-3-319-05479-7}
\BIBentrySTDinterwordspacing

\bibitem{Cai2013}
N.~Cai, ``{Localized Error Correction in Projective Space},'' \emph{IEEE Trans.
  Inf. Theory}, vol.~59, no.~6, pp. 3282--3294, Feb 2013.

\bibitem{Dubhashi2009}
\BIBentryALTinterwordspacing
D.~P. Dubhashi and A.~Panconesi, \emph{{Concentration of Measure for the
  Analysis of Randomized Algorithms}}.\hskip 1em plus 0.5em minus 0.4em\relax
  Cambridge: Cambridge University Press, Oct 2009. [Online]. Available:
  \url{http://ebooks.cambridge.org/ref/id/CBO9780511581274}
\BIBentrySTDinterwordspacing

\bibitem{Ahlswede2002}
\BIBentryALTinterwordspacing
R.~Ahlswede and A.~Winter, ``{Strong Converse for Identification via Quantum
  Channels},'' \emph{IEEE Trans. Inf. Theory}, vol.~48, no.~3, pp. 569--579,
  Mar 2002. [Online]. Available:
  \url{http://ieeexplore.ieee.org/document/985947/}
\BIBentrySTDinterwordspacing

\bibitem{Shields1996}
\BIBentryALTinterwordspacing
P.~Shields, \emph{{The Ergodic Theory of Discrete Sample Paths}}, ser. Graduate
  Studies in Mathematics.\hskip 1em plus 0.5em minus 0.4em\relax Providence,
  Rhode Island: American Mathematical Society, Jul 1996, vol.~13. [Online].
  Available: \url{http://www.ams.org/gsm/013}
\BIBentrySTDinterwordspacing

\bibitem{Wyrembelski2010a}
R.~F. Wyrembelski, I.~Bjelakovi{\'{c}}, T.~J. Oechtering, and H.~Boche,
  ``{Optimal coding strategies for bidirectional broadcast channels under
  channel uncertainty},'' \emph{IEEE Trans. Commun.}, vol.~58, no.~10, pp.
  2984--2994, Sep 2010.

\bibitem{Bertsekas2009}
\BIBentryALTinterwordspacing
D.~P. Bertsekas, \emph{{Convex Optimization Theory}}, ser. Athena Scientific
  optimization and computation series.\hskip 1em plus 0.5em minus 0.4em\relax
  Athena Scientific, 2009. [Online]. Available:
  \url{http://www.athenasc.com/convexduality.html}
\BIBentrySTDinterwordspacing

\end{thebibliography}
\bibliographystyle{IEEEtran}
\end{document}